\documentclass[11pt]{article}
\usepackage[a4paper,left=1.02in,right=1.02in,top=1.1in,bottom=1.15in]{geometry}
\usepackage{setspace}
\renewcommand{\arraystretch}{1.22}
\usepackage{caption}
\AtBeginDocument{\setlength{\abovedisplayskip}{11pt plus 3pt minus 3pt}\setlength{\belowdisplayskip}{11pt plus 3pt minus 3pt}\setlength{\abovedisplayshortskip}{7pt plus 2pt}\setlength{\belowdisplayshortskip}{8pt plus 2pt}}
\usepackage{amsmath,amssymb,amsthm,mathtools}
\usepackage{booktabs,graphicx,xcolor,enumitem}
\usepackage{tikz}
\usetikzlibrary{arrows.meta,positioning,calc,patterns,decorations.pathreplacing}
\usepackage{placeins}
\usepackage{pgfplots}\pgfplotsset{compat=1.18}\usepgfplotslibrary{groupplots,fillbetween}
\pgfplotsset{every axis/.append style={tick label style={font=\footnotesize},label style={font=\small},legend cell align=left},every axis legend/.append style={font=\footnotesize,draw=none,fill=none,at={(0.5,-0.26)},anchor=north,legend columns=2,/tikz/every even column/.append style={column sep=0.6em}}}
\usepackage[numbers,sort&compress]{natbib}
\usepackage[expansion=false]{microtype}
\usepackage[colorlinks,linkcolor=blue!50!black,citecolor=blue!50!black,urlcolor=blue!50!black]{hyperref}
\setlist{itemsep=0.35em,topsep=0.4em}

\theoremstyle{plain}\makeatletter\def\thm@space@setup{\thm@preskip=1.1em plus 0.3em\thm@postskip=1.1em plus 0.3em}\makeatother
\newtheorem{theorem}{Theorem}[section]
\newtheorem{proposition}[theorem]{Proposition}
\newtheorem{corollary}[theorem]{Corollary}

\theoremstyle{definition}
\newcounter{experiment}\renewcommand{\theexperiment}{E\arabic{experiment}}
\newtheorem{definition}[theorem]{Definition}

\newtheorem{example}[theorem]{Example}

\makeatletter\renewenvironment{proof}[1][\proofname]{\par\medskip\noindent\textbf{#1.}\ \ignorespaces}{\hfill$\square$\par\medskip}\makeatother

\newcommand{\calA}{\mathcal{A}}\newcommand{\calH}{\mathcal{H}}\newcommand{\calW}{\mathcal{W}}
\newcommand{\calD}{\mathcal{D}}\newcommand{\calM}{\mathcal{M}}\newcommand{\calE}{\mathcal{E}}
\newcommand{\calF}{\mathcal{F}}\newcommand{\calG}{\mathcal{G}}\newcommand{\calI}{\mathcal{I}}
\newcommand{\calC}{\mathcal{C}}\newcommand{\calL}{\mathcal{L}}
\newcommand{\R}{\mathbb{R}}\newcommand{\E}{\mathbb{E}}\newcommand{\Prob}{\mathbb{P}}
\newcommand{\Cov}{\operatorname{Cov}}\newcommand{\Var}{\operatorname{Var}}\newcommand{\argmin}{\operatorname*{arg\,min}}
\newcommand{\Min}{\operatorname{Min}}\newcommand{\Hs}{H^{\star}}\newcommand{\Hst}{H^{\star}_{t}}
\newcommand{\Ast}{\calA^{\star}_{t}}\newcommand{\eqv}{\sim_{\calD}}
\newcommand{\lex}{<_{\mathrm{lex}}}\newcommand{\Tr}{\top}\newcommand{\pto}{\xrightarrow{\;p\;}}
\newcommand{\ind}{\mathbf 1}

\title{{\LARGE\bfseries Admissible Portfolio Optimization}\\[6pt]\parbox{\textwidth}{\centering\large\bfseries\setstretch{1.0}Information Constraints, Conditional Efficient Frontiers,\\ and the Price of Causal Identification}}
\author{Alejandro Rodr\'iguez Dom\'inguez\thanks{Quantitative Analysis and Artificial Intelligence Department, Miralta Finance Bank S.A., Madrid, Spain; Department of Computer Science, University of Reading, Reading, United Kingdom; and Lecturer in Data and Artificial Intelligence, Albert School, Madrid, Spain. Email: \texttt{arodriguez@miraltabank.com}. The views expressed are those of the author and do not necessarily reflect those of the institutions with which he is affiliated.}}
\date{September 2026}

\begin{document}
\maketitle

\begin{abstract}
\noindent
Mean--variance portfolio choice takes the conditioning information as given, so that two errors about it, using variables not available at the decision time and treating common variation as idiosyncratic, pass into the portfolio unseen. This paper turns what that practice assumes into a condition declared before estimation, certified from data and priced when it fails. The conditioning information becomes a decision variable subject to admissibility constraints on availability, arbitrage, statistical separation and, for interventional claims, invariance across regimes; an order in which no risk--return quantity enters selects an information class, and the classical problem is solved inside it. We prove that solutions exist, that the value of information is monotone only under an embedding that portfolio constraints can break, so that selection by decision loss admits look-ahead information whenever offered, that the two-stage solution and the joint envelope are minimal elements of two orders on one set, that the diversifiable part of risk is a property of the admissible class, and that causal identification has an explicit price. Allocation performance is downstream of the framework, since the information decision enters the decision regret at second order and the reported risk at first order, while the weights themselves move at first order. On market data the condition holds on small sets of individual securities and fails on pre-diversified portfolios and on books of realistic size, where a factor model with diagonal residual understates its portfolio's variance by a factor that the consequence statistic announces in advance; selecting the same information by decision loss yields portfolios that promise a fraction of the risk they realize.
\end{abstract}

\noindent\textbf{Keywords:} portfolio choice; information selection; conditional efficient frontier; causal identification; decision-focused learning; oracle inequalities.\\
\textbf{JEL:} C44, C58, G11. \quad \textbf{MSC 2020:} 91G10, 62C05, 90C29.

\section{Introduction}\label{sec:intro}

A mean--variance investor chooses portfolio weights once expected returns and covariances have been specified relative to some information \cite{Markowitz1952,Merton1972}. The information is an input. The theory of conditioning information \cite{HansenRichard1987,FersonSiegel2001} shows how much it matters for the frontier the investor perceives, but in practice it is a factor model, a set of signals or the output of a variable-selection procedure, adopted before the optimizer is called and never itself optimized. Two errors about this input, using variables that were not legitimately available when the decision was taken and treating common variation across assets as if it were idiosyncratic, cause a backtest to overstate what a portfolio can deliver \cite{Michaud1989,ChopraZiemba1993,HarveyLiu2015}. Neither is visible once the moments have been estimated, and neither is corrected by covariance shrinkage \cite{LedoitWolf2004} or by training the forecasting model on the decision loss \cite{ElmachtoubGrigas2022}, since both take the information as given.

\sloppy The paper makes the conditioning information a decision variable, subject to hard constraints,\footnote{Admissibility is used throughout in the sense of eligibility for a declared claim, not in Wald's decision-theoretic sense of non-dominance; an admissible representation here may be dominated in decision value by an inadmissible one.} and the primitive of portfolio construction becomes an admissible information--portfolio pair, chosen in two stages that do not interact. A decision contract, written before any return is examined, declares the claim the portfolio is meant to support and the conditions a candidate information set must satisfy: availability at the decision time, an enlargement of the decision filtration that preserves the absence of arbitrage, removal of the common dependence between assets that mean--variance choice treats as absent, and, when the claim concerns interventions, invariance of the conditional return law across the declared regimes. Candidates that violate a required condition are infeasible. Among the admissible ones a lexicographic admissibility order, in which no risk--return quantity enters, selects an optimal information class; only then, inside that class, is the classical problem solved and its frontier, tangency portfolio and capital allocation line constructed. The order is not a utility, because chronology, publicity, separation and identification are propositions with truth values rather than attributes with rates of substitution.

Markowitz introduced diversification through the portfolio problem; the paper introduces admissibility, and causality as its interventional case, through the same problem, and the two notions are connected. Diversification concerns dependence between returns, dependence is, by the Commonality Principle \cite{RD2023,RD2025causal}, attributable to the common causal drivers of the constituents, and the admissible information class determines which drivers may explain it. Conditional on admissible information, covariation splits into a response to common drivers and an idiosyncratic residual, only the residual is diversifiable, and the split is a property of the information class rather than of the asset universe.

The paper's claim should be stated precisely. What mean--variance practice assumes, that its conditioning information removes the common variation its diagonal residual treats as absent, becomes a condition declared before estimation, certified from data and priced when it fails; and because an error in the conditional moments moves the weights at first order but their realized loss only at second order, while it moves the reported risk at first order, the consequences of the information decision are in the risk a model can be trusted to report, in how often the allocation moves and in what alternative selection rules deliver, which is what the experiments measure.

Population dominance of the admissible problem over any fixed representation is a population result, and Section~\ref{sec:finite} gives its finite-sample counterpart; causal restriction can lower static decision value because it solves a stronger problem, and a lower conditional loss never licenses a causal interpretation; and of the admissibility conditions only the mandate scope is a declaration, the others being certified from provenance, from the enlarged filtration or from returns.

\subsection{Contributions, practical use and limits}\label{sec:contrib}

The contributions are the following. The two-stage problem has solutions and the selected class is invariant to recodings of the same information. The value of information is monotone only under an embedding hypothesis that portfolio constraints can break, with the sharp consequence that any rule selecting information by decision loss, loss-based representation selection included, admits look-ahead information whenever the candidate set contains it. The two-stage solution and the joint envelope over all admissible representations are the minimal elements of two orders on one set of admissible pairs, differing only in whether risk may compare portfolios across representations. Under exact separation the diversifiable part of risk is determined by the admissible class and is invariant to interventions only under interventional admissibility, the portfolio form of the Commonality Principle. Requiring causal identification carries an oracle price bounded by an explicit distance in conditional moments and traded against representation-search complexity in an oracle inequality. An exact decomposition states what the diagonal-residual restriction costs a given portfolio, the residual share of its variance times its breadth times the average residual correlation, and supplies a consequence statistic that is certified at the weights the mandate holds and remains informative on universes where the screening defect no longer discriminates. The two-stage estimator is consistent, the regret of the estimated portfolio is second order in estimation error, and information selection separates across decision nodes. Controlled experiments verify the identities and rates. Field experiments on a book of individual equities of a size a manager would hold, on a cross-asset universe of equities and bonds, on sorted-portfolio universes and on a small equity panel, all against a declared universe of economic, financial, industry and style candidate drivers, find that the separator condition is attainable only on small sets of individual securities and never on pre-diversified portfolios or on a book of realistic size, where the maximum over asset pairs stops discriminating between representations and the portfolio-level consequence statistic becomes the operative diagnostic. On the large book the diagonal-residual restriction understates the variance of the model's own portfolio substantially when the candidates are macroeconomic series and modestly once industry and style candidates are admitted, and the model's bias statistic, the ratio of realized to predicted volatility, moves accordingly, so that the candidate universe is itself a design variable of the contract and the framework reports what a given choice of it costs. Demanding loadings that are stable across the regimes of the sample leaves very few admissible drivers and raises out-of-sample variance by a measurable amount, which is the price of causal identification observed rather than bounded. And selecting the same information by its decision loss, the loss-based alternative, produces portfolios that promise a small fraction of the risk they realize, the overfitting form of the look-ahead corollary.

What a practitioner can take from the paper is narrower than the framework and can be stated without its apparatus. First, a specification test for an existing factor risk model: the screening defect with a size-controlled tolerance, and the consequence statistic at the weights the mandate holds, say whether the diagonal specific-risk assumption is defensible for a given universe and factor set and how much variance it discards when it is not, a check that model-validation functions can run on the models they already have. Second, a map of where factor models with diagonal residual should not be expected to hold: pre-diversified assets, whether sector or style portfolios, funds of funds or strategy indices, retain a common residual component that no candidate set of the kind searched here removes. Third, a rule for when to change the information a book is conditioned on that uses certification statistics and not performance, so that the selection cannot learn from the profit and loss it produces. Fourth, a price for a claim that mandates often make implicitly, that the risk model should survive the next regime.

The limits are equally definite, and the most important of them concerns performance. The framework does not lower realized portfolio risk: against static and dynamic shrinkage estimators, with the same solver and paired tests, the covariance built on the selected class is at par and never significantly better, while it does dominate statistical factors of the same dimension. The reason is the setting, not the information. Inside the quadratic mean--variance problem the conditional information enters only through two conditional moments, and an error in them moves the weights at first order but their realized loss only at second order, so that two covariances agreeing on the systematic part and differing in the residual deliver nearly the same realized variance however differently they describe the risk; \cite{RD2025necessary} reached the same conclusion from a computational angle, that in this setting a causally better representation cannot be expected to pay in allocation. Converting the information the framework certifies into realized performance therefore requires leaving the quadratic setting for one in which the drivers enter through the trajectory of the constituents' sensitivities to them, the sensitivity-based formulation of \cite{RD2023} and the causal PDE--control model of \cite{RD2025cpcm}, in which portfolio choice takes place on the manifold of the common causal drivers and the sensitivities carry the trajectory; that step is downstream of certification and is not taken here. Two further limits follow from the evidence rather than the setting: the representations the framework selects are contemporaneous decompositions of risk and not signals, since nothing survives a one-session lag; and the interventional branch prices a stability restriction over observational environments, which coincides with the price of causal identification only under the identifying assumptions stated, and has no application until a mandate makes a claim about interventions that a conditional-law model cannot support.

The rest of the paper is organized as follows. Section~\ref{sec:lit} places the framework in the literatures on conditioning information, information orderings, decision-focused learning and causal structure in portfolio construction. Section~\ref{sec:apo} develops the theory: the decision contract with its claim levels and admissibility conditions, the admissible class and the order that ranks it, the two solution sets that order produces, and the consequences that carry the paper, the value of admissible information with its look-ahead corollary, the diversification result with the decomposition of what the diagonal-residual restriction costs a portfolio, and the price of causal identification. Section~\ref{sec:finite} prices the selection from data and confronts it with evidence, controlled and on market panels: whether the separator condition is attainable, what its failure costs, what the selected class delivers under the standard allocation rule, and how the three ways of choosing a portfolio's information compare on the same data. Section~\ref{sec:discussion} concludes.

\section{Literature review}\label{sec:lit}

Mean--variance theory \cite{Markowitz1952,Merton1972} chooses weights given a mean vector and a covariance matrix, and its conditional extensions compute those moments relative to information available at the decision time. That the conditioning information reshapes the frontier is the classical result of Hansen and Richard \cite{HansenRichard1987}; Ferson and Siegel \cite{FersonSiegel2001,FersonSiegel2009} characterized the portfolios that use a given conditioning set efficiently; Brandt and Santa-Clara \cite{BrandtSantaClara2006} incorporated it by augmenting the asset space with managed portfolios; Pe\~naranda and Sentana \cite{PenarandaSentana2016} showed that conditioning information and active strategies alter the duality between portfolio and discount-factor frontiers, and Pe\~naranda \cite{Penaranda2016} distinguished the notions of efficiency that predictability creates; Kandel and Stambaugh \cite{KandelStambaugh1996} and Abhyankar, Basu and Stremme \cite{Abhyankar2007,Abhyankar2012} measured the value of the information from the allocation side, and Fletcher \cite{Fletcher2011} found that, net of the turnover it requires, that value is largely dissipated in costs. In all of this the conditioning set is supplied by the model. The frontier these papers characterize is, for a fixed admissible class, the second-stage object of this paper, whose question is the prior one of which conditioning set may be used at all; Section~\ref{sec:order} places the classical objects next to their counterparts here, and Fletcher's finding reappears as a consequence of the perturbation bound of Section~\ref{sec:guarantees}, under which the information decision enters the decision regret at second order.

The estimation literature documents that optimized portfolios are fragile to errors in means \cite{Michaud1989,BestGrauer1991,ChopraZiemba1993} and that shrinkage and constraints stabilize them \cite{JagannathanMa2003,LedoitWolf2004,DeMiguelNogales2009,FanZhangYu2012}, with nonlinear and dynamic shrinkage \cite{LedoitWolf2017,LedoitWolf2020,LedoitWolf2022,EngleLedoitWolf2019} as the current state of the art for large universes. Approximate factor structures \cite{Ross1976,ChamberlainRothschild1983,BaiNg2002} justify treating residual risk as diversifiable, the fundamental factor models that practice runs descend from \cite{Rosenberg1974}, and the estimators of Fan and coauthors \cite{FanFanLv2008,FanLiaoMincheva2013} and of De Nard, Ledoit and Wolf \cite{DeNardLedoitWolf2021} allow a sparse or dynamic residual dependence to survive the factors. Every one of these estimators presupposes a split between systematic and diversifiable risk and estimates its two parts; the present paper asks on which conditioning information the split holds, certifies the answer from data, and shows that when it fails the resulting error in reported risk has a sign and a size, which the validation statistics of risk-model practice \cite{MencheroOrrWang2011,Patton2011} then measure. The decomposition behind that statement is developed in \cite{RD2026screening}, which works at the level of a fixed representation: it splits the conditional covariance of an arbitrary candidate into represented response, omitted response and residual dependence and proves that the two error terms have different sign structure. The question asked there is what a given representation leaves unexplained; the question asked here is which representation may be used at all, and the decomposition enters only as a tool once that choice has been made.

The validity of the information itself is addressed by two literatures that do not meet. Data snooping and multiple testing in the evaluation of strategies \cite{SullivanTimmermannWhite1999,White2000,HarveyLiu2015,HarveyLiuZhu2016} show how selection on past performance inflates what a backtest promises; the chronological condition of this paper excludes the mechanism rather than correcting for it, and the look-ahead corollary of Section~\ref{sec:voi} states why no correction based on decision loss can. The theory of enlargement of filtrations \cite{Jacod1985,PikovskyKaratzas1996,AmendingerImkellerSchweizer1998,Acciaio2016} characterizes when adding information to the decision filtration destroys the semimartingale property or creates arbitrage, which is the mechanism by which an insider's expected utility becomes unbounded; the paper takes the sufficient condition of Jacod as the content of its publicity condition and the no-arbitrage framework of Delbaen and Schachermayer \cite{DelbaenSchachermayer1994} as the sense in which the enlarged filtration must remain admissible. Decision-focused learning and integrated estimation--optimization \cite{ElmachtoubGrigas2022,Elmachtoub2023,Mandi2024,Mandi2022,ButlerKwon2023,CostaIyengar2023,Anis2025,Lee2025,Yang2025} address the mismatch between the loss used to estimate a model and the loss of the decision it feeds; these methods decide how the primitives of a given representation are estimated, they do not decide whether the representation is eligible for the claim attached to it, and the paper proves that they cannot, since selection by decision loss admits look-ahead information whenever the candidate set contains it. The selection of characteristics, factors and features for portfolios \cite{Bjerring2017,KellyPruittSu2019,KozakNagelSantosh2020,GuKellyXiu2020} selects by predictive fit on the same grounds. From model selection \cite{Vapnik1998,Massart2007,Bartlett2008,BLM2013} the paper takes the approximation--estimation logic of oracle inequalities and applies it to a nesting fixed by the claim rather than tuned to the data.

Two orderings of information enter the analysis. Blackwell's comparison of experiments \cite{Blackwell1951,Blackwell1953,LeCam1964,Torgersen1991} says that finer information weakly improves every decision problem because coarser policies remain available; the Commonality Principle, discussed below, removes the dependence between a portfolio's constituents by conditioning on their common causal drivers. The paper uses the second as the separator condition, with residual dependence after conditioning as the statistical defect of a candidate, and the first as the benchmark against which admissibility truncates the value of information: Blackwell's theorem presupposes that finer information is eligible and coarser policies implementable, and both fail here, the first for look-ahead and non-public information and the second when constraints depend on the representation. The word admissible is used in the sense of eligibility for a declared claim, not in Wald's decision-theoretic sense of non-dominance \cite{Wald1950}, and the two notions can disagree; the lexicographic order that ranks eligible representations admits no continuous utility representation \cite{Fishburn1974}, which is the formal reason admissibility cannot be exchanged against performance. The uncertainty set that the certification produces, a bound on every residual correlation, connects the framework to robust portfolio choice under parameter and model uncertainty \cite{GilboaSchmeidler1989,GoldfarbIyengar2003,GarlappiUppalWang2007}; the difference is that the set is not declared by the modeller but delivered by the certification, with the declared tolerance as its radius.

Causal structure entered portfolio construction through the Commonality Principle \cite{RD2023}: the drivers relevant to a portfolio's diversification are the common causes of its constituents, so that conditioning on them removes the dependence among the constituents and leaves the residual idiosyncratic. Its antecedent is Reichenbach's common-cause principle \cite{Reichenbach1956}; its structural proof, through the causal Markov condition, is in \cite{RD2025causal}; and its consequences were developed in the conditional risk-neutral dynamics of a portfolio given its common drivers \cite{RD2024pde}, in the static theory induced by separation as a single structural condition \cite{RD2026screening}, in the causal PDE--control model, in which allocation is a control problem on the common-driver manifold under partial information \cite{RD2025cpcm}, in the dynamic problem in which the conditioning representation is the state that rotates and jumps \cite{RD2026dynamic}, and in the attribution of a portfolio's price of risk to causal drivers with its obstruction to aggregation \cite{RD2026obstruction}. Causal-network representations have since been applied to factor design and timing \cite{Howard2025}. Pearl's hierarchy \cite{Pearl2009,Bareinboim2022} grades claims, structural causal models are formalized in \cite{PetersJanzing2017}, and invariant causal prediction \cite{Peters2016,HeinzeDeml2018,Buhlmann2020} identifies causal predictors by the stability of the conditional law across environments. Whether causal identification is a prerequisite for efficient portfolios is contested: \cite{LopezdePrado2024} argue that it is, and \cite{RD2025necessary} argues from a computational perspective that it is not, since predictive validity and model misspecification are separate questions. The present paper gives that disagreement a precise form. Causal identification is a level of the claim, with an interventional admissibility condition whose satisfaction can be certified by invariance across declared regimes, and it carries a price in decision value that the paper bounds and measures; a conditional-law contract is not improved by it and an interventional contract cannot do without it. The paper also adds the portfolio consequence, that the diversifiable component of risk is determined by the admissible class and is interventionally stable only under interventional admissibility. Coherent and spectral risk measures \cite{Artzner1999,Acerbi2002}, the expected-shortfall program of \cite{RockafellarUryasev2000} and mean--risk theory \cite{Miller2008,Rodrigues2019,Bermin2024,Lim2011} supply alternatives to variance for the second stage, orthogonal to the information contract of the first: the risk measure changes the portfolio chosen inside the class and leaves the selection of the class untouched.

\section{Admissible information and admissible portfolio choice}\label{sec:apo}
\label{sec:prelim}

This section builds one object, the admissible information--portfolio pair, and proves three things about it: that information value is monotone only under an embedding that portfolio constraints can break, that the diversifiable part of risk is fixed by the admissible class, and that requiring causal identification has a bounded price. Section~\ref{sec:setting} fixes notation; Section~\ref{sec:setup} defines the contract and its conditions; Section~\ref{sec:order} defines the order and the two solutions; Section~\ref{sec:results} states the three results. Everything of a technical nature beyond these is in the appendices.

\subsection{Notation}\label{sec:setting}

Fix a finite asset family $A$ with $|A|=n$, a decision node $t$, a horizon $h>0$, a probability space $(\Omega,\calF,\Prob)$ and baseline operative information $\calG^{0}_{t}\subseteq\calF$. Let $R_{t,h}\in L^{2}(\Prob;\R^{n})$ be the vector of horizon returns, written $R$ when the node and horizon are fixed. All equalities between $\sigma$-algebras, conditional laws and conditional moments are modulo $\Prob$-null sets; $\R^{n}$ carries the Euclidean topology, $\succ0$ and $\succeq0$ denote positive definiteness and semidefiniteness, and $\ind$ is the vector of ones. Upper hemicontinuity of a correspondence is in the sense of \cite[Def.~17.2]{AliprantisBorder2006}, and lower semicontinuity of an extended-real function means $\liminf_{x'\to x}f(x')\ge f(x)$. Calligraphic letters $\calA,\calH,\calW,\calD$ denote classes of representations, candidate sets, feasible portfolio sets and decision contracts; $H,\tilde H,K$ denote representations, $w$ a portfolio and $m$ a target return. For a representation $H$, $\mu_H:=\E[R\mid H]$ and $\Sigma_H:=\Cov(R\mid H)$ are its conditional mean and covariance, well defined because $R$ is square-integrable. A star marks optimal objects: $\Ast$ is the Stage-I solution set, $\Hst$ a representative of it, $\rho^{\star}_H(m)$ the minimum conditional risk at target $m$ under $H$, and $\rho^{\star}_t(m):=\rho^{\star}_{\Hst}(m)$ when the representative is fixed. Hats denote estimators. The classical results the construction rests on are stated with their hypotheses at the point where each is first used, and proved in Appendix~\ref{app:used}.

\subsection{Decision contracts, claim levels and admissibility conditions}\label{sec:setup}

This subsection defines what an information set must satisfy to be usable for a declared claim. The definitions are the framework's primitives; the reader who wants only the argument can take from them that a contract fixes the claim and the conditions in advance, that the conditions are constraints rather than penalties, and that all but one of them are certified from data.

\begin{definition}[Decision contract]\label{def:contract}
A decision contract $\calD$ specifies, before any data are examined, the asset family $A$, the decision node $t$ and horizon $h$, the baseline information $\calG^{0}_{t}$, a declared class $\calH$ of candidate information representations, the feasible portfolio set $\calW(H)\subseteq\R^{n}$ attached to each candidate $H$, the claim level of Definition~\ref{def:claim} with its identification requirements, and the decision loss $\ell$, with $J(w;H)=\E[\ell(w,R_{t,h};H)]$. A representation has a decision-time component $H^{-}$, with respect to which portfolios are measurable, and may have a horizon-closed component $H^{+}$ used only for certification.
\end{definition}

\begin{definition}[Claim levels]\label{def:claim}
A decision contract declares the level of the causal hierarchy at which its conclusions are to be read \cite{Pearl2009,Bareinboim2022}. A \emph{conditional-law claim} ($\calC=\mathrm{c}$) asserts properties of the conditional law $\calL(R_{t,h}\mid H)$ under the observational regime only; it requires the financial, chronological and statistical admissibility conditions and no causal semantics. An \emph{interventional claim} ($\calC=\mathrm{i}$) asserts in addition that the conditional primitives remain valid under a declared set $\mathcal E$ of interventions, transports or counterfactual regimes, which must contain the regimes about which the claim is made. Identification at this level is defined by invariance: writing $\calL^{e}$ for the law under regime $e\in\mathcal E$ and $\calL^{\mathrm{o}}$ for the observational law, the causal admissibility condition is satisfied if and only if $\calL^{e}(R_{t,h}\mid H)=\calL^{\mathrm{o}}(R_{t,h}\mid H)$ for every $e\in\mathcal E$. Because the interventional level adds a requirement to the conditional-law level, $\calA_t(\mathrm{i})\subseteq\calA_t(\mathrm{c})$ for a fixed candidate set.
\end{definition}

\label{rem:claim}The grading is the portfolio analogue of the association--intervention--counterfactual ladder \cite{Pearl2009}, and it is part of the contract because it determines which conditions a candidate must satisfy. Identification is the invariance criterion of \cite{Peters2016}, testable from data collected under the regimes in $\mathcal E$ and coinciding with graphical identification when $\mathcal E$ consists of interventions on non-descendants of $H$ in a structural model \cite[Ch.~3]{Pearl2009}. If the claim concerns interventions on the drivers themselves, $\mathcal E$ must contain them, and invariance under them is a statement about the absence of confounding between $H$ and $R$. Invariance does not select a minimal set, since every superset of an invariant conditioning set is invariant; minimality enters through the dimension coordinate of the admissibility order. The requirement that $\mathcal E$ contain the regimes of interest keeps the condition from being vacuous.

The statistical condition rests on one object. Let $Z$ be a vector of drivers with path $Z_{[t,t+h]}$ over the horizon, $H^{+}:=\sigma(\calG^{0}_t,Z_{[t,t+h]})$ the horizon-closed information it generates and $H^{-}:=\sigma(\calG^{0}_t,Z_t)$ its decision-time component; the drivers are a \emph{separator} for $A$ if the components of $R$ are mutually independent conditional on $H^{+}$, and a separator is \emph{least} if no separator generates a strictly smaller $\sigma$-algebra (Definition~\ref{def:separator}, Appendix~\ref{app:used}). By the law of total covariance, which needs only $R\in L^{2}$, $\Cov(R\mid H^{-})=\E[\Cov(R\mid H^{+})\mid H^{-}]+\Cov(\E[R\mid H^{+}]-\E[R\mid H^{-}]\mid H^{-})$; under separation the first term, $\Omega_H$, is diagonal, and when $\E[R\mid H^{+}]$ is affine in a statistic of the driver path the second is $BVB^{\Tr}$ (Proposition~\ref{prop:tower}, Appendix~\ref{app:used}). The \emph{screening defect} of a candidate $H$ is the maximal absolute conditional correlation between components of $R$ given $H$, zero if and only if the drivers are a separator in the second-moment sense, since a conditional correlation is zero almost surely exactly when the corresponding conditional covariance is; it is the statistic the statistical condition is certified with. The market implementation computes it from least-squares residuals, which equals the conditional correlation when the conditional mean is affine in the drivers and is the linear-projection analogue otherwise, a distinction Section~\ref{sec:market} keeps explicit.

\begin{definition}[Admissibility conditions]\label{def:conditions}
A representation $H\in\calH$ passes the \emph{financial admissibility condition} if (i) its decision-time component is measurable with respect to information released no later than the decision node, a property of data provenance certified from release timestamps and point-in-time vintages; (ii) the enlargement of the baseline filtration by $H^{-}$ preserves the semimartingale property of discounted prices and the absence of free lunches with vanishing risk \cite{DelbaenSchachermayer1994}, for which Jacod's condition on the conditional laws of the added information is sufficient \cite{Jacod1985,AmendingerImkellerSchweizer1998}; and (iii) the enlargement lies within the mandate and legal scope declared in the contract. Condition (i) concerns timing and (ii) publicity: for information that is public before $t$ the enlargement stays inside the market filtration and (ii) holds automatically, whereas (ii) is binding when a variable passes the timestamp check but is available to the decision maker alone, as for an insider; the \emph{statistical admissibility condition} if it satisfies the declared separator or sufficiency condition; the \emph{causal admissibility condition} if, for an interventional claim, it carries identified intervention semantics in the sense of Definition~\ref{def:claim}; and the \emph{implementation admissibility condition} if the representation can be implemented under the declared trading technology. We write $G_F(H)$, $G_S(H)$, $G_C(H)$ and $G_I(H)$ for the corresponding indicators, equal to one when the condition is satisfied. The contract determines which admissibility conditions are required, and the \emph{admissible class} $\calA_t$ is the set of candidates in $\calH$ satisfying every required admissibility condition.
\end{definition}

\begin{definition}[Information order and policy embedding]\label{def:order}
For decision-time information objects write $H_1\preceq H_2$ when $\sigma(H_1)\subseteq\sigma(H_2)$ modulo $\Prob$-null sets; when the objects are statistical experiments rather than sub-$\sigma$-algebras, $\preceq$ denotes the declared Blackwell (garbling) order \cite{Blackwell1953}. For $H\in\calA_t$ let $\Pi(H)$ denote the set of feasible policies, i.e.\ $\sigma(H^{-})$-measurable random vectors $\pi$ with $\pi\in\calW(H)$ almost surely, and write $J(\pi;H)=\E[\ell(\pi,R_{t,h};H)]$. For $H_1\preceq H_2$, say that $H_2$ embeds the $H_1$ decision problem if there is a map $\iota:\Pi(H_1)\to\Pi(H_2)$ with $J(\iota\pi;H_2)\le J(\pi;H_1)$ for every $\pi\in\Pi(H_1)$; the inequality is required of expected losses, and an almost-sure inequality between the losses is a sufficient condition for it. When the loss does not depend on the representation the inequality is an equality for $\iota$ the identity; when it does, as for the conditional mean--variance loss, the identity map is still an embedding whenever $\sigma(H_1^{-})\subseteq\sigma(H_2^{-})$ and $\calW(H_1)\subseteq\calW(H_2)$, by the law of total variance (Appendix~\ref{app:proofs}).
\end{definition}

The order $\preceq$ compares information content and $\eqv$ identifies representations indistinguishable for the declared decision. Embedding is stronger than inclusion of weight vectors when constraints or costs depend on the representation, and its failure is one of the two ways in which the monotonicity of information value fails (Corollary~\ref{cor:not-maximal}).

\subsection{The admissibility order and the admissible portfolio problem}\label{sec:order}

The message of this subsection is that the problem has two solution sets and that they are computed in a fixed order: an admissibility optimum over information, in which no risk--return quantity enters, and a conditional efficient frontier inside it. The joint envelope that optimizes over pairs at once is a different contract, and the relation between the two is stated below.
Fix a decision node $t$, the asset family, horizon, baseline information, candidate set $\calH_t$, claim level, tolerances and portfolio constraints.

\subsubsection{The admissible class and its order}

\begin{definition}[Admissible class and admissibility order]\label{def:adm-order}
The admissible class $\calA_t$ is the set of candidates in $\calH_t$ that satisfy every admissibility condition the contract requires.\label{eq:hard-class} To each $H\in\calA_t$ are attached four nonnegative quantities: the separator or sufficiency defect $\varepsilon_t(H)$; the causal-response defect $\delta_t(H)$, identically zero under a conditional-law contract; the information complexity $\kappa_t(H)$; and a predeclared implementation burden $\iota_t(H)$. The first two are scientific quantities, the last two tie-breaks. The contract may declare, in addition, a \emph{consequence statistic} for the portfolio it intends to hold, the relative misallocation $w^{\Tr}\Omega^{\mathrm{off}}_Hw/w^{\Tr}\Sigma_Hw$ of Proposition~\ref{prop:diversification}(d) evaluated at the intended weights; the screening defect certifies whether the representation separates, the consequence statistic certifies how much its failure costs the portfolio, and Section~\ref{sec:market} shows that the two need not agree.
\begin{enumerate}[label=(\alph*)]
\item (\emph{Exact order.}) $H_1\prec_{\calA}H_2$ if the first of $\varepsilon_t,\delta_t,\kappa_t,\iota_t$ at which $H_1$ and $H_2$ differ is smaller for $H_1$.
\item (\emph{Order at tolerance $\bar\varepsilon$.}) Let $\calA_t(\bar\varepsilon)$ be the candidates in $\calA_t$ with $\varepsilon_t(H)\le\bar\varepsilon$. On $\calA_t(\bar\varepsilon)$, $H_1\prec_{\calA}^{\bar\varepsilon}H_2$ if the first of $\delta_t,\kappa_t,\iota_t,\varepsilon_t$ at which they differ is smaller for $H_1$: within tolerance the defect is a satisfied constraint, and it re-enters only as the last tie-break.
\end{enumerate}
The Stage-I solution is the set of minimal elements,
\begin{equation}\label{eq:stage1}
\Ast=\Min_{\prec_{\calA}}\calA_t,\qquad \Ast(\bar\varepsilon)=\Min_{\prec_{\calA}^{\bar\varepsilon}}\calA_t(\bar\varepsilon),
\end{equation}
The two orders coincide at $\bar\varepsilon=0$ when an exact separator exists, since then (a) and (b) both select the exact separators of least dimension.
\end{definition}

No return target, Sharpe ratio, realized utility or point of the Stage-II frontier enters any of the four quantities. Figure~\ref{fig:stage1} illustrates on the \emph{reference system}, used for every population figure in this section and for the controlled experiments of Section~\ref{sec:evidence}: eight assets driven by two common states $Z_1,Z_2$ with loadings drawn once, idiosyncratic noise, a proxy $P$ correlated with $Z_1$ but not a cause of the assets, and a future variable $X^{+}$ that violates the chronological condition; the candidate set is the $41$ non-empty subsets of $\{Z_1,Z_2,P,X^{+},Z_1^{2},Z_2^{2}\}$ and every quantity shown is a population value computed from the system's exact moments, with parameters listed in Appendix~\ref{app:config}.

\begin{figure}[!htb]
\centering
\begin{tikzpicture}
\begin{axis}[width=11cm,height=5.6cm,xlabel={information dimension $\kappa_t(H)$},ylabel={screening defect $\varepsilon_t(H)$},xtick={1,2,3},xmin=0.5,xmax=3.5,ymin=-0.04,ymax=0.62,legend style={font=\footnotesize,at={(0.5,-0.32)},anchor=north,legend columns=2,/tikz/every even column/.append style={column sep=0.6cm}},tick label style={font=\footnotesize},label style={font=\small}]
\addplot[only marks,mark=x,mark size=2.6pt,red!70!black,thick] coordinates {(0.97,0.472) (1.95,0.615) (2.10,0.566) (2.07,0.615) (1.91,0.435) (1.78,0.472) (3.17,1.000) (3.12,0.615) (3.09,0.615) (2.94,0.615) (3.19,1.000) (3.06,0.646) (2.91,0.566) (3.10,0.615) (2.88,0.615) (3.15,0.435)};
\addlegendentry{contains $X^{+}$: fails chronological admissibility condition}
\addplot[only marks,mark=*,mark size=1.8pt,black] coordinates {(0.82,0.463) (0.88,0.562) (1.13,0.463) (1.04,0.418) (0.82,0.472) (1.99,0.000) (1.85,0.463) (2.10,0.463) (1.83,0.463) (2.01,0.000) (1.97,0.299) (2.04,0.562) (2.20,0.463) (1.91,0.463) (2.09,0.418) (3.21,0.000) (2.91,0.000) (2.92,0.000) (3.04,0.463) (2.99,0.463) (2.79,0.463) (2.82,0.000) (3.07,0.000) (2.87,0.299) (3.11,0.463)};
\addlegendentry{satisfies the admissibility conditions}
\draw[blue!60!black,thick,rounded corners=3pt] (axis cs:1.7,-0.025) rectangle (axis cs:2.3,0.03);
\node[blue!60!black,font=\footnotesize,anchor=south] at (axis cs:2.0,0.035) {$\calA^{\star}_t=\{Z_1,Z_2\}\sim\{\phi(Z_1),Z_2\}$};
\node[font=\footnotesize,anchor=east] at (axis cs:1.85,0.30) {$\{P,Z_2\}$: $\varepsilon_t=0.299$};
\end{axis}
\end{tikzpicture}
\caption{Stage~I is an information problem, shown on the 41 candidates of the reference system (population values). Crosses are the 16 candidates containing the future variable $X^{+}$; they include the best predictors and are infeasible. Among the 25 admissible candidates the lexicographic order first minimizes the screening defect, which selects the exact separators, and then the dimension, which selects the class of $\{Z_1,Z_2\}$ and its recoding over the three-dimensional exact separators. No risk--return quantity enters.}
\label{fig:stage1}
\end{figure}
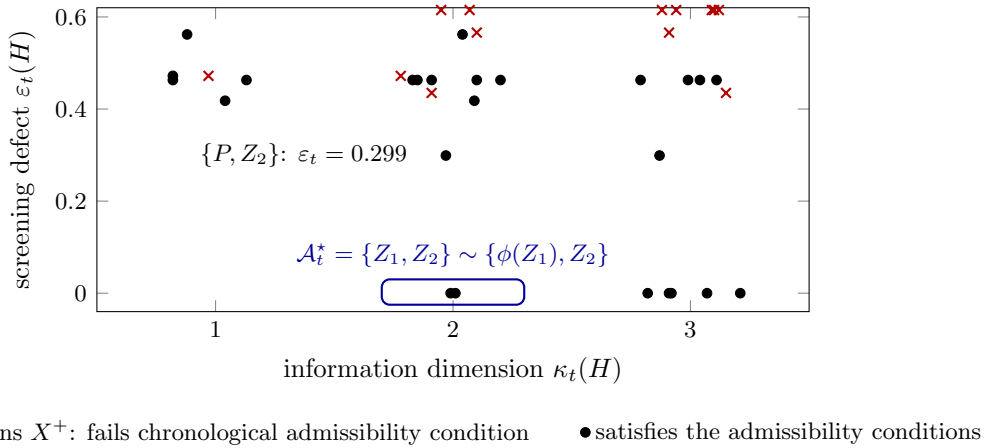

\label{rem:lex}The lexicographic order is chosen over a penalty because chronology, publicity, separation and identification are propositions with truth values, not attributes with rates of substitution: a penalty would let a large enough Sharpe ratio purchase a false proposition, and a lexicographic order admits no continuous utility representation \cite{Fishburn1974}. On a finite candidate set it has a real representation (Section~\ref{sec:finite}), and every result holds for each fixed pre-declared tie-break.

In the exact population problem an exact separator, when it exists, has $\varepsilon_t=0$ and is selected by the first coordinate of the order, the remaining coordinates deciding among exact separators; in finite samples the statistical condition is applied with a pre-declared tolerance calibrated to sampling error (Section~\ref{sec:finite}). Causal validity is never exchanged against portfolio performance.

Existence of the two solution sets, under compactness of the feasible set with continuous risk and mean, or on a compact quotient when the four quantities are lower semicontinuous, is Theorem~\ref{thm:two-sets-formal} in Appendix~\ref{app:formal}; the argument is given below.\label{thm:two-sets}

On a finite candidate set the admissibility order is a total preorder, and a finite set with a total preorder has minimal elements, so $\Ast$ is nonempty; the lexicographic structure gives them explicitly, by keeping the candidates of least defect, then among them those of least dimension, and so on through the four coordinates. On a compact quotient the same selection goes through because each coordinate is lower semicontinuous, so its set of minimizers over a compact set is nonempty and compact, and four nested minimizations preserve both properties. Inside a selected class the feasible set at a target $m$ is compact when the feasible set is compact and the conditional mean is continuous, and a continuous risk functional attains its minimum there by Weierstrass' theorem; the efficient set is closed in a compact set and therefore compact. Because neither the four quantities nor $\calA_t$ depend on $m$ or on any portfolio, the set $\Ast$ is the same for every feasible target, so a change of target moves the portfolio inside the fixed class and never reselects information. Appendix~\ref{app:formal} states the hypotheses in full.\label{prop:separation}

The finite candidate set is the case the market experiments use; the continuous case, in which candidates are directions in the driver space and the admissible class is a compact subset of the sphere, is treated in Appendix~\ref{app:proofs}.

In conditional form, Markowitz portfolio choice reads
\begin{equation}\label{eq:markowitz}
w^{\star}_{H}\in\argmin_{w\in\calW(H)}\Big\{\tfrac{\gamma}{2}\,w^{\Tr}\Sigma_{H}w-\mu_{H}^{\Tr}w+c_{H}(w)\Big\},
\end{equation}
where $H$ is the information relative to which $\mu_H,\Sigma_H$ are defined and $c_H$ a cost term; the representation $H$ is an input. The admissible problem makes it a variable.

\subsubsection{Two-stage solution and joint envelope}
\begin{definition}[Admissible envelope problem]\label{def:apo}
Let $J:\{(H,w):H\in\calA_t,\,w\in\calW(H)\}\to\R\cup\{+\infty\}$ be the population decision loss. The admissible envelope problem is
\begin{equation}\label{eq:apo}
V^{\star}=\inf_{H\in\calA_t}\ \inf_{w\in\calW(H)} J(w;H),
\end{equation}
and any minimizer $(\Hs,w^{\star})$ is an envelope-optimal pair. The conditional value function is $V(H)=\inf_{w\in\calW(H)}J(w;H)$, so that $V^{\star}=\inf_{H\in\calA_t}V(H)$. Problem~\eqref{eq:apo} ranks admissible representations by decision loss. The name APO is reserved for the two-stage problem of Definition~\ref{def:apo-canonical}, in which admissible representations are ranked by admissibility and loss enters only inside the selected class;~\eqref{eq:apo} is its envelope in the sense of Proposition~\ref{prop:lex}, and it is the object to which the finite-sample bounds of Section~\ref{sec:finite} apply.
\end{definition}

Two facts about~\eqref{eq:apo} are immediate. If $H\eqv\tilde H$ then $V(H)=V(\tilde H)$, because the value depends on $H$ only through the primitives that define decision equivalence, so $V$ is well defined on $\calA_t/\!\eqv$ and the problem may be posed on the quotient; and $V^{\star}\le V(H_0)$ for every $H_0\in\calA_t$, since the fixed-representation problem is one point of the admissible one.\label{prop:quotient}

Existence for the envelope problem under the same hypotheses is Theorem~\ref{thm:existence-formal} in Appendix~\ref{app:formal}.\label{thm:existence}

The hypotheses of the general clause are verified on a continuous class of driver directions in Appendix~\ref{app:proofs} (Proposition~\ref{prop:continuous}); elsewhere in the paper the candidate set is finite. For a finite candidate set nothing beyond compactness of each feasible set is needed; for continuous classes the topology on the quotient and the continuity of $H\mapsto\calW(H)$ are modeling requirements.

\begin{definition}[APO as a partial order on admissible pairs]\label{def:apo-canonical}
Fix a feasible target $m$ and write $\Gamma_t(m)=\{(H,w):H\in\calA_t,\ w\in\calW(H),\ \mu_H(w)\ge m\}$ for the admissible pairs at that target. Define the strict partial order $\prec_{\mathrm{APO}}$ on $\Gamma_t(m)$ by
\begin{equation}\label{eq:apo-order}
(H',w')\prec_{\mathrm{APO}}(H,w)\iff H'\prec_{\calA}H\ \ \text{or}\ \ \big[H'=H\ \text{and}\ \rho_H(w')<\rho_H(w)\big].
\end{equation}
Admissible portfolio optimization at target $m$ is the problem of finding the $\prec_{\mathrm{APO}}$-minimal elements of $\Gamma_t(m)$; its solution set is denoted $\mathrm{APO}_t(m)$.
\end{definition}

The relation is the lexicographic order on the pair (admissibility of $H$, risk of $w$) with the second coordinate compared only for the same representation, hence a strict partial order; it compares representations by admissibility and portfolios by risk only within a representation, and two portfolios attached to different representations are incomparable. The requirement that loss must not select information is expressed by this incomparability.

The alternative to selecting first and optimizing second is to optimize over pairs at once. The joint envelope at target $m$ is the least risk attainable over all admissible pairs,
\begin{equation}\label{eq:joint}
\rho^{\mathrm{joint}}(m)=\inf_{H\in\calA_t}\rho^{\star}_H(m)
\end{equation}
in which the representation may change with the target; it encodes a different contract, and the two are related as follows.

\begin{proposition}[Two-stage solution and envelope as two orders on the same set]\label{prop:lex}
Let $m$ be a target that is feasible under every representation in $\calA_t$, that is, $m$ lies in the common target domain $\bigcap_{H\in\calA_t}\{\mu_H^{\Tr}w:w\in\calW(H)\}$. Then
\begin{equation}\label{eq:lex-char}
\mathrm{APO}_t(m)=\bigcup_{H\in\Ast}\{H\}\times\calE_H(m),
\end{equation}
so the two-stage construction of this section computes the $\prec_{\mathrm{APO}}$-minimal pairs. If $\prec_{\mathrm{APO}}$ is replaced by the total preorder $(H',w')\preceq_{\mathrm{env}}(H,w)\iff\rho_{H'}(w')\le\rho_H(w)$ on $\Gamma_t(m)$, the minimal elements are the pairs attaining the joint envelope~\eqref{eq:joint}; with $\rho$ replaced by $J$ and the target constraint dropped, they are the minimizers of the envelope problem~\eqref{eq:apo}.
\end{proposition}
\noindent The proof is in Appendix~\ref{app:proofs}.

Targets outside the common domain are excluded because a representation under which $m$ is infeasible contributes no pair to $\Gamma_t(m)$, and a less admissible representation could then be minimal in $\Gamma_t(m)$ for the sole reason that every more admissible one cannot reach the target; on the common domain this cannot happen. The two concepts differ only in whether risk may compare portfolios across representations. No admissibility order reproduces the envelope, since portfolios attached to different representations are incomparable under $\prec_{\mathrm{APO}}$; the envelope is the total preorder, appropriate when the representation is a free modeling choice, whereas APO is appropriate whenever the representation carries a claim.

Assume first that Stage~I identifies a single information-equivalence class and choose any representative $\Hst\in\Ast$. Conditional on it the portfolio problem is classical in form; in the quadratic case with full investment its solution is the conditional frontier of Proposition~\ref{prop:merton} (Appendix~\ref{app:used}) applied to $\mu_{\Hst},\Sigma_{\Hst}$.

\label{def:frontier}\subsubsection{The conditional efficient frontier}
Once the class is fixed, the second solution set is the classical one. For conditional moments $\mu_H$ and $\Sigma_H\succ0$ with $\mu_H$ not proportional to $\ind$, Merton's closed forms \cite{Merton1972} give the unique fully invested minimum-variance portfolio at every target, the equivalence of the target and risk-aversion parametrizations of the efficient branch, and, when a riskless rate below the mean of the minimum-variance portfolio is available, the tangency portfolio and capital allocation line (Proposition~\ref{prop:merton}, Appendix~\ref{app:used}); nothing in them depends on how $H$ was chosen, which is what makes the two stages separable, and Figure~\ref{fig:mpt-apo} places each classical object next to its counterpart in the admissible problem. For each feasible target $m$,
\begin{equation}\label{eq:frontier}
\calE^{\star}_t(m)=\argmin_{w\in\calW(\Hst):\ \mu_{\Hst}(w)\ge m}\rho_{\Hst}(w),\qquad
\rho^{\star}_t(m)=\min_{w\in\calW(\Hst):\ \mu_{\Hst}(w)\ge m}\rho_{\Hst}(w).
\end{equation}
The correspondence $m\mapsto\calE^{\star}_t(m)$ is the second solution set of APO and the graph of $m\mapsto\rho^{\star}_t(m)$ is its efficient frontier. Transaction and holding costs belong to this stage; the cost term $c_H$ of~\eqref{eq:markowitz} may be carried in the risk functional or in the feasible set, in which case the frontier is net of costs; costs never enter the first stage.

\begin{figure}[!htb]
\centering
\begin{tikzpicture}[scale=0.9,declare function={hyp(\y,\a,\b,\c)=\a+\b*(\y-\c)^2;}]
\begin{scope}
\draw[-{Latex[length=2mm]}] (0,0)--(6.3,0); \node[below,font=\small] at (3.1,-0.05) {risk $\sigma$};
\draw[-{Latex[length=2mm]}] (0,0)--(0,4.6); \node[rotate=90,font=\small] at (-0.45,2.3) {expected return $\mu$};
\fill[gray!10,domain=0.4:3.6,samples=60,variable=\y] plot ({hyp(\y,1.2,0.42,2.0)},\y) -- (5.2,3.6) -- (5.2,0.4) -- cycle; \node[gray!60,font=\footnotesize] at (4.0,1.0) {feasible set};
\draw[gray!60,dashed,domain=0.4:2.0,samples=40,variable=\y] plot ({hyp(\y,1.2,0.42,2.0)},\y);
\draw[thick,domain=2.0:3.9,samples=60,variable=\y] plot ({hyp(\y,1.2,0.42,2.0)},\y);
\fill (1.2,2.0) circle (1.8pt); \node[font=\footnotesize,anchor=west] at (1.3,2.0) {MVP};
\coordinate (rf) at (0,1.15); \fill (rf) circle (1.8pt); \node[left,font=\footnotesize] at (rf) {$r_f$};
\coordinate (T) at (1.656,3.042); \fill (T) circle (2pt); \node[below right,font=\footnotesize] at (T) {tangency};
\draw[blue!60!black,thick] (rf) -- ($(rf)!1.6!(T)$); \node[blue!60!black,font=\footnotesize,anchor=north west] at ($(rf)!0.22!(T)+(0.1,-0.02)$) {CAL};
\path ($(rf)!1.35!(T)$) coordinate (U); \fill[red!70!black] (U) circle (2pt);
\draw[red!70!black,shift={(U)},domain=-1.0:1.0,samples=40,variable=\t] plot (\t,{0.55*\t+0.35*\t*\t});
\node[red!70!black,font=\footnotesize,anchor=west] at ($(U)+(1.05,0.9)$) {indifference curve};
\end{scope}
\node[font=\small\bfseries] at (3.1,5.15) {Modern portfolio theory};
\begin{scope}[xshift=7.6cm]
\draw[-{Latex[length=2mm]}] (0,0)--(6.3,0); \node[below,font=\small] at (3.1,-0.05) {conditional risk $\rho_H$};
\draw[-{Latex[length=2mm]}] (0,0)--(0,4.6); \node[rotate=90,font=\small] at (-0.45,2.3) {conditional expected return $\mu_H$};
\draw[gray!55,dashed,domain=1.6:3.9,samples=50,variable=\y] plot ({hyp(\y,0.7,0.30,1.8)},\y);
\node[gray!55,font=\footnotesize,align=left,anchor=north west] at (3.2,1.75) {$H_1$ (dashed): violates the\\ chronological condition;\\ infeasible however good};
\draw[gray!85,domain=1.4:3.9,samples=50,variable=\y] plot ({hyp(\y,2.2,0.45,1.9)},\y);
\node[gray!85,font=\footnotesize,anchor=west] at (4.3,3.55) {$H_2\in\calA^{0}_t\setminus\calA^{\star}_t$};
\draw[thick,domain=1.9:3.9,samples=60,variable=\y] plot ({hyp(\y,1.5,0.42,1.9)},\y);
\fill (1.5,1.9) circle (1.8pt); \node[font=\footnotesize,anchor=west] at (1.6,1.9) {MVP$_{\Hst}$};
\coordinate (rf2) at (0,1.15); \fill (rf2) circle (1.8pt); \node[left,font=\footnotesize] at (rf2) {$r_f$};
\coordinate (T2) at (2.192,3.183); \fill (T2) circle (2pt); \node[below right,font=\footnotesize] at (T2) {tangency$_{\Hst}$};
\draw[blue!60!black,thick] (rf2) -- ($(rf2)!1.6!(T2)$); \node[blue!60!black,font=\footnotesize,anchor=north west] at ($(rf2)!0.22!(T2)+(0.1,-0.02)$) {CAL$_{\Hst}$};
\path ($(rf2)!1.35!(T2)$) coordinate (U2); \fill[red!70!black] (U2) circle (2pt);
\draw[red!70!black,shift={(U2)},domain=-1.0:1.0,samples=40,variable=\t] plot (\t,{0.55*\t+0.35*\t*\t});
\node[font=\footnotesize,anchor=north west,align=left] at (3.3,2.75) {Stage I fixes $\Hst$;\\ Stage II is classical on it};
\end{scope}
\node[font=\small\bfseries] at (10.7,5.15) {Admissible portfolio optimization};
\end{tikzpicture}
\caption{Modern portfolio theory and its admissible counterpart. Left: the classical picture, in which the feasible set (shaded), the minimum-variance portfolio (MVP), the efficient frontier, the capital allocation line (CAL) through the tangency portfolio and the investor's indifference curve are defined relative to one unquestioned set of moments. Right: every one of those objects reappears, relative to the conditional moments of the admissibility-optimal information $\Hst$. A candidate that violates an admissibility condition ($H_1$) is infeasible whatever its frontier, and a candidate that satisfies the conditions but is not admissibility-optimal ($H_2$) is not compared on risk and return. Preferences enter after the information has been fixed.}
\label{fig:mpt-apo}
\end{figure}
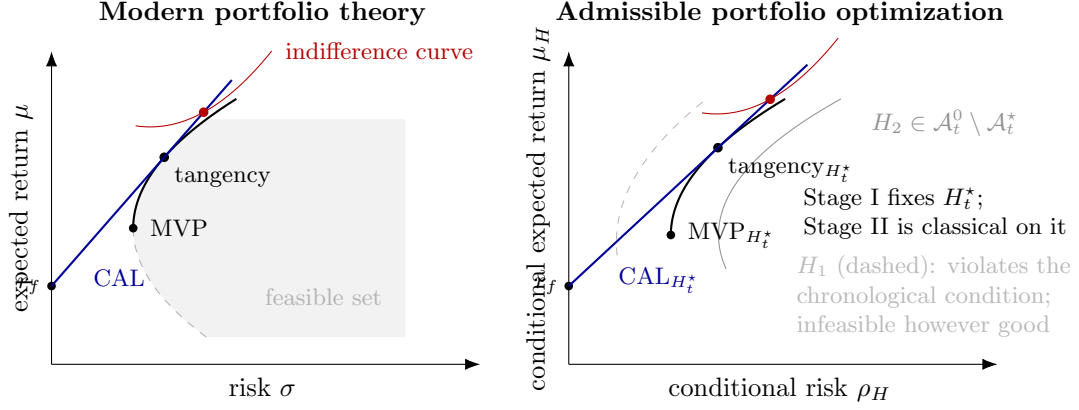

The first solution is generally set-valued. For $\Ast$ with several non-equivalent elements APO returns the correspondence $m\mapsto\{\rho^{\star}_H(m):H\in\Ast\}$; it is single-valued when all elements share the feasible set and the conditional mean and risk functionals, and equal admissibility records alone do not imply this (Example~\ref{ex:two-asset} in Appendix~\ref{app:used}, with equal defects, gives equal records and different frontiers).

When Stage~I ties, a predeclared implementation criterion may select one class before portfolio optimization; performance is never used to redefine admissibility. \label{sec:joint}Admissibility determines which information may define the decision; it does not force variance to be the risk measure, and it does not by itself carry the classical objects, which are now recovered conditional on the fixed class.

Inside the selected class every object of classical portfolio theory reappears conditional on its moments: the tangency portfolio, the capital allocation line and, because the moments depend on the decision state, a capital allocation surface whose root-mean-square slope is the unconditional Sharpe ratio of the conditional tangent strategy. Figure~\ref{fig:geometry} draws the frontier, the tangency portfolio and the capital allocation line at one decision state for the unconditional information, for a representation that fails separation and for the admissible class; the objects and their geometry across states are developed in Appendix~\ref{app:cal}, and the one result used later is that enlarging the information inside the admissible class cannot lower the Sharpe ratio of the conditional tangent strategy (Proposition~\ref{prop:sharpe}, Appendix~\ref{app:cal}).

\begin{figure}[!htb]
\centering
\begin{tikzpicture}
\begin{groupplot}[group style={group size=1 by 1},tick label style={font=\footnotesize},label style={font=\small}]
\nextgroupplot[width=11cm,height=7.2cm,xlabel={risk $\sigma$},ylabel={expected return $\mu$},xmin=0,xmax=1.05,ymin=0,ymax=0.75,legend style={legend columns=2}]
\addplot[thick,blue!70!black] coordinates {(0.284,0.000) (0.276,0.016) (0.269,0.032) (0.265,0.048) (0.264,0.064) (0.266,0.080) (0.270,0.096) (0.276,0.112) (0.285,0.128) (0.296,0.144) (0.309,0.160) (0.323,0.176) (0.339,0.192) (0.356,0.208) (0.374,0.224) (0.393,0.239) (0.413,0.255) (0.434,0.271) (0.455,0.287) (0.477,0.303) (0.499,0.319) (0.521,0.335) (0.544,0.351) (0.567,0.367) (0.591,0.383) (0.615,0.399) (0.639,0.415) (0.663,0.431) (0.687,0.447) (0.711,0.463) (0.736,0.479) (0.761,0.495) (0.785,0.510) (0.810,0.526) (0.835,0.542) (0.860,0.558) (0.885,0.574) (0.911,0.590) (0.936,0.606) (0.961,0.622) (0.986,0.638) (1.012,0.654) (1.037,0.670)}; \addlegendentry{admissible frontier $\mathcal F(\Hst)$}
\addplot[gray!80,dashed] coordinates {(0.288,0.014) (0.275,0.030) (0.267,0.046) (0.264,0.062) (0.267,0.078) (0.276,0.094) (0.289,0.109) (0.307,0.125) (0.328,0.141) (0.353,0.157) (0.379,0.173) (0.408,0.189) (0.438,0.205) (0.470,0.221) (0.502,0.237) (0.535,0.253) (0.569,0.269) (0.604,0.285) (0.639,0.301) (0.674,0.317) (0.710,0.333) (0.746,0.349) (0.782,0.365) (0.819,0.380) (0.855,0.396) (0.892,0.412) (0.929,0.428) (0.966,0.444) (1.003,0.460) (1.041,0.476)}; \addlegendentry{proxy frontier (fails separation)}
\addplot[black,dotted,thick] coordinates {(0.698,0.012) (0.504,0.028) (0.340,0.044) (0.269,0.059) (0.353,0.075) (0.521,0.091) (0.717,0.107) (0.923,0.123)}; \addlegendentry{unconditional frontier}
\addplot[blue!70!black,thick] coordinates {(0,0.010) (1.050,0.679)}; \addlegendentry{admissible CAL}
\addplot[gray!80,dashed] coordinates {(0,0.010) (1.050,0.488)};
\addplot[black,dotted,thick] coordinates {(0,0.010) (1.050,0.215)};
\addplot[red!70!black,thick] coordinates {(0.000,0.051) (0.020,0.052) (0.040,0.055) (0.060,0.060) (0.080,0.067) (0.100,0.076) (0.120,0.087) (0.140,0.100) (0.160,0.115) (0.180,0.132) (0.200,0.151) (0.220,0.172) (0.240,0.195) (0.260,0.220) (0.280,0.247) (0.300,0.276) (0.320,0.307) (0.340,0.340) (0.360,0.375)}; \addlegendentry{indifference curve, $\gamma=5$}
\addplot[only marks,mark=*,mark size=2pt,blue!70!black] coordinates {(0.264,0.063) (0.834,0.542)};
\addplot[only marks,mark=*,mark size=2pt,red!70!black] coordinates {(0.127,0.091)};
\addplot[only marks,mark=o,mark size=2pt,gray!80] coordinates {(0.625,0.294)};
\node[font=\scriptsize,anchor=west] at (axis cs:0.284,0.063) {GMV$_{\Hst}$};
\node[font=\scriptsize,anchor=south east] at (axis cs:0.824,0.552) {tangency$_{\Hst}$};
\node[font=\scriptsize,red!70!black,anchor=south west] at (axis cs:0.135,0.105) {optimum};
\node[font=\scriptsize,anchor=west] at (axis cs:0.01,0.035) {$r_f$};

\end{groupplot}
\end{tikzpicture}
\caption{The mean--variance geometry of the admissible problem at one decision state ($z=0.3$), from the population moments of the reference system with the admissibility-optimal class $\{Z_1,Z_2\}$, $r_f=0.01$ and risk aversion $\gamma=5$. The conditional efficient frontier of the admissible class carries its minimum-variance portfolio (GMV) and its tangency portfolio; the admissible capital allocation line runs from $r_f$ through the tangency portfolio, and the investor's indifference curve is tangent to it at the optimum. The frontier and line of a representation that fails the separator condition (the proxy, open marker at its tangency portfolio) and the unconditional frontier and line are flatter, so the same investor obtains less expected return per unit of risk; the difference is attributable to the choice of information.}
\label{fig:geometry}
\end{figure}
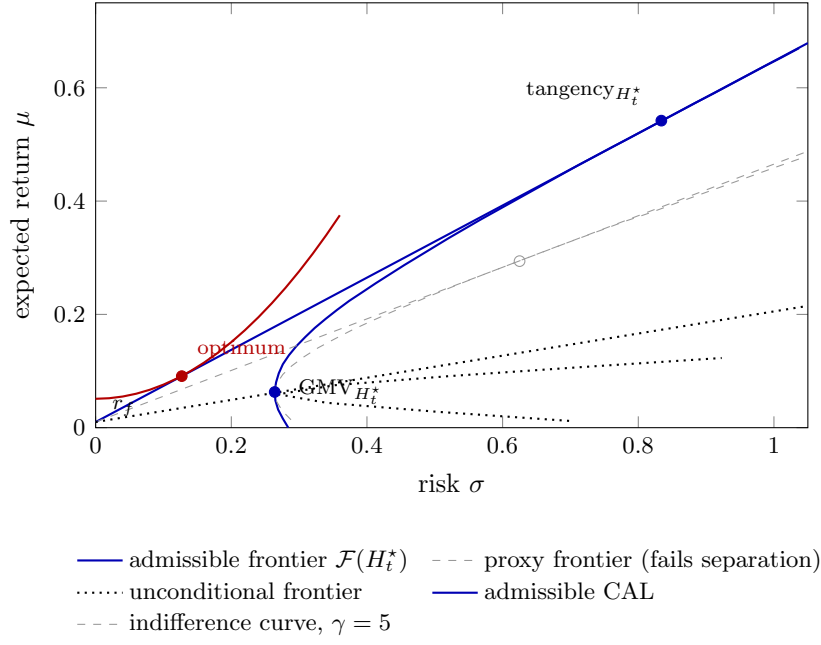

\subsection{Consequences}\label{sec:results}

Three results follow from the construction. Each is stated with the hypothesis that carries it, because in each case the hypothesis is where the content lies.

\subsubsection{Value of admissible information}\label{sec:voi}

That more information cannot hurt is Blackwell's theorem \cite{Blackwell1951,Blackwell1953}: for a finite parameter space, an experiment attains no larger Bayes risk than another in every decision problem with a finite action space and bounded loss, for every prior, if and only if the other is a garbling of it, an equivalence that extends to general parameter spaces under domination and is quantified by Le Cam's deficiency \cite{LeCam1964,Torgersen1991}. The paper uses only its elementary direction, in the form it takes for sub-$\sigma$-algebras: if $\sigma(H_1)\subseteq\sigma(H_2)$, every $\sigma(H_1)$-measurable policy is $\sigma(H_2)$-measurable. The theorem presupposes that the finer information is eligible and that the coarser policies remain implementable, and both presuppositions fail in the admissible problem, the first for look-ahead and non-public information and the second when the feasible set depends on the representation.

\begin{theorem}[Value of admissible information]\label{thm:voi}
Let $H_1,H_2\in\calA_t$ with $H_1\preceq H_2$ and suppose $H_2$ embeds the $H_1$ decision problem (Definition~\ref{def:order}). Then
\begin{equation}\label{eq:voi}
V(H_2)\le V(H_1).
\end{equation}
If moreover $V(H_1)$ is attained at some $\pi_1\in\Pi(H_1)$ and there exists $\pi_2\in\Pi(H_2)$ with $J(\pi_2;H_2)<J(\iota\pi_1;H_2)$, then $V(H_2)<V(H_1)$.
\end{theorem}
\noindent The proof is in Appendix~\ref{app:proofs}.

\begin{corollary}[The envelope need not choose maximal information]\label{cor:not-maximal}
Let $H_1\preceq H_2$. If $H_2\notin\calA_t$, or if $H_2\in\calA_t$ but no embedding $\iota$ exists, Theorem~\ref{thm:voi} is inapplicable and $V(H_2)\le V(H_1)$ may fail. In particular maximal raw information need not solve the envelope problem even when it is more informative in the Blackwell order.
\end{corollary}
\noindent The proof is in Appendix~\ref{app:proofs}.

Enlarging a financial filtration can violate the financial condition, a future variable improves prediction while failing chronology, and a richer observational representation need not carry interventional semantics: admissibility truncates the information order. Figure~\ref{fig:voichain} shows both the monotonicity inside the admissible class and the two ways it fails on a chain of representations of the reference system.

\begin{figure}[!htb]
\centering
\begin{tikzpicture}
\begin{axis}[width=11cm,height=6.2cm,ybar,bar width=14pt,xtick={0,1,2,3,4},xticklabels={{unconditional},{$\{Z_1\}$},{$\{Z_1,Z_2\}$},{$\{P,Z_2\}$},{$\{Z_1,Z_2,X^{+}\}$}},x tick label style={font=\scriptsize},ylabel={decision value $V(H)$ (lower is better)},ymin=-1.80,ymax=0.3,tick label style={font=\footnotesize},label style={font=\small},legend style={font=\footnotesize},legend cell align=left,nodes near coords,every node near coord/.append style={font=\scriptsize,/pgf/number format/fixed,/pgf/number format/precision=2}]
\addplot[fill=black!70,draw=black] coordinates {(0,0.1214) (1,-0.2206) (2,-0.5701)}; \addlegendentry{admissible, nested}
\addplot[fill=gray!40,draw=gray] coordinates {(3,-0.3059)}; \addlegendentry{admissible, fails separation}
\addplot[fill=red!60,draw=red!70!black] coordinates {(4,-1.5632)}; \addlegendentry{contains look-ahead $X^{+}$: inadmissible}
\end{axis}
\end{tikzpicture}
\caption{The value of information along a chain of representations, reference system, risk aversion $\gamma=5$, full investment, averaged over the state. Along the nested admissible chain the value improves, as Theorem~\ref{thm:voi} requires; the proxy pair, which carries less information than $\{Z_1,Z_2\}$, is worse; and the representation containing the look-ahead variable has by far the best value, which is the trap of Corollary~\ref{prop:lookahead}: a rule that selects by value selects it. The second failure of monotonicity is not shown here because it needs representation-dependent constraints: with the same system, $\{Z_1\}$ under a position bound of one has value $-0.218$ while the richer $\{Z_1,Z_2\}$ under a bound of $0.2$ has value $-0.158$, so more information with a tighter constraint is worth less, and no embedding exists.}
\label{fig:voichain}
\end{figure}
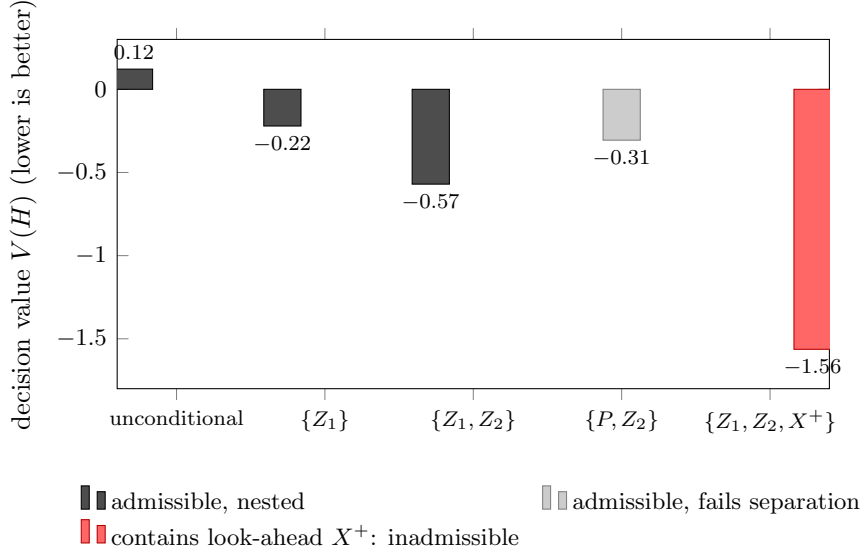

The failure of chronology has a consequence that is immediate from the theorem but worth stating, because it is the mechanism by which look-ahead enters a data-driven selection.

\begin{corollary}[Decision-loss selection under look-ahead contamination]\label{prop:lookahead}
Let $H$ be a representation with decision-time component $H^{-}$, and let $X^{+}$ be a random variable measurable with respect to the horizon-closed information and not $\sigma(H^{-})$-measurable. Let $H^{\sharp}$ be the representation with decision-time component $\sigma(H^{-},X^{+})$ and $\calW(H^{\sharp})=\calW(H)$. Then $V(H^{\sharp})\le V(H)$; and the inequality is strict whenever the $H^{\sharp}$-conditional optimizer differs from the $H^{-}$-measurable optimizer on a set of positive probability, which for the quadratic loss with $\calW$ containing a neighbourhood of the optimizer holds whenever $\Cov(R\mid H^{\sharp})^{-1}\E[R\mid H^{\sharp}]\neq\Cov(R\mid H^{-})^{-1}\E[R\mid H^{-}]$ on a set of positive probability. Consequently, on any candidate set that contains $H$ and $H^{\sharp}$, no rule that selects by minimizing $V$, the envelope problem included, can select $H$; and if $H$ minimizes $V$ among the candidates that exclude $X^{+}$, or if every such candidate has its own look-ahead augmentation in the set, every $V$-minimizer contains $X^{+}$; decision-loss selection is consistent only on candidate sets from which the chronological admissibility condition has removed every such candidate.
\end{corollary}
\noindent The proof is in Appendix~\ref{app:proofs}.

The look-ahead of this corollary should be distinguished from the one that empirical practice already controls. The standard notion is a property of the data, a value used at the decision node that was not available then, through a revised series, a publication lag, a survivorship-biased universe or a misaligned timestamp, and it is removed by point-in-time data and by purging the evaluation sample. The corollary is a property of the selection rule. With every series correctly dated, a rule that chooses the conditioning information by the decision loss it produces prefers, by the first direction of Theorem~\ref{thm:voi}, the finest information available in the candidate set, and horizon information is the finest there is; the rule does not fail to detect look-ahead, it seeks it, and no loss-based criterion can be asked to police it. The same mechanism has an overfitting form that involves no future variable at all, in which the rule selects whatever information fits the sample at hand, and Section~\ref{sec:sidebyside} measures it on market data where every timestamp is correct. The chronological condition is therefore placed before the loss, as a hard constraint on the candidate set, which is why no risk--return quantity enters the admissibility order; and the paper uses horizon information where it is legitimate, in the certification of separation on a block that is never used for the decision, and excludes it from the decision-time component on which portfolios are measurable.

The proposition is the population form of the fact, measured in Section~\ref{sec:evidence}, that the look-ahead representation attains the lowest in-sample loss in every replication: the chronological condition cannot be replaced by a loss criterion, because the criterion is how look-ahead enters.

\label{sec:two-stage}

The information problem and the portfolio problem must be kept separate. The first selects a class by admissibility, and only after that class is fixed does the second construct the efficient set conditional on it (Figure~\ref{fig:architecture}). The conditional portfolio problem is used in two parametrizations of one frontier, the utility form with risk aversion $\gamma$ and the target form with return target $m$, related by Proposition~\ref{prop:merton}.

\begin{figure}[!htb]
\centering
\resizebox{\textwidth}{!}{\begin{tikzpicture}[
  box/.style={draw,rounded corners=2pt,align=center,text width=30mm,minimum height=11mm,inner sep=3pt,font=\small},
  stage/.style={box,fill=blue!7,thick},
  arr/.style={-{Latex[length=2.2mm]},thick}]
\node[box]   (data)  at (0,0)    {point-in-time data; candidate set $\calH_t$};
\node[box]   (cond) at (4.1,0)  {admissibility conditions satisfied};
\node[box]   (adm)   at (8.2,0)  {admissible class $\calA_t$};
\node[stage] (s1)    at (12.3,0) {\textbf{Stage I}\\ $\Ast=\Min_{\prec_{\calA}}\calA_t$};
\node[box]   (fix) at (12.3,-2.6) {fix $\Hst$ (or its $\eqv$-class)};
\node[box]   (prim)   at (8.2,-2.6)  {conditional primitives $\mu_{\Hst},Q_{\Hst},\calW(\Hst)$};
\node[stage] (s2)     at (4.1,-2.6)  {\textbf{Stage II}\\ $\calE^{\star}_t(m)$, $\rho^{\star}_t(m)$};
\node[box,draw=none,fill=none,font=\small\itshape] (note) at (0,-2.6) {target return $m$ enters only here};
\draw[arr] (data)--(cond); \draw[arr] (cond)--(adm); \draw[arr] (adm)--(s1);
\draw[arr] (s1)--(fix); \draw[arr] (fix)--(prim); \draw[arr] (prim)--(s2);
\draw[arr,dashed] (note)--(s2);
\end{tikzpicture}}
\caption{The two-solution architecture. Candidates violating a required admissibility condition are infeasible, not penalized. Stage~I ranks the survivors by the admissibility order alone and fixes the minimal class. Stage~II constructs the conditional efficient set and frontier inside that class; portfolio performance never relaxes a violated admissibility condition.}
\label{fig:architecture}
\end{figure}
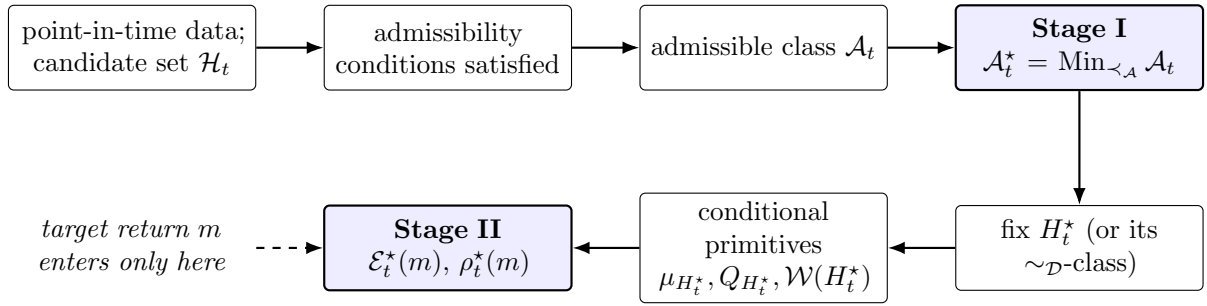

\subsubsection{Diversification under admissible information}\label{sec:diversification}
Suppose the conditional mean is affine in a statistic of the driver path as in Proposition~\ref{prop:tower}, so that with $u_H:=R-\E[R\mid H^{+}]$ the conditional return innovation decomposes as
\begin{equation}\label{eq:decomp}
R-\mu_H=B_H\xi_H+u_H,\quad \Cov(\xi_H\mid H)=V_H,\quad \Cov(u_H,\xi_H\mid H)=0,\quad \Omega_H:=\Cov(u_H\mid H)\succeq0 ,
\end{equation}
where the orthogonality of $u_H$ and $\xi_H$ is a consequence of the construction, $u_H$ having zero conditional mean given $H^{+}$ while $\xi_H$ is $H^{+}$-measurable. Then $\Sigma_H=B_HV_HB_H^{\Tr}+\Omega_H$, and a risk functional built on the response term alone agrees with covariance risk on $\calW(H)$ if and only if $w\mapsto w^{\Tr}\Omega_Hw$ is constant there, in particular when $\Omega_H=0$; otherwise the two rank portfolios differently, and the difference is the residual term the rest of this subsection studies.\label{prop:residual}

By Proposition~\ref{prop:tower} an admissible representation may induce a response operator $B_H\in\R^{n\times k}$ and an innovation covariance $V_H\in\R^{k\times k}$. Define the represented response risk and the response-based objective
\begin{equation}\label{eq:response-risk}
Q^{\mathrm{resp}}_{H}=B_HV_HB_H^{\Tr},\qquad
J^{\mathrm{resp}}(w;H)=\tfrac{\gamma}{2}\,w^{\Tr}Q^{\mathrm{resp}}_{H}w-r_H^{\Tr}w .
\end{equation}

The causal origin of the separator condition is the Commonality Principle \cite{RD2023}, which states that the drivers relevant to a portfolio's diversification are the common causes of its constituents, and whose structural proof \cite{RD2025causal} runs through Reichenbach's common-cause principle \cite{Reichenbach1956} and the causal Markov condition: in a structural causal model in which the declared drivers are the complete set of common causes of the constituents, with no latent common cause and no causal link among the constituents, the causal Markov condition \cite{Pearl2009,Spirtes2000} makes the constituents mutually independent given the drivers, so the common causes are a structural separator and in particular a separator, while omitting a common cause leaves dependence and conditioning on a common effect creates it. Conditioning on the common causes places their innovations in the systematic term and leaves the residual diagonal; the principle is recorded as Proposition~\ref{prop:commonality} in Appendix~\ref{app:used}. The framework does not assume the model; it makes the principle's conclusion, separation, a certifiable condition. Markowitz diversification acts on the quadratic form $w\mapsto w^{\Tr}\Sigma w$, and under an approximate factor structure the idiosyncratic part is diversified away while the factor part is not (Figure~\ref{fig:divcurve}) \cite{Ross1976,ChamberlainRothschild1983}. Part (a) below is that argument; the content of the proposition is that the split it presupposes is a property of the conditioning class and is stable under interventions only under interventional admissibility.

\begin{figure}[!htb]
\centering
\begin{tikzpicture}
\begin{axis}[width=10.5cm,height=6.2cm,xlabel={number of assets $n$ (equal weights)},ylabel={portfolio variance},xmin=1,xmax=60,ymin=0,tick label style={font=\footnotesize},label style={font=\small},legend style={font=\footnotesize},legend cell align=left]
\addplot[thick,red!70!black] coordinates {(1,1.0093) (2,0.7193) (3,0.6226) (4,0.5742) (5,0.5452) (6,0.5259) (7,0.5121) (8,0.5017) (9,0.4937) (10,0.4872) (11,0.4820) (12,0.4776) (13,0.4738) (14,0.4707) (15,0.4679) (16,0.4655) (17,0.4633) (18,0.4614) (19,0.4597) (20,0.4582) (21,0.4568) (22,0.4556) (23,0.4544) (24,0.4534) (25,0.4524) (26,0.4515) (27,0.4507) (28,0.4499) (29,0.4492) (30,0.4486) (31,0.4479) (32,0.4473) (33,0.4468) (34,0.4463) (35,0.4458) (36,0.4453) (37,0.4449) (38,0.4445) (39,0.4441) (40,0.4437) (41,0.4434) (42,0.4430) (43,0.4427) (44,0.4424) (45,0.4421) (46,0.4418) (47,0.4416) (48,0.4413) (49,0.4411) (50,0.4408) (51,0.4406) (52,0.4404) (53,0.4402) (54,0.4400) (55,0.4398) (56,0.4396) (57,0.4394) (58,0.4392) (59,0.4390) (60,0.4389)}; \addlegendentry{representation fails separation: true variance}
\addplot[thick,blue!60!black,dashed] coordinates {(1,1.0093) (2,0.6797) (3,0.5698) (4,0.5149) (5,0.4820) (6,0.4600) (7,0.4443) (8,0.4325) (9,0.4234) (10,0.4160) (11,0.4100) (12,0.4050) (13,0.4008) (14,0.3972) (15,0.3941) (16,0.3913) (17,0.3889) (18,0.3867) (19,0.3848) (20,0.3831) (21,0.3815) (22,0.3801) (23,0.3788) (24,0.3776) (25,0.3765) (26,0.3755) (27,0.3745) (28,0.3737) (29,0.3728) (30,0.3721) (31,0.3714) (32,0.3707) (33,0.3701) (34,0.3695) (35,0.3689) (36,0.3684) (37,0.3679) (38,0.3675) (39,0.3670) (40,0.3666) (41,0.3662) (42,0.3658) (43,0.3654) (44,0.3651) (45,0.3648) (46,0.3644) (47,0.3641) (48,0.3638) (49,0.3636) (50,0.3633) (51,0.3630) (52,0.3628) (53,0.3625) (54,0.3623) (55,0.3621) (56,0.3619) (57,0.3617) (58,0.3615) (59,0.3613) (60,0.3611)}; \addlegendentry{separator, or what the diagonal model reports}
\addplot[thick,black!50] coordinates {(1,0.3501) (2,0.3501) (3,0.3501) (4,0.3501) (5,0.3501) (6,0.3501) (7,0.3501) (8,0.3501) (9,0.3501) (10,0.3501) (11,0.3501) (12,0.3501) (13,0.3501) (14,0.3501) (15,0.3501) (16,0.3501) (17,0.3501) (18,0.3501) (19,0.3501) (20,0.3501) (21,0.3501) (22,0.3501) (23,0.3501) (24,0.3501) (25,0.3501) (26,0.3501) (27,0.3501) (28,0.3501) (29,0.3501) (30,0.3501) (31,0.3501) (32,0.3501) (33,0.3501) (34,0.3501) (35,0.3501) (36,0.3501) (37,0.3501) (38,0.3501) (39,0.3501) (40,0.3501) (41,0.3501) (42,0.3501) (43,0.3501) (44,0.3501) (45,0.3501) (46,0.3501) (47,0.3501) (48,0.3501) (49,0.3501) (50,0.3501) (51,0.3501) (52,0.3501) (53,0.3501) (54,0.3501) (55,0.3501) (56,0.3501) (57,0.3501) (58,0.3501) (59,0.3501) (60,0.3501)}; \addlegendentry{systematic floor $\beta^{\Tr}V\beta$}
\end{axis}
\end{tikzpicture}
\caption{The textbook diversification curve and what admissibility adds to it. Equal-weight portfolio variance against the number of assets, with the systematic variance and residual variance of the reference system and a residual correlation of $0.12$ left by a representation that fails separation, the portfolio-weighted average observed on the equity panel. Under a separator the residual decays as $1/n$ to the systematic floor, which is the classical picture and also what a model with diagonal residual reports whatever the representation. When the representation does not separate, the true curve plateaus at $\rho\,\sigma^{2}$ above the floor, Proposition~\ref{prop:diversification}(b): the gap between the two curves at large $n$ is risk counted as diversifiable that is not, and it is what the consequence statistic of part (d) measures for a given portfolio.}
\label{fig:divcurve}
\end{figure}

\begin{proposition}[Diversification is a property of the admissible class]\label{prop:diversification}
Let~\eqref{eq:decomp} hold for a representation $H$ and let $\calW\subseteq\{w\in\R^{n}:\ind^{\Tr}w=1\}$.
\begin{enumerate}[label=(\alph*)]
\item (\emph{Separation makes the residual idiosyncratic.}) If $H$ satisfies the separator condition $G_S$ in the exact sense of Definition~\ref{def:separator}, and $\xi_H$ in~\eqref{eq:decomp} is the innovation of $H^{+}$ relative to $H^{-}$, then $\Omega_H$ is diagonal, and for every $w\in\calW$
\begin{equation}\label{eq:idio-bound}
w^{\Tr}\Omega_Hw\ \le\ \Big(\max_{i}\Omega_{H,ii}\Big)\,\|w\|_{2}^{2},
\end{equation}
so the residual risk of an equal-weight portfolio, or of any portfolio whose squared Euclidean norm is of order $1/n$, is of order $\max_i\Omega_{H,ii}/n$ and vanishes along any sequence of models in which $n\to\infty$ with $\sup_i\Omega_{H,ii}$ uniformly bounded, the approximate-factor-structure setting of \cite{ChamberlainRothschild1983}. The systematic component $w^{\Tr}B_HV_HB_H^{\Tr}w$ admits no such bound: it equals $\beta_w^{\Tr}V_H\beta_w$ with $\beta_w:=B_H^{\Tr}w$ the portfolio's exposure to the common drivers, and is unaffected by spreading weights among assets with equal exposures.
\item (\emph{Failure of the admissibility condition misclassifies risk.}) If $H$ fails $G_S$, so that $\Omega_H$ has a nonzero off-diagonal block $\Omega^{\mathrm{off}}_H$, then $w^{\Tr}\Omega_Hw$ equals $\sum_i w_i^2\Omega_{H,ii}$ plus $w^{\Tr}\Omega^{\mathrm{off}}_Hw$, and the second term does not vanish under spreading: for equal weights it equals $n^{-2}\ind^{\Tr}\Omega^{\mathrm{off}}_H\ind$, whose order is that of the average off-diagonal residual covariance rather than $1/n$. Common variation not explained by $H$ is then counted as diversifiable when it is not.
\item (\emph{Interventional admissibility makes the mechanism invariant.}) Let $H$ be admissible for an interventional contract with declared regime set $\mathcal E$, so that $\calL^{e}(R\mid H^{+})=\calL^{\mathrm{o}}(R\mid H^{+})$ for every $e\in\mathcal E$ by Definition~\ref{def:claim}. Then for every $e\in\mathcal E$ the response $B_H$, the residual covariance $\Omega_H$ and the diagonality in (a) are the same under the post-intervention law as under the observational one. The systematic term $B_HV_HB_H^{\Tr}$ is invariant in addition only if the law of the driver innovation given $H^{-}$, and hence $V_H$, is also invariant, which is a condition on the drivers' own dynamics and not on the response mechanism; an intervention can leave the mechanism $R\mid H^{+}$ unchanged while changing how much the drivers move, in which case the diversifiable component is the same before and after and the systematic component is not.
\item (\emph{What the failure costs a given portfolio.}) Let $s_i=\Omega_{H,ii}^{1/2}$, $v=w\odot s$, $c_{ij}=\Omega_{H,ij}/(s_is_j)$ the residual correlations, and $D=(\sum_i v_i)^{2}-\sum_i v_i^{2}$. If $D\neq0$, the variance that the diagonal-residual restriction discards satisfies the exact decomposition
\begin{equation}\label{eq:misallocation}
\frac{w^{\Tr}\Omega^{\mathrm{off}}_Hw}{w^{\Tr}\Sigma_Hw}=\underbrace{\frac{\|v\|_2^{2}}{w^{\Tr}\Sigma_Hw}}_{\text{residual share}}\;\times\;\underbrace{\Big(\frac{(\sum_iv_i)^{2}}{\|v\|_2^{2}}-1\Big)}_{\text{breadth}-1}\;\times\;\underbrace{\frac{1}{D}\sum_{i\neq j}v_iv_jc_{ij}}_{\bar c_w},
\end{equation}
where $\bar c_w$ is the portfolio-weighted average residual correlation. If $v\ge0$ then $|\bar c_w|\le\varepsilon_t(H)$, so the relative misallocation is at most the screening defect times the residual share times breadth minus one. The defect therefore bounds the cost but does not order it: it is a maximum over pairs, whereas the cost is governed by $\bar c_w$, a weighted average.
\end{enumerate}
\end{proposition}
\noindent The proof is in Appendix~\ref{app:proofs}. Under a conditional-law contract part (c) has no basis: an intervention may sever the channel through which $H$ explained the covariation, after which $\Omega_H$ acquires off-diagonal mass and part (b) applies; Experiment~\ref{exp:E3} measures this case. The proposition is the portfolio-level form of the Commonality Principle stated in Proposition~\ref{prop:commonality}: Markowitz diversification acts on $\Sigma$, admissible diversification on the residual $\Omega_{\Hst}$ of the selected class, and causal diversification requires the class to be interventionally admissible.


Two consequences of the decomposition can be stated as propositions, one about when the separator condition can be met at all and one about what the certification permits when it is not. The first explains the ordering across universes that Section~\ref{sec:market} finds.

\begin{proposition}[Aggregation raises the defect]\label{prop:aggregation}
Let the residuals of $N$ securities given a representation $H$ satisfy $u_i=f+e_i$, where $f$ is a common component that $H$ leaves unexplained, with $\Var f=\sigma_f^{2}>0$, and the $e_i$ are pairwise uncorrelated and uncorrelated with $f$, with $\Var e_i=\sigma_e^{2}$. For two disjoint equally weighted baskets of $k$ securities the residual correlation given $H$ is
\[
\rho_k=\frac{\sigma_f^{2}}{\sigma_f^{2}+\sigma_e^{2}/k},
\]
strictly increasing in $k$ with $\rho_k\to1$. The separator condition at tolerance $\tau\in(0,1)$ holds for baskets of size $k$ if and only if $k\le k^{\star}(\tau):=\tau\,\sigma_e^{2}/\big((1-\tau)\,\sigma_f^{2}\big)$; for $k>k^{\star}(\tau)$ no representation that leaves $f$ unexplained is admissible, whatever else it contains.
\end{proposition}
\begin{proof}
For disjoint baskets $A$ and $B$ with equal weights $1/k$, $\bar u_A=f+\bar e_A$ with $\Var\bar e_A=\sigma_e^{2}/k$, so $\Cov(\bar u_A,\bar u_B)=\sigma_f^{2}$ and $\Var\bar u_A=\Var\bar u_B=\sigma_f^{2}+\sigma_e^{2}/k$, which gives $\rho_k$. The map $k\mapsto\rho_k$ is increasing with limit one, and $\rho_k\le\tau$ rearranges to $k\le k^{\star}(\tau)$.
\end{proof}

Averaging within an asset removes the idiosyncratic part that a representation was never needed to explain and leaves the common part that it cannot explain, so pre-diversified assets fail the condition for a reason that is structural with respect to any representation leaving the common component unexplained, rather than a shortage of candidates within the family searched; with heterogeneous idiosyncratic variances the same holds with $\sigma_e^{2}$ replaced by the within-basket average. The second proposition gives the covariance a certified class supports when its defect is not zero.

\begin{proposition}[Certified robust covariance]\label{prop:robust}
Let $H$ have residual standard deviations $s=(\Omega_{H,ii}^{1/2})_i$ and certified defect $\varepsilon=\varepsilon_t(H)$, so that every residual correlation satisfies $|c_{ij}|\le\varepsilon$, and let $v=w\odot s$. Over all correlation matrices $C$ with $|c_{ij}|\le\varepsilon$ for $i\neq j$,
\[
\max_{C}\ w^{\Tr}\big(B_HV_HB_H^{\Tr}+\operatorname{diag}(s)\,C\,\operatorname{diag}(s)\big)w
=w^{\Tr}B_HV_HB_H^{\Tr}w+(1-\varepsilon)\|v\|_2^{2}+\varepsilon\|v\|_1^{2},
\]
attained at $C^{\star}=(1-\varepsilon)I+\varepsilon\,\sigma\sigma^{\Tr}$ with $\sigma=\operatorname{sign}(v)$, which is positive semidefinite with unit diagonal. For a long-only portfolio $C^{\star}=(1-\varepsilon)I+\varepsilon\ind\ind^{\Tr}$ does not depend on $w$, and the minimax portfolio is the minimum-variance portfolio of the covariance $\Sigma^{\varepsilon}_H:=B_HV_HB_H^{\Tr}+(1-\varepsilon)\operatorname{diag}(\Omega_H)+\varepsilon\,ss^{\Tr}$; the same holds with $\varepsilon$ replaced by any $\rho\in[0,\varepsilon]$ as a declared level, $\rho=0$ being the diagonal model.
\end{proposition}
\begin{proof}
$w^{\Tr}\operatorname{diag}(s)C\operatorname{diag}(s)w=\sum_i v_i^{2}+\sum_{i\neq j}v_iv_jc_{ij}\le\|v\|_2^{2}+\varepsilon\sum_{i\neq j}|v_i||v_j|=\|v\|_2^{2}+\varepsilon(\|v\|_1^{2}-\|v\|_2^{2})$, with equality at $c_{ij}=\varepsilon\operatorname{sign}(v_iv_j)$; that choice is $C^{\star}$, which is a convex combination of $I$ and $\sigma\sigma^{\Tr}$ and hence positive semidefinite with unit diagonal. If $v\ge0$ then $\sigma=\ind$ and $C^{\star}$ is free of $w$, so the inner maximum is a fixed quadratic form and the minimax problem is a minimum-variance problem for $\Sigma^{\varepsilon}_H$.
\end{proof}

The robust covariance is the bound of part~(d) taken as an objective: the diagonal model completed by an equicorrelated residual at the level the certification permits, with no parameter to tune. Section~\ref{sec:sidebyside} reports it at the certification-block level $\rho=\bar c$ and finds what the perturbation bound predicts, that it repairs the promised risk and does not change the realized one.

\subsubsection{The price of causal admissibility}\label{sec:causal}
Causality changes the feasible information class when the claim requires interventional meaning. With the claim levels of Definition~\ref{def:claim}, let $\calA_t(\mathrm{c})$ denote the representations admissible for a conditional-law contract and $\calA_t(\mathrm{i})\subseteq\calA_t(\mathrm{c})$ those admissible for an interventional contract, i.e.\ additionally carrying an identified intervention response under the declared design.

\label{def:causal-apo}For an interventional contract,
\begin{equation}\label{eq:causal-apo}
V^{\star}_{\mathrm{i}}=\inf_{H\in\calA_t(\mathrm{i})}\ \inf_{w\in\calW(H)}J(w;H);
\end{equation}
for a conditional-law contract the problem is posed over $\calA_t(\mathrm{c})$ with value $V^{\star}_{\mathrm{c}}$. For the same population loss and nested classes the causal admissibility price is
\begin{equation}\label{eq:price}
\Delta_{\mathrm{ci}}:=V^{\star}_{\mathrm{i}}-V^{\star}_{\mathrm{c}}.
\end{equation}

Since $\calA_t(\mathrm{i})\subseteq\calA_t(\mathrm{c})$, $V^{\star}_{\mathrm{c}}\le V^{\star}_{\mathrm{i}}$ and $\Delta_{\mathrm{ci}}\ge0$. The inequality is not evidence against causality, since the interventional problem solves a stronger semantic contract, and $\Delta_{\mathrm{ci}}$ is the oracle price of that contract at fixed numerical loss. Conversely a lower conditional-law loss does not license an interventional interpretation.

\begin{theorem}[Distance bound for the causal admissibility price]\label{thm:distance}
Let $(\calI,d)$ be a metric space of decision-equivalence classes containing $\calA_t(\mathrm{c})$ and suppose $V$ is $L$-Lipschitz on $\calA_t(\mathrm{c})$, i.e.\ $|V(H)-V(G)|\le L\,d(H,G)$ for $H,G\in\calA_t(\mathrm{c})$. If $\Hs_0\in\calA_t(\mathrm{c})$ satisfies $V(\Hs_0)=V^{\star}_{\mathrm{c}}$ and $\calA_t(\mathrm{i})\neq\emptyset$, then
\begin{equation}\label{eq:distance}
0\le\Delta_{\mathrm{ci}}\le L\,d(\Hs_0,\calA_t(\mathrm{i})),\qquad d(H,\calA_t(\mathrm{i})):=\inf_{G\in\calA_t(\mathrm{i})}d(H,G).
\end{equation}
For the quadratic loss $J(w;H)=\tfrac{\gamma}{2}w^{\Tr}\Sigma_Hw-\mu_H^{\Tr}w$ on a common feasible set with $\bar r:=\sup_{w\in\calW}\|w\|_2<\infty$, the hypothesis holds with $L=\bar r$ for the moment metric
\begin{equation}\label{eq:moment-metric}
d(H,G):=\|\mu_H-\mu_G\|_2+\tfrac{\gamma}{2}\bar r\,\|\Sigma_H-\Sigma_G\|_{\mathrm{op}},
\end{equation}
and, when the moments are random functions of the decision state, for its expectation over the state.
\end{theorem}
\noindent The proof is in Appendix~\ref{app:proofs}.

On a finite class every function is Lipschitz in the discrete metric, so the bound has content only for a decision-relevant metric such as~\eqref{eq:moment-metric}.

The causal admissibility price is therefore controlled by the distance, in conditional mean and covariance, between the best predictive representation and the nearest interventionally admissible one; in Experiment~\ref{exp:E3} this distance is zero before the intervention, which is why proxy and state tie. Requiring identification removes otherwise useful representations in exchange for interventional, transport or counterfactual meaning; it is neither a general advantage nor a general prerequisite. When identification establishes an intervention-response map $\Psi^{\mathrm{do}}_H$, the response operator may be taken as its derivative at the operating point, $B^{\mathrm{do}}_H=D\Psi^{\mathrm{do}}_H(z_0)$, and only then may a response-based objective be called causal-response optimization; a causal graph selected by an algorithm does not by itself supply it.


\section{Finite-sample theory and evidence}\label{sec:finite}

This section prices the selection from data, verifies the theory on a controlled system, and then asks of market data whether the separator condition is attainable, what its failure costs, what the selected class delivers under the standard allocation rule, and how the three ways of choosing a portfolio's information compare. The finite-sample results are given here and formally in Appendix~\ref{app:formal}; the controlled experiments that verify them are in Section~\ref{sec:controlled}.

\subsection{Finite-sample guarantees}\label{sec:guarantees}

This subsection prices the selection from data: both stages are consistent, the regret of selecting jointly is bounded by a uniform deviation over the class, the oracle inequality trades the causal price against search complexity, and reselection across nodes is bounded by a hysteresis rule.

Let $\widehat\calA^{0}_{t}$ be produced on a certification block and the four quantities of Definition~\ref{def:adm-order} estimated without the final evaluation block, and define $\widehat\calA^{\star}_{t}=\Min_{\widehat\prec_{\calA}}\widehat\calA^{0}_{t}$ with $\widehat\prec_{\calA}$ the admissibility order computed from the estimated quantities. For a finite class the admissibility order is a total preorder on a finite set and therefore admits a real-valued representation: there is $s_t:\calA_t\to\R$ with $H_1\prec_{\calA}H_2\iff s_t(H_1)<s_t(H_2)$, and $\Ast=\argmin s_t$. We state consistency for such a score.

\subsubsection{Consistency of the two stages}
Stage~I is selection-consistent on a finite candidate set whenever the certified class equals the admissible class with probability tending to one and the estimated score converges uniformly, provided the population score has a positive gap at the optimum; with a shrinking tolerance the selected class is admissible with probability tending to one (Proposition~\ref{prop:stage1-formal}, Appendix~\ref{app:formal}).\label{prop:stage1-consistency}

Given a consistent Stage~I, the estimated frontier converges uniformly to the population frontier on compact interior ranges of targets, and the estimated efficient portfolios to the population efficient set, under compactness of the feasible set, continuity and convexity of the risk functional, continuity of the mean and uniform convergence of the estimated moments (Proposition~\ref{prop:stage2-formal}, Appendix~\ref{app:formal}).\label{prop:stage2-consistency}

\label{rem:dependence}Three qualifications concern the sample. If the feasible set is estimated with Hausdorff distance $\delta_N\pto0$ and $\rho,\mu$ are Lipschitz near $\calW$, the sandwich in the proof of Proposition~\ref{prop:stage2-consistency} holds with $\delta_N$-neighbourhoods and an additional term $L_\rho\delta_N$, and the conclusion is unchanged.
All finite-sample statements are for independent observations. Returns are serially dependent; under $\beta$-mixing the same conclusions hold with contiguous blocks separated by a gap in place of independent observations and with the effective sample size reduced accordingly, by the blocking arguments of \cite{Yu1994}. The certification, estimation and evaluation blocks are disjoint for a second reason: the frontier is estimated after a selection over the candidate set, and without splitting the Stage-II estimates would be subject to the selection effect and their confidence statements invalid. Splitting is the post-selection device that makes Proposition~\ref{prop:stage2-consistency} conditional on the selected class alone.

Let $\widehat J_N(w;H)$ be an estimated loss on a selection sample and
\begin{equation}\label{eq:joint-hat}
(\widehat H,\hat w)\in\argmin_{H\in\widehat\calA_N,\ w\in\calW(H)}\widehat J_N(w;H),
\end{equation}
where the empirical admissible set $\widehat\calA_N$ is itself estimated and can make false-pass or false-fail errors. Consider first the case in which the population admissible class is known.

\subsubsection{Regret of joint selection}
\begin{proposition}[Uniform decision-regret bound]\label{thm:regret}
Assume $\calA$ is known and that on an event $E_N$ of probability at least $1-\delta$,
\begin{equation}\label{eq:uniform-dev}
\sup_{H\in\calA}\ \sup_{w\in\calW(H)}\big|\widehat J_N(w;H)-J(w;H)\big|\le\varepsilon_N .
\end{equation}
Let $(\Hs,w^{\star})$ minimize $J$ and $(\widehat H,\hat w)$ minimize $\widehat J_N$ over $\{(H,w):H\in\calA,w\in\calW(H)\}$. Then on $E_N$,
\begin{equation}\label{eq:regret}
J(\hat w;\widehat H)-J(w^{\star};\Hs)\le2\varepsilon_N .
\end{equation}
\end{proposition}
\noindent The proof is in Appendix~\ref{app:proofs}.

Representation choice raises the statistical burden because the uniform deviation must hold over the joint class, not over the weights of one model.

If the admissible class is itself estimated by a certification procedure independent of the selection sample, and on an event of probability at least $1-\alpha$ that procedure includes no inadmissible representation and retains a population-optimal class, then~\eqref{eq:regret} holds over the certified class with probability at least $1-\alpha-\delta$ by the same chain of inequalities and the union bound; outside that event no admissibility or causal conclusion follows from decision performance.

Hence the sample is split, so that certification decides eligibility on one block and only certified representations are ranked on another; reusing the same outcomes to relax the admissibility conditions destroys the interpretation of admissibility. The generic bound does not explain when a stronger admissibility condition can help statistically. The bound reflects the approximation--estimation tradeoff of oracle inequalities for model selection \cite{Bartlett2008}, here for a claim-dependent admissible class.

\subsubsection{Oracle inequality}
\begin{theorem}[Finite-class oracle inequality]\label{thm:oracle}
Let $\calA_r=\{H_1,\dots,H_{M_r}\}$ be the admissible class for a conditional-law ($r=\mathrm{c}$) or interventional ($r=\mathrm{i}$) contract, with $\calA_t(\mathrm{i})\subseteq\calA_t(\mathrm{c})$. Suppose the inner optimized empirical values $\widehat V_N(H)$ are computed on data independent of condition certification and satisfy, for every fixed $H\in\calA_t(\mathrm{c})$ and every $x>0$,
\begin{equation}\label{eq:subgauss}
\Prob\{|\widehat V_N(H)-V(H)|>x\}\le2\exp\{-Nx^{2}/(2\sigma^{2})\}.
\end{equation}
Let $\widehat H_r\in\argmin_{H\in\calA_r}\widehat V_N(H)$. Then for every $\delta\in(0,1)$, with probability at least $1-\delta$,
\begin{equation}\label{eq:oracle}
V(\widehat H_r)-\min_{H\in\calA_r}V(H)\le2\sigma\sqrt{\frac{2\log(2M_r/\delta)}{N}},
\end{equation}
and, relative to the conditional-law oracle,
\begin{equation}\label{eq:oracle-c2}
V(\widehat H_{\mathrm{i}})-V^{\star}_{\mathrm{c}}\le\Delta_{\mathrm{ci}}+2\sigma\sqrt{\frac{2\log(2M_{\mathrm{i}}/\delta)}{N}} .
\end{equation}
\end{theorem}
\noindent The proof is in Appendix~\ref{app:proofs}.

\begin{corollary}[Causal admissibility as regularization]\label{cor:regularization}
Suppose $M_{\mathrm{i}}<M_{\mathrm{c}}$. The interventional class has larger oracle loss by $\Delta_{\mathrm{ci}}$ and a smaller search term, and the upper bound~\eqref{eq:oracle-c2} is tighter than the conditional-law bound~\eqref{eq:oracle} with $r=\mathrm{c}$ whenever
\begin{equation}\label{eq:regularization}
\Delta_{\mathrm{ci}}<2\sigma\sqrt{\frac{2}{N}}\Big[\sqrt{\log(2M_{\mathrm{c}}/\delta)}-\sqrt{\log(2M_{\mathrm{i}}/\delta)}\Big].
\end{equation}
\end{corollary}
\noindent The proof is in Appendix~\ref{app:proofs}.

This is structural risk minimization \cite{Vapnik1998} over a nesting fixed by the claim rather than tuned to the data, with the causal price as approximation term and the size of the eligible class as complexity term; it is a regularization statement, not a claim that causal representations predict better. For infinite classes $\log M_r$ is replaced by a complexity measure and the deviation assumption by a uniform empirical-process bound \cite{BLM2013}; when candidates are nested a chaining argument sharpens the logarithmic term.

Consider a sequence of decision nodes $t=1,\dots,T$, each with its own contract, candidate set and admissible class, and a portfolio process $(w_t)$ that may be constrained intertemporally, for instance by transaction costs $c(w_t-w_{t-1})$. The multi-period problem is to choose an admissible pair at every node.

When the admissibility conditions and quantities at node $t$ depend only on node-$t$ information, the multi-period Stage-I problem is the product of the one-node problems, whatever the intertemporal constraints on portfolios, because nothing in the order at node $t$ involves a portfolio; given the selected classes, Stage~II is a dynamic program coupled only through the cost $c(w_t-w_{t-1})$.\label{prop:separable} What remains is how the \emph{estimated} class should change between nodes.

\subsubsection{Reselection across decision nodes}
When the problem is solved at successive nodes $t=1,2,\dots$, Stage-I reselection is itself a source of turnover. Let $\hat s_t$ be the estimated score of Proposition~\ref{prop:stage1-consistency} and $\widehat H_{t-1}$ the incumbent class. \label{def:hysteresis}The rule proposed here is a hysteresis rule: for a predeclared margin $\eta>0$, set $\widehat H_t=\widehat H_{t-1}$ if $\hat s_t(\widehat H_{t-1})\le\min_{\widehat\calA^{0}_t}\hat s_t+\eta$, and $\widehat H_t\in\argmin\hat s_t$ otherwise.

\begin{proposition}[Reselection only on population change]\label{prop:hysteresis}
Suppose $\max_{H}|\hat s_t(H)-s_t(H)|\le\delta_t$ at every node.

\smallskip\noindent\textit{Part (a).}\quad A reselection at node $t$ implies $s_t(\widehat H_{t-1})-\min s_t>\eta-2\delta_t$.

\smallskip\noindent\textit{Part (b).}\quad If the four population quantities are constant in $t$, the incumbent is population-optimal and $\delta_t<\eta/2$ for all $t$, no reselection occurs.

\smallskip\noindent\textit{Part (c).}\quad Every reselection satisfies $s_t(\widehat H_t)\le\min s_t+2\delta_t$.
\end{proposition}
\noindent The proof is in Appendix~\ref{app:proofs}.

The margin $\eta$ is in the units of the score and must exceed twice the estimation error for part (b); under a rolling protocol the incumbent is certified only on observations that post-date its fixing. A larger margin suppresses reselection driven by noise at the cost of a longer lag when the population changes, and because it is declared before any portfolio is evaluated it does not reintroduce performance into Stage~I.

\label{sec:evidence-sec}

The finite-sample results of Section~\ref{sec:finite} rest on one population fact about the quadratic Stage-II problem, stated here because it involves no sampling: the regret of a portfolio computed from perturbed conditional moments is second order in the perturbation.

\subsubsection{Perturbation bound for the conditional portfolio}
\begin{proposition}[Regret is second order in estimation error]\label{prop:regret2}
Let $J(w)=\tfrac{\gamma}{2}w^{\Tr}\Sigma w-\mu^{\Tr}w$ and $\widehat J(w)=\tfrac{\gamma}{2}w^{\Tr}\widehat\Sigma w-\hat\mu^{\Tr}w$ on a convex feasible set $\calW$ with $\sup_{\calW}\|w\|_2\le\bar r$, and suppose $\Sigma\succeq\lambda I$ with $\lambda>0$. Let $w^{\star}$ minimize $J$ and $\hat w$ minimize $\widehat J$ over $\calW$. Then
\begin{equation}\label{eq:regret2}
J(\hat w)-J(w^{\star})\ \le\ \frac{2}{\gamma\lambda}\Big(\gamma\bar r\,\|\widehat\Sigma-\Sigma\|_{\mathrm{op}}+\|\hat\mu-\mu\|_2\Big)^{2}.
\end{equation}
In particular, if $\hat\mu$ and $\widehat\Sigma$ are $\sqrt N$-consistent, the regret of the estimated portfolio is $O_p(N^{-1})$.
\end{proposition}
\noindent The proof is in Appendix~\ref{app:proofs}.

Figure~\ref{fig:firstsecond} draws the two orders of magnitude on the reference system of Section~\ref{sec:order}: the representation $\{P,Z_2\}$, which fails separation with a maximal residual correlation of $0.68$, is given the diagonal-residual covariance it would report, its frontier is compared with the true frontier, and the portfolios it chooses are evaluated at their true risk.

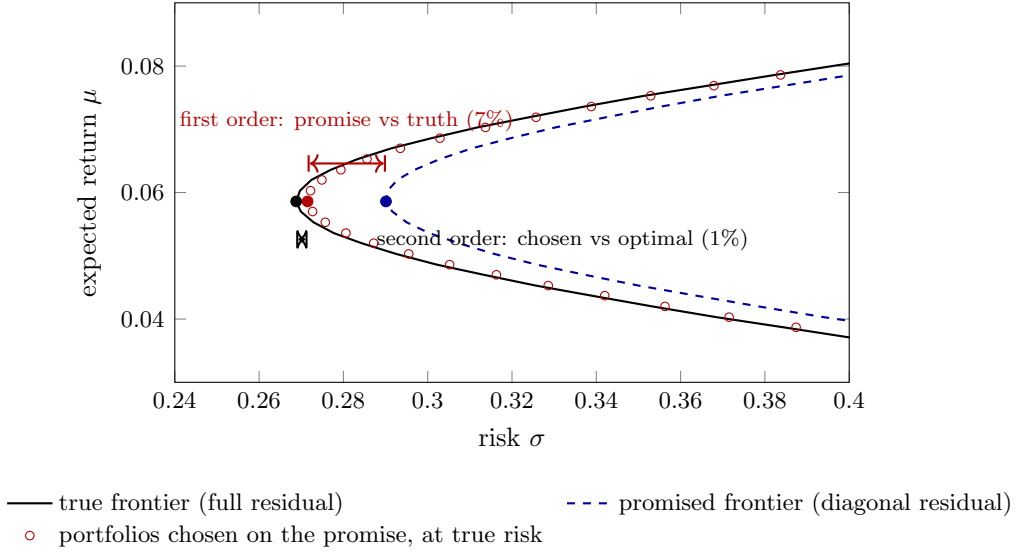
\begin{figure}[!htb]
\centering
\begin{tikzpicture}
\begin{axis}[width=10.5cm,height=6.6cm,xlabel={risk $\sigma$},ylabel={expected return $\mu$},xmin=0.24,xmax=0.40,ymin=0.030,ymax=0.090,yticklabel style={/pgf/number format/fixed,/pgf/number format/precision=3},xticklabel style={/pgf/number format/fixed,/pgf/number format/precision=2},tick label style={font=\footnotesize},label style={font=\small},legend style={font=\footnotesize},legend cell align=left,scaled y ticks=false]
\addplot[thick,black] coordinates {(0.5302,0.0253) (0.5107,0.0270) (0.4915,0.0287) (0.4726,0.0303) (0.4540,0.0320) (0.4359,0.0337) (0.4182,0.0353) (0.4010,0.0370) (0.3844,0.0387) (0.3685,0.0403) (0.3533,0.0420) (0.3389,0.0437) (0.3256,0.0453) (0.3133,0.0470) (0.3022,0.0486) (0.2925,0.0503) (0.2843,0.0520) (0.2777,0.0536) (0.2729,0.0553) (0.2699,0.0570) (0.2688,0.0586) (0.2697,0.0603) (0.2724,0.0620) (0.2770,0.0636) (0.2834,0.0653) (0.2914,0.0670) (0.3009,0.0686) (0.3118,0.0703) (0.3240,0.0719) (0.3372,0.0736) (0.3514,0.0753) (0.3665,0.0769) (0.3824,0.0786) (0.3989,0.0803) (0.4160,0.0819) (0.4336,0.0836) (0.4517,0.0853) (0.4702,0.0869) (0.4891,0.0886) (0.5083,0.0903)}; \addlegendentry{true frontier (full residual)}
\addplot[thick,blue!60!black,dashed] coordinates {(0.5602,0.0253) (0.5401,0.0270) (0.5203,0.0287) (0.5009,0.0303) (0.4819,0.0320) (0.4632,0.0337) (0.4451,0.0353) (0.4274,0.0370) (0.4104,0.0387) (0.3940,0.0403) (0.3785,0.0420) (0.3637,0.0437) (0.3500,0.0453) (0.3373,0.0470) (0.3259,0.0486) (0.3158,0.0503) (0.3073,0.0520) (0.3003,0.0536) (0.2950,0.0553) (0.2916,0.0570) (0.2901,0.0586) (0.2905,0.0603) (0.2928,0.0620) (0.2969,0.0636) (0.3029,0.0653) (0.3105,0.0670) (0.3197,0.0686) (0.3303,0.0703) (0.3423,0.0719) (0.3554,0.0736) (0.3695,0.0753) (0.3846,0.0769) (0.4005,0.0786) (0.4171,0.0803) (0.4344,0.0819) (0.4522,0.0836) (0.4706,0.0853) (0.4894,0.0869) (0.5086,0.0886) (0.5282,0.0903)}; \addlegendentry{promised frontier (diagonal residual)}
\addplot[only marks,mark=o,mark size=1.6pt,red!70!black] coordinates {(0.5332,0.0253) (0.5137,0.0270) (0.4945,0.0287) (0.4756,0.0303) (0.4570,0.0320) (0.4389,0.0337) (0.4212,0.0353) (0.4040,0.0370) (0.3874,0.0387) (0.3715,0.0403) (0.3563,0.0420) (0.3420,0.0437) (0.3286,0.0453) (0.3163,0.0470) (0.3052,0.0486) (0.2955,0.0503) (0.2872,0.0520) (0.2806,0.0536) (0.2757,0.0553) (0.2727,0.0570) (0.2715,0.0586) (0.2722,0.0603) (0.2749,0.0620) (0.2794,0.0636) (0.2856,0.0653) (0.2935,0.0670) (0.3029,0.0686) (0.3137,0.0703) (0.3257,0.0719) (0.3388,0.0736) (0.3529,0.0753) (0.3679,0.0769) (0.3837,0.0786) (0.4002,0.0803) (0.4172,0.0819) (0.4348,0.0836) (0.4528,0.0853) (0.4713,0.0869) (0.4901,0.0886) (0.5093,0.0903)}; \addlegendentry{portfolios chosen on the promise, at true risk}
\draw[|<->|,thick,red!70!black] (axis cs:0.2715,0.0646) -- (axis cs:0.2901,0.0646) node[midway,above=8pt,font=\scriptsize] {first order: promise vs truth ($7\%$)};
\draw[|<->|,thick,black] (axis cs:0.2688,0.0526) -- (axis cs:0.2715,0.0526) node[right=22pt,font=\scriptsize] {second order: chosen vs optimal ($1\%$)};
\addplot[only marks,mark=*,mark size=2pt,black] coordinates {(0.2688,0.0586)};
\addplot[only marks,mark=*,mark size=2pt,blue!60!black] coordinates {(0.2901,0.0586)};
\addplot[only marks,mark=*,mark size=2pt,red!70!black] coordinates {(0.2715,0.0586)};
\end{axis}
\end{tikzpicture}
\caption{Why the information decision shows in reported risk and not in realized risk. Population moments of the reference system with the proxy pair $\{P,Z_2\}$, which fails the separator condition (maximal residual correlation $0.68$). The dashed frontier is what a model with diagonal residual promises; the solid frontier is the truth. The portfolios chosen on the promise (circles), evaluated at their true risk, lie within $1\%$ of the true frontier, the second-order effect of Proposition~\ref{prop:regret2}, while the promise itself is off by $7\%$ at the same target, the first-order effect. In this system the discarded term is negative and the promise overstates risk; on the equity panel it is positive and the promise understates it; the proposition allows either sign, and in both cases the allocation is nearly right and the reported risk is not.}
\label{fig:firstsecond}
\end{figure}
 The frontier is first order in the moments and inherits their $N^{-1/2}$ error; the regret at a declared risk aversion is second order and converges at $N^{-1}$, which is why the sensitivity of the frontier to errors in means does not carry over to the decision.


\subsection{Implementation}\label{sec:implementation}\label{sec:practice}
This subsection states the procedure a practitioner runs, the comparison it supports with mean--variance practice, and what should and should not be expected from it. The procedure, of which the market experiments of Section~\ref{sec:market} are worked instances, is the following.

\begin{enumerate}[label=(\arabic*),leftmargin=*]
\item \emph{Declare the contract.} Fix the asset universe, horizon, rebalancing calendar, portfolio constraints, claim level with its regime set, candidate driver universe with its release calendar, tolerance $\bar\varepsilon$ and tie-break quantities, before looking at returns; nothing in later steps may alter them.
\item \emph{Split the sample.} For each window, three disjoint blocks, certification, estimation and evaluation; the split is the post-selection device of Section~\ref{sec:finite}.
\item \emph{Certify.} On the certification block apply the required conditions to each candidate: chronology and publicity from the release calendar, separation by the maximal absolute residual correlation after conditioning, causality by invariance across the declared regimes; record the four quantities of Definition~\ref{def:adm-order} for every candidate that passes.
\item \emph{Select.} Rank the passing candidates by the order of Definition~\ref{def:adm-order}(b) and fix the minimal class; if none passes, the admissible class is empty and the output is that the declared universe does not support the claim at the declared tolerance. Across rebalancing dates apply the hysteresis rule.
\item \emph{Estimate and optimize.} On the estimation block regress returns on the selected drivers to obtain $\widehat B$, $\widehat V$ and $\widehat\Omega$, form $\widehat B\widehat V\widehat B^{\Tr}+\operatorname{diag}(\widehat\Omega)$ and the fitted conditional mean at the decision-time state, and solve the classical problem inside the selected class only.
\item \emph{Evaluate once.} On the evaluation block report realized risk, the realized-to-predicted variance ratio, out-of-sample residual dependence, turnover and reselections, without returning to step 4.
\end{enumerate}

What a practitioner takes from the framework is listed in Section~\ref{sec:contrib}; the numbers that support each item are those of Section~\ref{sec:market} and the sections that follow it, and the limits stated there hold with the same numbers.

\subsection{Controlled experiments}\label{sec:controlled}\label{sec:evidence}
Experiments~\ref{exp:E1} to~\ref{exp:E4} are controlled, with a known data-generating structure against which every claim can be checked. Figure~\ref{fig:backtest} shows the fact they share. On a single sample the frontiers that dominate in the backtest are those of candidates containing look-ahead information, which do not exist at the decision time. They are implemented as specified in Appendix~\ref{app:config}. The market experiment on public data is reported in Section~\ref{sec:market}.

The designs are the following. Experiment~\ref{exp:E1} verifies exact recovery: on the reference system of Section~\ref{sec:order} the admissibility order must return the exact separators $\{Z_1,Z_2\}$ and their re-codings and must exclude every candidate containing $X^{+}$ even though those candidates have the lowest decision loss. The look-ahead variable is constructed as data-level look-ahead would never be: $X^{+}$ is a sum of the idiosyncratic shocks realized over the holding horizon, a quantity that becomes known only at $t+h$, and it is offered to every selection rule as if it were an ordinary candidate at $t$. A rule that ranks candidates by in-sample decision loss, computed on the certification block with $X^{+}$ in the information set, prefers the candidates that contain it in every window; the admissibility order excludes them before any loss is computed because they fail the chronological condition, which is checked on the variable's measurability and not on its performance.

Experiment~\ref{exp:E2} draws samples of size $N\in\{250,\dots,8000\}$ from the system, $300$ replications each, splits them 50/25/25 into certification, estimation and evaluation blocks, runs the two stages and records whether Stage~I returned the population class, the uniform error of the estimated frontier against the population frontier, and which candidate minimized in-sample loss. Experiment~\ref{exp:E3} rescales the standard deviation of $Z_1$ by $1.5$ and recomputes the diversifiable component under the admissible and under a proxy representation. Experiment~\ref{exp:E4} runs the hysteresis rule across a sequence of nodes with the population drifting and records reselections against a fixed margin. 

Two facts from these runs are worth seeing before the detail. Figure~\ref{fig:e2} shows the Stage-I error rates and the Stage-II frontier error falling with the sample size at the rates the theory gives, with the look-ahead candidate winning on in-sample loss in every replication, which is the theorem and its corollary verified exactly. Figure~\ref{fig:backtest} shows the second, computed on a single sample of $N=2000$ from the same system: each candidate's frontier is estimated on the certification and estimation blocks and then evaluated on the evaluation block with the information available at the decision time; on that sample, every frontier that dominates in the backtest belongs to a candidate containing look-ahead information, and every one of them lies to the right of the admissible frontier once evaluated with the information available at the decision time.

\refstepcounter{experiment}\label{exp:E1}\subsubsection{Experiment~\theexperiment: exact recovery of the admissible class and Markowitz nesting}
A Gaussian structural system with eight assets and two common states generates six candidate variables: the states $Z_1,Z_2$, an affine recoding $\phi(Z_1)=3Z_1-2$, a noisy proxy $P=Z_1+0.7e$, an independent noise $\nu$ and a look-ahead variable $X^{+}$ equal to a sum of horizon idiosyncratic shocks. The candidate set contains the 41 nonempty subsets of size at most three; the chronological admissibility condition removes the 16 containing $X^{+}$. The population screening defect $\varepsilon_t(H)$ is the maximal absolute conditional residual correlation, computed analytically. The four identities hold. The look-ahead candidate has $\varepsilon_t=0.435$ and is removed before ranking irrespective of its predictive value. The sets $\{Z_1,Z_2\}$ and $\{\phi(Z_1),Z_2\}$ have identical conditional covariances (maximal entrywise difference $1.1\times10^{-16}$) and both attain $\varepsilon_t=0$ at dimension two, so Stage~I returns their common class; every other admissible candidate has $\varepsilon_t\ge0.299$ (the proxy pair $\{P,Z_2\}$) or dimension three. The Stage-II frontier computed from the fixed class coincides with Markowitz computed from the same conditional moments to machine precision (maximal difference $0$ on a 41-point target grid).

\begin{figure}[!htb]
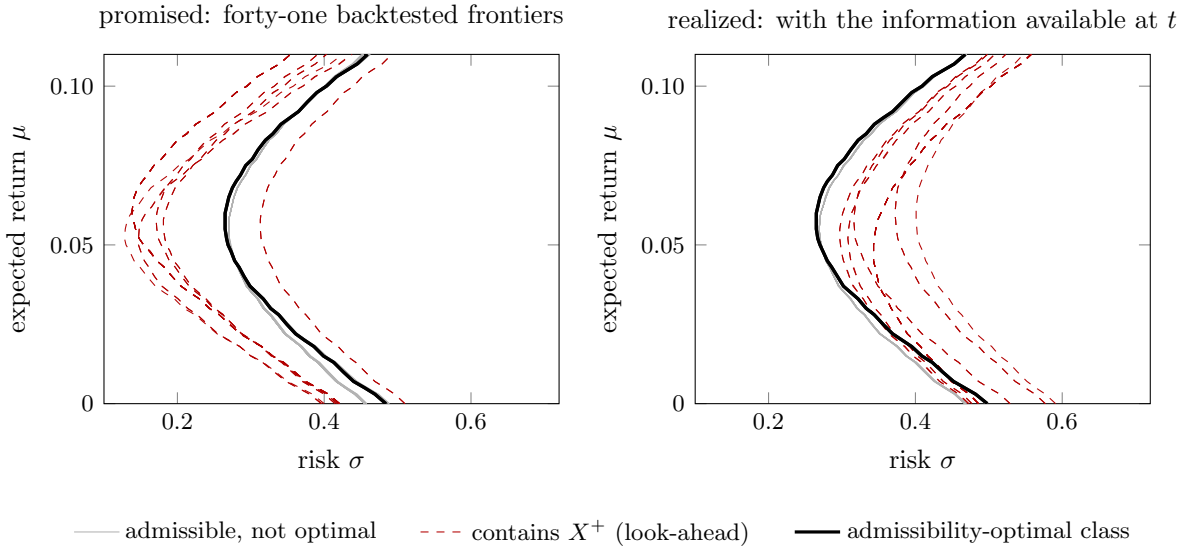

\centering

\caption{Backtest versus reality on the reference system, one sample of $2{,}000$ observations. Left: the in-sample frontier of every one of the 41 candidate conditioning sets, at the mean state. The frontiers that dominate are those containing the look-ahead variable $X^{+}$ (dashed); the admissibility-optimal class (thick) is not the best backtest. Right: the same portfolios evaluated under the information actually available at the decision time, so that $X^{+}$ is replaced by what is known of it, namely nothing. The look-ahead frontiers move to the right of the admissible frontier, whose promised and realized risks nearly coincide. Selecting information on backtest performance would have chosen the frontiers that do not exist; the chronological admissibility condition removes them before any frontier is drawn.}
\label{fig:backtest}
\end{figure}

\refstepcounter{experiment}\label{exp:E2}\subsubsection{Experiment~\theexperiment: finite-sample selection and estimation}
The same system is estimated from $N$ observations split $50/25/25$ into discovery, certification and evaluation blocks, with $\hat\varepsilon_t$ the maximal absolute residual correlation on the certification block, tolerance $0.10+2/\sqrt{N_{\mathrm{cert}}}$ (the declared level plus the $97.5\%$ Fisher bound for a correlation; the size-controlled multiplier-bootstrap calibration of the maximum \cite{Chernozhukov2013} is used alongside it below and on the market panels), ties broken by dimension and then by $\hat\varepsilon_t$, 300 replications. Table~\ref{tab:e2} and Figure~\ref{fig:e2} report the Stage-I error rates and the Stage-II frontier error $\sup_m|\hat\rho^{\star}(m)-\rho^{\star}(m)|$ over the central 60\% of the target range. Class recovery reaches $0.99$ at $N=1{,}000$ and $1$ from $N=2{,}000$ (Proposition~\ref{prop:stage1-consistency}).

The frontier error with estimated means decreases slowly, the mean-error sensitivity of \cite{ChopraZiemba1993}; with the mean known it halves for each fourfold increase in $N$. The last column shows that the look-ahead representation attains lower in-sample decision loss than the true class in essentially every replication. The decision-relevant quantity is the regret $J(\hat w;\Hst)-J(w^{\star};\Hst)$ at the declared risk aversion under the true moments, and Table~\ref{tab:regret} shows it decaying as $N^{-1}$, the rate of Proposition~\ref{prop:regret2}, from 0.035 at $N=250$ to 0.0008 at $N=8{,}000$, with Jorion-type shrinkage \cite{Jorion1986} of the conditional mean reducing it by 35\% at $N=250$ and 12\% at $N=8{,}000$. The global-minimum-variance error, a covariance-only quantity, decays at the parametric rate (0.0124 to 0.0013), and the multiplier-bootstrap tolerance reaches class recovery one from $N=1{,}000$; both are in Appendix~\ref{app:config}.

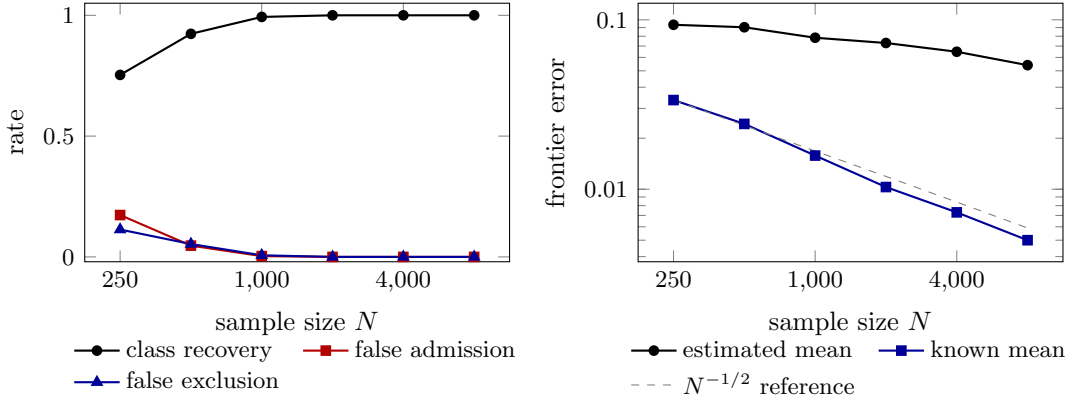
\begin{figure}[!htb]
\centering
\begin{tikzpicture}
\begin{groupplot}[group style={group size=2 by 1,horizontal sep=1.7cm},width=7.2cm,height=5.0cm,tick label style={font=\footnotesize},label style={font=\small},legend style={font=\footnotesize},xmode=log,log ticks with fixed point,xtick={250,1000,4000},xlabel={sample size $N$}]
\nextgroupplot[ylabel={rate},ymin=-0.02,ymax=1.05,]
\addplot[thick,mark=*,mark size=1.6pt] coordinates {(250,0.7533) (500,0.9233) (1000,0.9933) (2000,1.0000) (4000,1.0000) (8000,1.0000)}; \addlegendentry{class recovery}
\addplot[thick,mark=square*,mark size=1.6pt,red!70!black] coordinates {(250,0.1733) (500,0.0467) (1000,0.0033) (2000,0.0000) (4000,0.0000) (8000,0.0000)}; \addlegendentry{false admission}
\addplot[thick,mark=triangle*,mark size=2pt,blue!60!black] coordinates {(250,0.1133) (500,0.0533) (1000,0.0067) (2000,0.0000) (4000,0.0000) (8000,0.0000)}; \addlegendentry{false exclusion}
\nextgroupplot[ylabel={frontier error},ymode=log]
\addplot[thick,mark=*,mark size=1.6pt] coordinates {(250,0.0935) (500,0.0904) (1000,0.0784) (2000,0.0730) (4000,0.0648) (8000,0.0540)}; \addlegendentry{estimated mean}
\addplot[thick,mark=square*,mark size=1.6pt,blue!60!black] coordinates {(250,0.0336) (500,0.0243) (1000,0.0158) (2000,0.0103) (4000,0.0073) (8000,0.0050)}; \addlegendentry{known mean}
\addplot[dashed,gray] coordinates {(250,0.0336) (500,0.0238) (1000,0.0168) (2000,0.0119) (4000,0.0084) (8000,0.0059)}; \addlegendentry{$N^{-1/2}$ reference}
\end{groupplot}
\end{tikzpicture}
\caption{Experiment~\ref{exp:E2}. Left: Stage-I error rates against sample size (300 replications). Right: median Stage-II frontier error $\sup_m|\hat\rho^{\star}(m)-\rho^{\star}(m)|$ with the conditional mean estimated and with the mean known; the dashed line is proportional to $N^{-1/2}$. Selection is consistent, the covariance-driven part of the frontier error decays at the parametric rate, and the mean-driven part is the classical sensitivity to errors in means.}
\label{fig:e2}
\end{figure}

\begin{table}[!htb]\centering\footnotesize
\caption{Experiment~\ref{exp:E2}: Stage-I selection and Stage-II frontier error against sample size (300 replications; medians for frontier errors). Rates carry a binomial Monte Carlo standard error of at most $0.03$.}
\label{tab:e2}
\begin{tabular}{@{}rcccccc@{}}
\toprule
$N$ & \shortstack{false\\admission} & \shortstack{false\\exclusion} & \shortstack{class\\recovery} & \shortstack{frontier\\error} & \shortstack{frontier error,\\known mean} & \shortstack{look-ahead wins\\on loss}\\
\midrule
250 & 0.173 & 0.113 & 0.753 & 0.094 & 0.0336 & 0.997\\
500 & 0.047 & 0.053 & 0.923 & 0.090 & 0.0243 & 1.000\\
1,000 & 0.003 & 0.007 & 0.993 & 0.078 & 0.0158 & 1.000\\
2,000 & 0.000 & 0.000 & 1.000 & 0.073 & 0.0103 & 1.000\\
4,000 & 0.000 & 0.000 & 1.000 & 0.065 & 0.0073 & 1.000\\
8,000 & 0.000 & 0.000 & 1.000 & 0.054 & 0.0050 & 1.000\\
\bottomrule
\end{tabular}
\end{table}

\begin{table}[!htb]\centering\small
\caption{Experiment~\ref{exp:E2}: median decision regret $J(\hat w;\Hst)-J(w^{\star};\Hst)$ at $\gamma=5$ under the true conditional moments, with the sample mean and with Jorion-type shrinkage toward the grand mean (300 replications).}
\label{tab:regret}
\begin{tabular}{@{}rcc@{}}
\toprule
$N$ & sample mean & shrunk mean\\
\midrule
250 & 0.0350 & 0.0228\\
500 & 0.0159 & 0.0115\\
1,000 & 0.0076 & 0.0055\\
2,000 & 0.0038 & 0.0028\\
4,000 & 0.0017 & 0.0014\\
8,000 & 0.0008 & 0.0007\\
\bottomrule
\end{tabular}
\end{table}

\refstepcounter{experiment}\label{exp:E3}\subsubsection{Experiment~\theexperiment: causal admissibility under an intervention}
A one-state system contains the structural state $Z$ and a proxy $P=Z+0.05e$ calibrated so that both have screening defect $0.036$ before intervention, indistinguishable at $N=4{,}000$. The intervention $\mathrm{do}(P:=\tilde P)$ replaces the proxy by an independent variable with the same marginal and leaves the mechanism from $Z$ to returns unchanged. The interventional admissibility condition is implemented as invariance of the conditional regression across the two regimes (Definition~\ref{def:claim}); under the null of two observational regimes the same test admits both $Z$ and $P$ in every one of 300 replications, so its size at the declared tolerances is below $1/300$, and the rejections of $P$ reported below are not false rejections. Portfolios are the unconstrained conditional mean--variance weights normalized to full investment, a deliberate choice that isolates the effect of the information from that of constraints. Table~\ref{tab:e3} and Figure~\ref{fig:e3} report 200 replications. Under a conditional-law contract Stage~I retains the proxy in $47\%$ of replications, the two candidates being tied; the interventional admissibility condition admits $Z$ in every replication and $P$ in none. After the intervention the proxy's screening defect rises to $0.563$ and its response error to $1.25$, while both quantities for $Z$ are unchanged. The portfolio fixed on proxy-conditioned moments incurs a post-intervention loss of $0.157$ against $0.136$ for the state-conditioned portfolio; in the loss units of~\eqref{eq:markowitz} with $\gamma=5$ this is a certainty-equivalent return gap of $0.021$, about one fifth of the average conditional mean return in the design. The experiment measures the statement of Proposition~\ref{prop:diversification}(c), that the split of risk survives the intervention only for the interventionally admissible representation.

\begin{table}[!htb]\centering\small
\caption{Experiment~\ref{exp:E3}: structural state versus decision-equivalent proxy before and after an intervention severing the proxy channel (means, standard deviations in parentheses, 200 replications).}
\label{tab:e3}
\begin{tabular}{@{}lcc@{}}
\toprule
 & structural state $Z$ & proxy $P$\\
\midrule
screening defect $\hat\varepsilon_t$, before intervention & 0.036 (0.007) & 0.037 (0.007)\\
screening defect $\hat\varepsilon_t$, after intervention & 0.036 (0.007) & 0.563 (0.011)\\
response error $\|\widehat B-B\|_2$, after intervention & 0.035 (0.009) & 1.248 (0.022)\\
retained by conditional-law Stage~I (fraction) & 0.530 & 0.470\\
admitted by interventional admissibility condition (fraction) & 1.000 & 0.000\\
post-intervention loss of the fixed portfolio & 0.136 (0.010) & 0.157 (0.024)\\
\bottomrule
\end{tabular}
\end{table}

\begin{figure}[!htb]
\centering
\begin{tikzpicture}
\begin{groupplot}[group style={group size=2 by 1,horizontal sep=1.9cm},width=7.0cm,height=4.9cm,ybar,/pgf/bar width=9pt,tick label style={font=\footnotesize},label style={font=\small},legend style={font=\footnotesize},symbolic x coords={state $Z$,proxy $P$},xtick=data,enlarge x limits=0.5,ymin=0]
\nextgroupplot[ylabel={screening defect $\hat\varepsilon_t$},legend style={legend columns=2}]
\addplot[fill=gray!40] coordinates {(state $Z$,0.036) (proxy $P$,0.037)}; \addlegendentry{before intervention}
\addplot[fill=red!60!black] coordinates {(state $Z$,0.036) (proxy $P$,0.563)}; \addlegendentry{after intervention}
\nextgroupplot[ylabel={after intervention},legend style={legend columns=2}]
\addplot[fill=gray!40] coordinates {(state $Z$,0.035) (proxy $P$,1.248)}; \addlegendentry{response error $\|\widehat B-B\|_2$}
\addplot[fill=blue!50!black] coordinates {(state $Z$,0.136) (proxy $P$,0.157)}; \addlegendentry{loss of fixed portfolio}
\end{groupplot}
\end{tikzpicture}
\caption{E3. Left: screening defect of the structural state and of the decision-equivalent proxy before and after the intervention severing the proxy channel; indistinguishable before, separated after. Right: post-intervention response error and loss of the portfolio fixed on each representation's pre-intervention moments. Means over 200 replications.}
\label{fig:e3}
\end{figure}
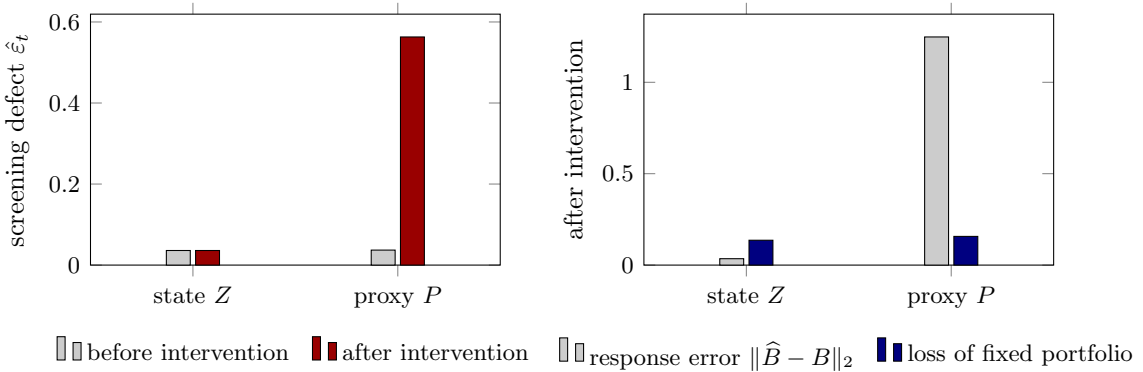

\refstepcounter{experiment}\label{exp:E4}\subsubsection{Experiment~\theexperiment: reselection across decision nodes}
Forty decision nodes are simulated, each with a fresh sample of $200$ observations of which $100$ certify; at node 20 the second common state changes from $Z_2$ to an independent $Z_3$, so the population-optimal class changes once. Candidates are the subsets of size at most two of $\{Z_1,Z_2,Z_3,P,\nu\}$, scored by $\hat\varepsilon_t+0.02\,\kappa_t$, a real representation of the admissibility order at the tolerance scale. The hysteresis rule of Section~\ref{sec:guarantees} is applied with margins $\eta\in\{0,0.02,0.05,0.10,0.20\}$; $\eta=0$ is reselection by argmin at every node. Figure~\ref{fig:e4} reports, over 100 replications, the mean number of reselections, the fraction of nodes at which the incumbent is the population-optimal class, and the mean lag in nodes before the new optimal class is adopted after the change. Without hysteresis the incumbent changes 5.7 times on average, all but one driven by estimation noise; a margin of $0.10$ reduces this to 1.16 with the incumbent correct at 99.2\% of nodes and a lag of 0.12 nodes; a margin of $0.20$ suppresses reselection almost entirely but delays adoption of the new class by 0.62 nodes on average. The pattern is that of Proposition~\ref{prop:hysteresis}. Reselections not driven by a population change vanish once $\eta$ exceeds twice the estimation error, at the cost of a lag when the population does change.

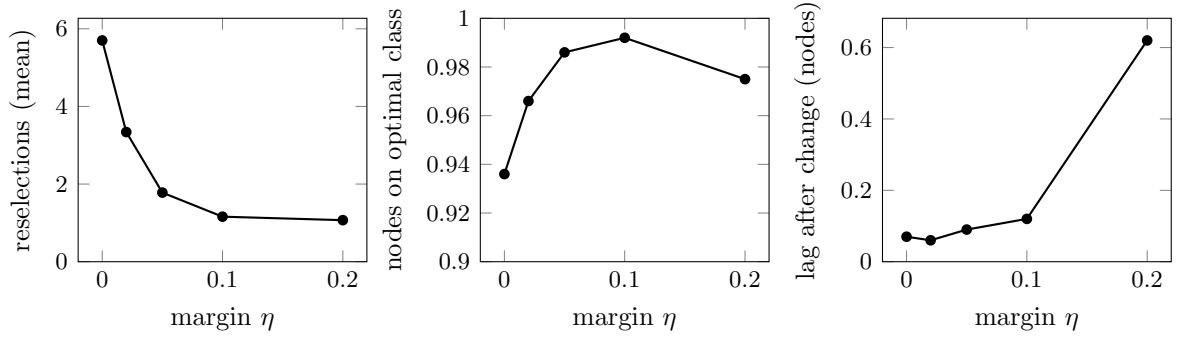
\begin{figure}[!htb]
\centering
\begin{tikzpicture}
\begin{groupplot}[group style={group size=3 by 1,horizontal sep=1.5cm},width=5.4cm,height=4.8cm,tick label style={font=\footnotesize},label style={font=\small},xlabel={margin $\eta$},xtick={0,0.1,0.2},xticklabels={$0$,$0.1$,$0.2$}]
\nextgroupplot[ylabel={reselections (mean)},ymin=0]
\addplot[thick,mark=*,mark size=1.6pt] coordinates {(0.0,5.70) (0.02,3.34) (0.05,1.78) (0.1,1.16) (0.2,1.07)};
\nextgroupplot[ylabel={nodes on optimal class},ymin=0.9,ymax=1.0]
\addplot[thick,mark=*,mark size=1.6pt] coordinates {(0.0,0.936) (0.02,0.966) (0.05,0.986) (0.1,0.992) (0.2,0.975)};
\nextgroupplot[ylabel={lag after change (nodes)},ymin=0]
\addplot[thick,mark=*,mark size=1.6pt] coordinates {(0.0,0.07) (0.02,0.06) (0.05,0.09) (0.1,0.12) (0.2,0.62)};
\end{groupplot}
\end{tikzpicture}
\caption{E4. Reselection of the Stage-I class across forty decision nodes with one population change, under the hysteresis rule with margin $\eta$ (100 replications; $\eta=0$ is argmin reselection). Left: mean number of reselections. Centre: fraction of nodes at which the incumbent is the population-optimal class. Right: mean lag before the new optimal class is adopted after the change.}
\label{fig:e4}
\end{figure}

\subsection{Attainability of the separator condition on market data}\label{sec:market}

Two outputs of the procedure must be kept apart in everything that follows. When the admissible class is non-empty, the framework's output is a portfolio of the admissible class. When it is empty, the framework's output is the statement that the declared universe does not support the claim, and no admissible portfolio exists; what the paper then reports, in order to measure what the failure costs, is built on the \emph{least-defect representation} $H^{\flat}_t$, the minimizer of the admissibility order with the tolerance removed, and is labelled accordingly. A least-defect portfolio is a diagnostic object, not an admissible one, and the labels ``least-defect'' in the tables and figures below mark exactly the windows in which the condition was not met.

The primary panel has the size of a managed book. It consists of all $147$ S\&P~500 constituents with complete daily prices from 2010 to 2024, $3,312$ aligned sessions, against a candidate universe of 180 series, the 127 economic and financial series used throughout the paper together with the 48 industry portfolios and the five factors of Kenneth French's library \cite{French2024}, which are public and point-in-time and therefore admissible, with every subset eligible and the search by greedy forward selection up to twelve. The 127-series universe alone, which is the right candidate set for the small panels, is reported for this panel alongside the primary results below and in Section~\ref{sec:market-cost}, as the case of a candidate universe that does not match the asset class. Windows stepped by 63 sessions have a certification block of 378 sessions, an estimation block of 189 or 378 sessions, and an evaluation block of 189, giving $41$ and $38$ windows; portfolios are long-only minimum-variance, the rule a book of this size would run. Two further panels are used where they answer a question the primary one cannot: a cross-asset universe of fifty equities and fifty Spanish bonds over 2014--2018, and a twelve-equity panel over 2001--2023, one name per sector, which is small enough for the subset structure of the admissibility order to be visible and long enough for a twenty-year strategy path. Series that are strictly positive enter as log differences, rates as first differences.

The following properties of the design hold for every panel in this section, and the first two state what is and is not computed at each rebalance. First, the field experiments run the conditional-law contract: Stage~I certifies chronology and publicity from the release calendar and the separator condition by forward selection on the maximal residual correlation, the causal-response defect is identically zero under this contract, and no causal discovery algorithm runs in the loop. The interventional contract is run separately in two forms, with calendar regimes as environments later in this subsection and, inside the rebalance loop, with the three consecutive thirds of the certification block as environments in the sense of invariant causal prediction \cite{Peters2016}, a driver being retained when its loadings on every asset are stable across the three within pooled standard errors given the drivers already selected. Second, Stage~II uses $\widehat B\widehat V\widehat B^{\Tr}+\operatorname{diag}(\widehat\Omega)$ with least-squares loadings, the represented-response term of the tower decomposition, which coincides with the factor-model construction to machine precision; the nonlinear response the decomposition allows is not used on market data. Third, the certification regresses returns on drivers of the same session, so the representation it selects is a contemporaneous decomposition of risk, the object a factor risk model estimates, and not a predictive signal. With the drivers lagged by one session the reduction of the defect falls, on the primary panel, from $8\%$ to $1\%$, and on the twelve-equity panel from $50\%$ to $4\%$, and it does not recover at two or five sessions. This is what the framework's two components require, the horizon-closed component certifies the separator and only estimated moments enter the portfolio, and it locates the contribution, which is not a better factor model but the test of whether a factor set separates, the exclusion of look-ahead and non-public information, and an order under which portfolio loss cannot select the information. Fourth, forward selection is a heuristic for the minimum over subsets; where that minimum can be enumerated it is attained in every trial (Appendix~\ref{app:config}), so the defects reported are tight upper bounds. Fifth, the primary panel is the set of names with complete coverage over 2010--2024 and therefore carries survivorship; the statements below concern dependence and covariance, on which the bias is of second order, and the sorted-portfolio universes do not share it.

On the primary panel the separator condition is not met, and at this scale the screening statistic is not the quantity that shows it. The maximal absolute residual correlation over the $10,731$ pairs is $0.90$ before conditioning and $0.71$ after twelve drivers, against a pair-calibrated tolerance of $0.34$, and no window is admissible: with ten thousand pairs some two names are always nearly collinear, two share classes, two refiners, two regional banks, and a maximum sits near one whatever the representation. The statistic certifies the separator property of the matrix, which a book of this size does not have; the question that matters at this scale is the one the consequence statistic answers in Section~\ref{sec:market-cost}. The selection itself is economically legible. Of the twelve drivers chosen, a median of $6$ are industry portfolios, the value factor appears in most windows, and the volatility index remains; with the 127 macro series alone the selection falls back on the volatility index, equity indices and currencies and lowers the maximal residual correlation by $8\%$ against $21\%$ here, the reduction measured on the certification block relative to its value before any conditioning.

Whether the condition can be met at all depends on the asset universe, and in a way the theory anticipates. Table~\ref{tab:grid} repeats the exercise on five universes against the same 127 candidates, with the tolerance calibrated to the number of asset pairs by a \v{S}id\'ak correction so that the defect is comparable across universes of different size. On the twelve single names the selected representation removes half of the maximal residual correlation and clears the tolerance in $17\%$ of windows, the first universe on which the admissible class is non-empty on market data. On the 48 industry portfolios the reduction falls to $18\%$ and on the 25 size and value portfolios to $5\%$, with no admissible window in either; the European and Developed size and momentum portfolios lie in between. The ordering follows the number of names inside each asset rather than the number of assets, and the design below shows it directly. A sorted portfolio is an average of hundreds of securities whose idiosyncratic components have already cancelled, so what remains is common variation, and a representation cannot screen off a dependence that is the common factor itself. Proposition~\ref{prop:diversification} describes this from the other side, and the empirical consequence is that the separator condition is a property of the asset universe as much as of the candidate set. It is attainable on individual securities and, at the tolerances a practitioner would declare, not attainable on pre-diversified portfolios.

\begin{table}[!htb]\centering\small
\caption{Statistical admissibility condition across asset universes, 127 candidate drivers, daily 2001--2023, rolling 756-session windows. The tolerance is $0.10$ plus a \v{S}id\'ak-calibrated bound for the maximum over the $n(n-1)/2$ asset pairs. Defects are medians over windows of the maximal absolute residual correlation before conditioning and after greedy forward selection of at most eight drivers.}
\label{tab:grid}
\begin{tabular}{@{}lrrrrrrr@{}}
\toprule
universe & $n$ & windows & tolerance & before & after & reduction & admissible\\
\midrule
12 US equities & 12 & 82 & 0.27 & 0.67 & 0.34 & 49\% & 17\% \\
48 industries & 48 & 78 & 0.31 & 0.92 & 0.76 & 18\% & 0\% \\
25 size/value (US) & 25 & 78 & 0.29 & 0.98 & 0.93 & 5\% & 0\% \\
25 size/momentum (Europe) & 25 & 82 & 0.29 & 0.96 & 0.81 & 16\% & 0\% \\
25 size/momentum (Developed) & 25 & 82 & 0.29 & 0.97 & 0.90 & 7\% & 0\% \\
\bottomrule
\end{tabular}
\end{table}

The ordering across universes could reflect anything that distinguishes single names from sorted portfolios. A design that holds everything else fixed isolates the mechanism. From 147 U.S.\ equities with complete coverage between 2010 and 2024 we form universes of twelve assets in which each asset is an equally weighted basket of $k$ names, for $k=1,2,4,8,12$: the same names and the same 127 candidates, aggregated differently, in six random groupings per $k$ and the same 756-session windows stepped by 189 sessions, so the only quantity that varies is how much idiosyncratic risk has already been diversified away inside each asset. The defect before conditioning rises monotonically with $k$, from $0.67$ to $0.88$, and the defect attainable after selecting eight drivers rises with it beyond the smallest baskets, from $0.39$ at two names per asset to $0.43$, $0.55$ and $0.64$ at twelve, while the fraction of certified windows falls from $10\%$ and $8\%$ to zero (Figure~\ref{fig:ea}). 

Aggregation is therefore the operative variable, not the identity of the assets, and Proposition~\ref{prop:aggregation} gives the mechanism: a common residual component that the representation does not explain survives averaging while the idiosyncratic part does not, so the basket correlation rises with the number of names and the condition can hold only below a basket size set by the ratio of idiosyncratic to common residual variance. Each doubling of the number of names inside an asset removes idiosyncratic variation that the drivers were not needed to explain and leaves common variation that they cannot explain. The single-name design is the one exception to monotonicity after conditioning, and the reason is statistical rather than structural, since the maximum over pairs of an estimated correlation matrix is biased upward and the bias is largest for the noisiest residuals, which are those of individual names.

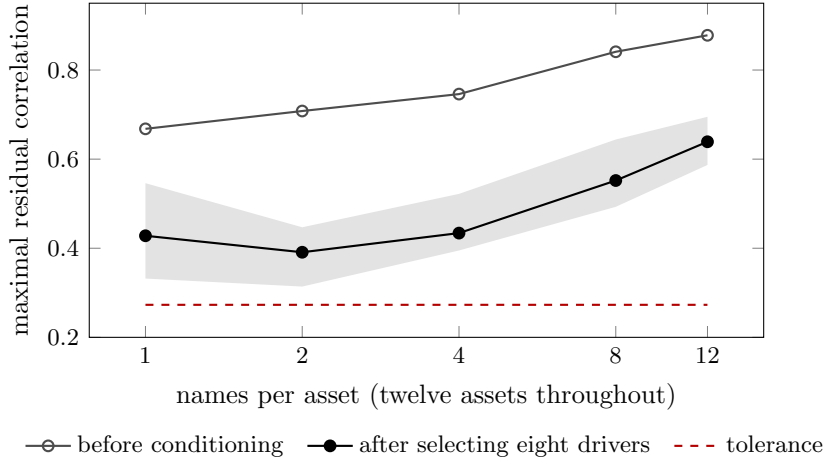
\begin{figure}[!htb]
\centering
\begin{tikzpicture}
\begin{axis}[width=10.5cm,height=6cm,xmode=log,log basis x=2,xtick={1,2,4,8,12},xticklabels={1,2,4,8,12},xlabel={names per asset (twelve assets throughout)},ylabel={maximal residual correlation},ymin=0.2,ymax=0.95,tick label style={font=\footnotesize},label style={font=\small},legend style={font=\footnotesize,legend columns=3},legend cell align=left]
\addplot[name path=hi,draw=none,forget plot] coordinates {(1,0.546) (2,0.447) (4,0.522) (8,0.644) (12,0.695)};
\addplot[name path=lo,draw=none,forget plot] coordinates {(1,0.332) (2,0.314) (4,0.395) (8,0.493) (12,0.587)};
\addplot[gray!22,forget plot] fill between[of=hi and lo];
\addplot[thick,mark=o,mark size=2pt,gray!60!black,forget plot] coordinates {(1,0.668) (2,0.708) (4,0.746) (8,0.841) (12,0.878)};
\addplot[thick,mark=*,mark size=2pt,black,forget plot] coordinates {(1,0.428) (2,0.391) (4,0.434) (8,0.552) (12,0.639)};
\addplot[thick,red!70!black,dashed,domain=1:12,forget plot] {0.273};
\addlegendimage{thick,mark=o,gray!60!black}\addlegendentry{before conditioning}
\addlegendimage{thick,mark=*,black}\addlegendentry{after selecting eight drivers}
\addlegendimage{thick,red!70!black,dashed}\addlegendentry{tolerance}
\end{axis}
\end{tikzpicture}
\caption{Prior diversification and the attainable defect, with the asset universe held fixed. Twelve assets, each an equally weighted basket of $k$ names drawn from the same 147 equities, six draws per design, median over windows with the interquartile band of the conditioned defect. Both the dependence to be removed and the dependence that survives conditioning increase with the number of names inside each asset.}
\label{fig:ea}
\end{figure}

Whether a larger candidate universe would close the gap, whether the greedy search attains the admissibility optimum, how the optimum depends on the number of assets, and how a cross-asset universe of equities and bonds behaves are reported in Section~\ref{sec:more}. The failure of the separator condition is therefore of two different kinds. For individual securities it is a shortage of candidates, and the framework says how many more would be needed. For pre-diversified portfolios it is structural within the declared family of candidates, since the residual dependence is the common factor that the sorting has already isolated, and no set of observable drivers of this type removes it.

\subsection{Robustness of attainability}\label{sec:more}

This subsection reports the results that bear on the attainability of the separator condition and were not needed to state the answer of Section~\ref{sec:market}: the optimality of the selection and its dependence on the number of assets, the attainable defect against the size of the candidate search, the cross-asset design, and the large-universe comparison with loss-based representation selection on the 127-series candidate universe.

\subsubsection{Optimality of the selection and dependence on the number of assets}
Two questions about the selection itself remain: whether the set it returns is the admissibility optimum, and how that optimum depends on the number of assets. The first cannot be settled by enumeration, since there are of the order of $10^{12}$ subsets of eight drivers among 127, but it can be settled to the precision that matters. On the equity panel a beam search of width four, backward elimination from the forty best single drivers and a swap local search with random restarts reach medians of $0.328$, $0.355$ and $0.332$ against $0.336$ for forward selection: the best of them improves on it by $0.006$ at the median and by at most $0.013$, and shares five of every eight drivers with it. The reason is the shape of the objective along the selection path, $0.65$ before conditioning, $0.45$ after one driver, $0.39$ after two, $0.36$ after three and $0.34$ after eight: the first driver does two thirds of the work, the second and third most of the rest, and beyond three the gains are within the sampling error of the statistic, so many sets are near-optimal and differ only in which series fill the last places. The admissibility order is designed for this situation. At a declared tolerance the order selects the set of least dimension that meets it, and what the independent selections return is stable in its first coordinates, the volatility index in every window followed by an equity index and a government yield, while the identity of the marginal drivers, which the data cannot determine, does not enter the decision. The second question has a monotone answer. On universes of $6$, $12$, $24$ and $48$ random S\&P~500 names the reduction attainable with eight drivers is $49\%$, $40\%$, $30\%$ and $20\%$ and the certified fraction $30\%$, $5\%$, $2\%$ and $0\%$; the selection is run afresh in every window and every universe, and the fact that it returns the same leading drivers is a finding about the data, not a constraint of the design. A maximum over $n(n-1)/2$ pairs is a harder statistic as $n$ grows even after calibration, and a fixed number of drivers explains a smaller share of a larger cross-section; for a large portfolio the framework therefore certifies sub-universes, or relies on the consequence statistic of Proposition~\ref{prop:diversification}(d), which is defined at the portfolio's weights and does not grow with the number of pairs.

\subsubsection{Size of the candidate search}
Whether a larger candidate universe would close the gap has a different answer in each case, and the difference is measurable. Repeating the certification with candidate sets of $4$, $8$, $16$, $32$, $64$ and $127$ drivers drawn at random from the declared universe traces the attainable defect against the size of the search (Figure~\ref{fig:eb}). On the twelve equities the defect falls steadily, from $0.56$ with four candidates to $0.34$ with all 127, at about $0.07$ per doubling of the search; extrapolating the log-linear trend, of the order of three to four hundred candidates of the same kind would reach the tolerance. On the size and value portfolios the same exercise moves the defect from $0.97$ to $0.93$, about $0.010$ per doubling, and the same extrapolation puts the required search beyond any conceivable universe. A third universe, twelve bonds drawn from 1{,}336 Spanish government and corporate issues with daily prices from 2014 to 2018, gives a reduction of $17\%$ and no admissible window against the same drivers, which are macroeconomic and equity series rather than the sovereign curve that prices those bonds. The candidate universe must match the asset class, and a declared universe that does not is reported as such rather than repaired after the fact. 
\subsubsection{A cross-asset universe}
A cross-asset design completes the picture. From 147 S\&P~500 constituents with complete daily prices and 1{,}336 Spanish government and corporate bonds with daily prices from 2014 to 2018, we form universes of twelve individual equities, twelve individual bonds, and six of each, against the same 127 candidates, over the period the two sources share. Two features of the bond data determine the design. Bond prices carry stale marks (five to fifteen per cent of zero daily returns) and extreme kurtosis, so dependence is measured here by rank correlation, which leaves the equity results unchanged; and the bond file is dated one session later than the market data, a convention visible in the cross-correlation of bond returns with the euro curve, which peaks at a lag of one session for every curve variable, including the U.S.\ one, vanishes at two, and is absent between bonds themselves, so the bond dates are shifted back one session before alignment. Table~\ref{tab:cross} reports the attainable defect with same-session drivers and with drivers lagged one session. All three universes behave alike. Same-session drivers remove a substantial part of the dependence, $44\%$ on equities, $22\%$ on bonds, where the selection picks the German five-year bund and the euro five-year swap, and $38\%$ on the mixed universe, and little survives a one-session lag. Only the equity universe is certified, in $42\%$ of windows; on bonds and on the mixed universe the maximum over pairs is taken by bond pairs that share a sovereign curve the candidate set does not contain, and no window is admissible.

\begin{figure}[!htb]
\centering
\begin{tikzpicture}
\begin{axis}[width=10.5cm,height=6.2cm,xmode=log,log basis x=2,xlabel={number of candidate drivers in the search},ylabel={attainable defect},xtick={4,8,16,32,64,128},xticklabels={4,8,16,32,64,127},ymin=0.2,ymax=1.02,tick label style={font=\footnotesize},label style={font=\small},legend style={font=\footnotesize,legend columns=3},legend cell align=left]
\addplot[name path=eqhi,draw=none] coordinates {(4,0.608) (8,0.579) (16,0.548) (32,0.490) (64,0.447) (127,0.373)};
\addplot[name path=eqlo,draw=none] coordinates {(4,0.493) (8,0.475) (16,0.426) (32,0.356) (64,0.329) (127,0.285)};
\addplot[gray!25,forget plot] fill between[of=eqhi and eqlo];
\addplot[thick,mark=*,mark size=1.8pt,black,forget plot] coordinates {(4,0.559) (8,0.536) (16,0.503) (32,0.426) (64,0.386) (127,0.336)}; 
\addplot[name path=svhi,draw=none] coordinates {(4,0.974) (8,0.971) (16,0.966) (32,0.964) (64,0.953) (127,0.948)};
\addplot[name path=svlo,draw=none] coordinates {(4,0.959) (8,0.955) (16,0.948) (32,0.944) (64,0.929) (127,0.917)};
\addplot[blue!12,forget plot] fill between[of=svhi and svlo];
\addplot[thick,mark=square*,mark size=1.6pt,blue!60!black,forget plot] coordinates {(4,0.969) (8,0.966) (16,0.958) (32,0.954) (64,0.945) (127,0.931)}; 
\addplot[thick,red!70!black,dashed,domain=4:128,forget plot] {0.273};
\addlegendimage{thick,mark=*,black}\addlegendentry{12 US equities}
\addlegendimage{thick,mark=square*,blue!60!black}\addlegendentry{25 size/value portfolios}
\addlegendimage{thick,red!70!black,dashed}\addlegendentry{tolerance}
\end{axis}
\end{tikzpicture}
\caption{Attainable defect against the size of the candidate search, median over windows and over eight random draws of the candidate set at each size, with the interquartile band. On individual securities the defect declines steadily with the size of the search and the extrapolated trend reaches the tolerance at a few hundred candidates; on pre-diversified portfolios it is almost flat, and no candidate universe of a realistic size would certify them.}
\label{fig:eb}
\end{figure}
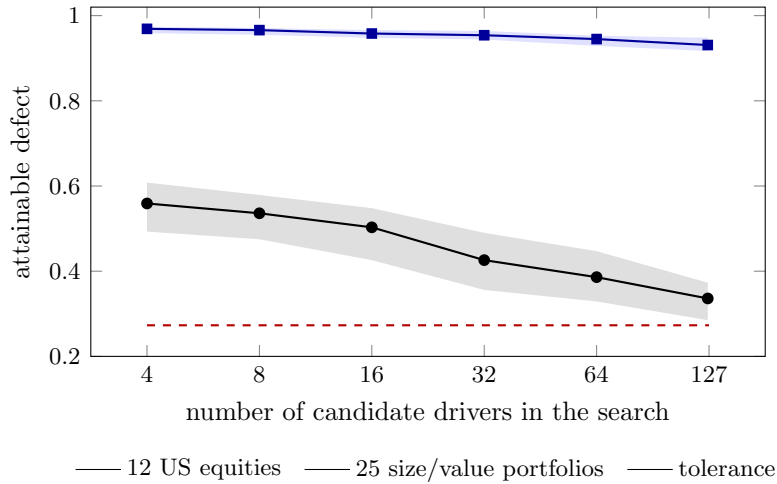

\begin{table}[!htb]\centering\small
\caption{Cross-asset design, rank-correlation defect, 2014--2018, 504-session windows, bond dates realigned, tolerance calibrated to the number of pairs. Attainable defect after forward selection of at most eight drivers, with same-session drivers and with drivers lagged one session.}
\label{tab:cross}
\begin{tabular}{@{}lrrrrrr@{}}
\toprule
 & \multicolumn{4}{c}{same session} & \multicolumn{2}{c}{lagged one session}\\
\cmidrule(lr){2-5}\cmidrule(lr){6-7}
universe & before & after & reduction & admissible & after & reduction\\
\midrule
Individual equities (S\&P names) & 0.64 & 0.36 & 44\% & 42\% & 0.62 & 9\% \\
Individual bonds (Spanish issues) & 0.92 & 0.72 & 22\% & 0\% & 0.85 & 9\% \\
Cross-asset, six equities and six bonds & 0.90 & 0.55 & 38\% & 0\% & 0.74 & 14\% \\
\bottomrule
\end{tabular}
\end{table}

\subsection{The cost of failed separation}\label{sec:market-cost}

The condition fails on every panel of realistic size; this subsection measures what the failure costs the portfolio, first through the risk the model reports and then through the price of demanding invariance across regimes. The consequence statistic of Proposition~\ref{prop:diversification}(d) is evaluated at the long-only minimum-variance weights the model itself produces, on the certification block, before any portfolio is held.

On the primary panel with the 180-series candidate universe the diagonal-residual restriction understates the variance of its own portfolio by $14\%$ at both estimation lengths: the variance implied by the full residual is $1.14$ times the variance the diagonal model predicts. With the 127 macro series alone the factor is $2.39$ at the median, above one in every window, with an interquartile range of $1.71$ to $3.83$; in volatility terms the model predicts about two thirds of the risk it carries. The difference is the candidate universe. Eight macro drivers leave the industry and style structure of a cross-section of 147 names in the residual, and a model that treats that remainder as idiosyncratic understates the risk by more than half; industry and style candidates, which the framework admits as readily as macro series, remove most of it. The evaluation block confirms both readings. Figure~\ref{fig:field} shows the defect and the bias statistics window by window for the macro candidate universe, where the restriction is largest. The Barra bias statistic, the ratio of realized to predicted volatility of the held portfolio, is $1.50$ with the full residual and $1.76$ with the diagonal restriction at 189 estimation sessions and $1.22$ and $1.34$ at 378, against $1.39$ and $2.03$ with the macro universe; the excess over one with the full residual is the non-stationarity every estimator faces, and the gap between the two ratios is the restriction, which the consequence statistic announced from the certification block. The two random halves of the universe and the cross-asset universe, reported in Section~\ref{sec:more}, repeat the pattern.

\begin{figure}[!htb]
\centering
\begin{tikzpicture}
\begin{groupplot}[group style={group size=2 by 1,horizontal sep=1.8cm,vertical sep=2.4cm},height=5.6cm,xmin=2010.3,xmax=2021.5,xtick={2011,2013,2015,2017,2019,2021},xticklabel style={/pgf/number format/1000 sep=}]
\nextgroupplot[width=7.6cm,ylabel={maximal residual correlation},ymin=0.6,ymax=1.0]
\addplot[only marks,mark=o,mark size=1.1pt,gray] coordinates {(2010.33,0.888) (2010.58,0.873) (2010.83,0.889) (2011.08,0.896) (2011.33,0.901) (2011.58,0.887) (2011.75,0.841) (2012.08,0.931) (2012.33,0.929) (2012.58,0.936) (2012.83,0.931) (2013.08,0.928) (2013.33,0.934) (2013.58,0.817) (2013.83,0.844) (2014.08,0.843) (2014.33,0.880) (2014.58,0.892) (2014.83,0.914) (2015.08,0.920) (2015.33,0.935) (2015.58,0.921) (2015.83,0.911) (2016.08,0.898) (2016.33,0.880) (2016.58,0.877) (2016.83,0.857) (2017.08,0.865) (2017.33,0.871) (2017.58,0.895) (2017.83,0.896) (2018.08,0.899) (2018.33,0.884) (2018.58,0.871) (2018.83,0.931) (2019.08,0.934) (2019.33,0.927) (2019.58,0.927) (2019.83,0.925) (2020.08,0.932) (2020.33,0.911)}; \addlegendentry{before conditioning}
\addplot[only marks,mark=*,mark size=1.1pt,black] coordinates {(2010.33,0.750) (2010.58,0.759) (2010.83,0.795) (2011.08,0.804) (2011.33,0.806) (2011.58,0.790) (2011.75,0.763) (2012.08,0.927) (2012.33,0.924) (2012.58,0.930) (2012.83,0.924) (2013.08,0.920) (2013.33,0.928) (2013.58,0.750) (2013.83,0.785) (2014.08,0.787) (2014.33,0.828) (2014.58,0.794) (2014.83,0.793) (2015.08,0.797) (2015.33,0.814) (2015.58,0.808) (2015.83,0.820) (2016.08,0.819) (2016.33,0.831) (2016.58,0.830) (2016.83,0.810) (2017.08,0.827) (2017.33,0.845) (2017.58,0.875) (2017.83,0.875) (2018.08,0.883) (2018.33,0.867) (2018.58,0.846) (2018.83,0.872) (2019.08,0.863) (2019.33,0.851) (2019.58,0.842) (2019.83,0.844) (2020.08,0.846) (2020.33,0.842)}; \addlegendentry{after selection}
\addplot[thick,red!70!black,dashed,domain=2010.3:2021.5] {0.336};
\nextgroupplot[width=7.6cm,ylabel={realized / predicted volatility},ymin=0.5,ymax=4.5]
\addplot[only marks,mark=*,mark size=1.1pt,black] coordinates {(2010.33,1.153) (2010.58,1.842) (2010.83,1.815) (2011.08,2.025) (2011.33,1.254) (2011.58,1.395) (2011.75,2.238) (2012.08,2.776) (2012.33,2.891) (2012.58,3.196) (2012.83,2.903) (2013.08,3.375) (2013.33,3.175) (2013.58,1.256) (2013.83,1.119) (2014.08,1.108) (2014.33,1.215) (2014.58,1.052) (2014.83,1.051) (2015.08,1.671) (2015.33,3.378) (2015.58,3.975) (2015.83,2.705) (2016.08,2.138) (2016.33,2.323) (2016.58,2.979) (2016.83,1.865) (2017.08,1.623) (2017.33,10.457) (2017.58,6.458) (2017.83,6.591) (2018.08,0.890) (2018.33,0.885) (2018.58,0.677) (2018.83,1.631) (2019.08,1.820) (2019.33,3.218) (2019.58,3.035) (2019.83,3.064) (2020.08,2.259) (2020.33,1.762)}; \addlegendentry{diagonal residual}
\addplot[only marks,mark=o,mark size=1.1pt,gray] coordinates {(2010.33,0.971) (2010.58,1.421) (2010.83,1.388) (2011.08,1.602) (2011.33,1.006) (2011.58,1.211) (2011.75,1.903) (2012.08,1.130) (2012.33,1.402) (2012.58,1.408) (2012.83,1.552) (2013.08,1.473) (2013.33,1.445) (2013.58,1.009) (2013.83,0.885) (2014.08,0.861) (2014.33,0.871) (2014.58,0.815) (2014.83,0.723) (2015.08,1.109) (2015.33,1.726) (2015.58,2.029) (2015.83,1.752) (2016.08,1.402) (2016.33,1.587) (2016.58,1.629) (2016.83,0.817) (2017.08,0.854) (2017.33,3.279) (2017.58,4.274) (2017.83,4.257) (2018.08,0.651) (2018.33,0.573) (2018.58,0.514) (2018.83,0.805) (2019.08,0.947) (2019.33,1.416) (2019.58,1.566) (2019.83,1.737) (2020.08,1.294) (2020.33,1.028)}; \addlegendentry{full residual}
\addplot[thick,red!70!black,dashed,domain=2010.3:2021.5] {1};
\end{groupplot}
\end{tikzpicture}
\caption{The primary panel window by window, 147 S\&P~500 names against the 127 macroeconomic and market series, one point per window at its start date. Left: maximal absolute residual correlation before conditioning and after selecting eight drivers, with the pair-calibrated tolerance dashed; the selection removes dependence in every window and reaches the tolerance in none. Right: the bias statistic of the held long-only minimum-variance portfolio, realized over predicted volatility, with the diagonal residual and with the full residual; the gap between the two series is the diagonal-residual restriction, and it is announced before the evaluation block by the consequence statistic of Section~\ref{sec:market-cost}.}
\label{fig:field}
\end{figure}
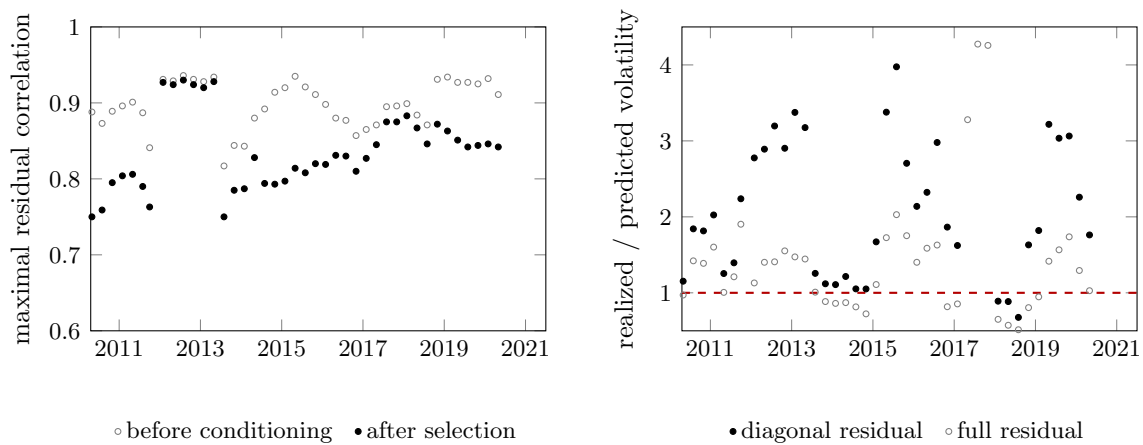

On the twelve-equity panel the same quantities are smaller and better identified, because eight drivers explain more of twelve names than of a hundred and forty-seven. The discarded term is $15.5\%$ of predicted variance, positive in $78$ of $81$ windows, and the bias statistic moves from $1.06$ to $1.14$. There the sign and the mechanism can be read directly through~\eqref{eq:misallocation}. The residual correlations the representation fails to remove are predominantly positive, with a portfolio-weighted average of $0.12$, so the discarded term is positive for any portfolio whose weighted exposures are mostly of one sign; the proposition permits either sign, and the sign observed is a property of equity residuals. The effect is not an artefact of the optimizer, since equal weights and random long-only portfolios give the same sign in every window, and its size tracks the weighted average residual correlation of the decomposition and not the screening defect; this separates the two questions the framework asks, whether the representation separates, the first answered by the maximum and the second by the quadratic form, and the second question has its own certifiable statistic, declared and certified in advance as the defect is.

The interventional condition also has an instance on these data, with calendar regimes as observational environments; they are not interventions, and what the instance measures is the price of a stability restriction. On the primary panel, with 2010--2014, 2015--2019 and 2020--2024 as three regimes, a driver is retained as invariant if the loadings of every name on it, estimated within each regime, lie within three pooled standard errors of the full-sample loading; stability of the loadings is necessary for invariance of the conditional law of returns given the drivers and not sufficient, since intercepts and the residual covariance would have to be invariant as well, so the retained set is an upper bound on the invariant one and the price reported is a lower bound. Only $4$ of the 127 drivers pass, none of them among the series the conditional-law selection uses; restricting the search to them raises the attained defect from $0.83$ to $0.90$ and the out-of-sample variance of the minimum-variance portfolio by $16\%$ at the median, in $81\%$ of windows. With environments defined by the policy instrument rather than the calendar, the Federal Reserve's target-rate regimes over the sample (zero lower bound to December 2015, tightening to December 2018, easing to March 2022, tightening to July 2023), only $2$ drivers pass, neither of them economically related to equity risk, and the price is $17\%$ of out-of-sample variance in $71\%$ of windows; the environments are the instrument's regimes, which is as close to an intervention as observational data allow, and the answer does not change. On the twelve-equity panel over four regimes spanning 2001--2023 the count is $9$ of 127 at three standard errors and the price $17\%$ of variance, and the count rises to 23 at five standard errors and 59 under sign stability alone, so what invariance costs depends on the criterion and the environments the contract declares. That is the price of an invariance restriction, which equals the causal admissibility price of Theorem~\ref{thm:distance} only under the identifying assumption that the declared environments stand for interventions on the drivers and that the stability of the loadings, a necessary condition for invariance of the conditional law, is also sufficient; under weaker assumptions it is a lower bound on what identification would cost. The representations that explain equity risk are not the ones whose relation to it is stable across regimes, and a decision maker who needs stability pays for it in dispersion.

\subsection{Allocation under the standard minimum-variance rule}\label{sec:market-portfolio}

The framework ends where an allocation algorithm begins. Its output is an information class that satisfies the declared conditions, the covariance that class implies, and the two statements a risk model does not otherwise make, which information it is conditional on and what that information fails to remove. Turning that output into realized performance is a separate design problem. An allocation rule built for a least-defect class can use what the certification establishes, the structure of the residual, the consequence statistic as a constraint or a penalty, the stability of the class across rebalances, in ways the standard minimum-variance program does not, and the choice of that rule is the algorithmic step that follows the framework rather than part of it. The paper does not take that step. It does report, in the remainder of this subsection, what the covariance of the selected class, admissible where the condition holds and least-defect where it does not, delivers when it is simply inserted into the standard rule in place of the sample covariance or a shrinkage estimator, as the reference point against which a purpose-built rule would be measured: on the primary panel it is indistinguishable from the best static or dynamic estimator at each estimation length (Table~\ref{tab:hh}), and on the twelve-equity panel it matches linear shrinkage on volatility at a third less turnover and lower leverage, and the perturbation bound of Proposition~\ref{prop:regret2} says why no more should be expected from that insertion alone. An error in the conditional moments moves the minimum-variance weights at first order but their realized loss only at second order, so two covariances that agree on the systematic part and differ in the residual produce allocations whose realized variance is nearly the same even where the weights differ, while they differ at first order in the risk that allocation is reported to carry. The information decision therefore shows, in the experiments that follow, where the theory places it, in the reported risk, in the stability of the allocation, and in what alternative selection rules deliver. The limit is the quadratic setting, in which the information enters only through two conditional moments; \cite{RD2025necessary} reached the same conclusion computationally, and converting a certified representation into realized performance requires formulations in which the drivers enter through the trajectory of the constituents' sensitivities to them, the sensitivity-based portfolios of \cite{RD2023} and the causal PDE--control model of \cite{RD2025cpcm}, a design step downstream of certification.

\label{sec:alloc}

The remainder of this subsection reports what the covariance of the selected class delivers when inserted, without any purpose-built allocation rule, into the standard minimum-variance program, against the estimators practitioners use and on the measures the literature reports; it is the reference point against which a purpose-built rule would be measured, not a claim of the framework. The portfolio consequences are reported under the protocol of the covariance-estimation literature \cite{DeMiguelGarlappiUppal2009,LedoitWolf2017}: out-of-sample volatility, Sharpe ratio, turnover per rebalance and gross leverage of the global minimum variance portfolio, the same under a no-short-sale constraint \cite{JagannathanMa2003}, and paired differences against the conditional estimator tested by a circular block bootstrap of the sequence of per-window differences, in blocks of four consecutive windows, which covers the dependence that overlapping windows induce since adjacent windows share three quarters of their sessions and evaluation blocks a session apart share none beyond four steps. Table~\ref{tab:ec} gives the equity panel. On volatility the conditional covariance and linear shrinkage are within two hundredths of a point at the median, $14.22\%$ against $14.20\%$, while the paired difference in realized variance, window by window, is lower for the conditional estimator in $69\%$ of windows and significant at $p=0.007$; the two statements are compatible because the medians summarize levels and the paired test summarizes differences across the same windows, which have a long tail in favour of the conditional estimator. In the same sense it is and it is lower than the sample covariance, the matched principal-component estimator, nonlinear shrinkage and equal weights, the last by small margins; the differences in realized variance are significant at the one per cent level and none of the differences in Sharpe ratio is significant at any conventional level. Where it separates from every alternative is stability: turnover of $0.35$ per rebalance against $0.44$ for linear shrinkage, $0.54$ for the sample covariance and $1.02$ for the principal-component estimator, and gross leverage of $1.20$ against $1.31$, $1.42$ and $1.76$. A representation of eight drivers selected once per rebalance moves the portfolio less than any estimator that re-reads the whole covariance from returns. At a one-way cost of $10$ basis points per unit of turnover and four rebalances a year, this is $14$ basis points a year for the conditional estimator against $18$ for linear shrinkage, $22$ for the sample covariance and $41$ for the principal-component estimator; the volatility advantage over the sample covariance ($0.33$ points) and the cost advantage are of the same order. Two results deserve their own reading. Equal weights attain the highest Sharpe ratio, the finding of \cite{DeMiguelGarlappiUppal2009}, but the comparison is between objectives rather than between estimators: minimum-variance portfolios do not use expected returns, and decomposing the difference shows that the advantage of equal weights is an annualised mean higher by $4.8$ percentage points, which is not significant at the five per cent level, against a volatility higher by $3.0$ points, which is significant at $10^{-14}$. On the objective the estimators optimize, the conditional covariance wins; on an objective none of them targets, weighting equally wins, and no covariance estimator can be held responsible for that. Under a no-short-sale constraint the differences narrow sharply, from a range of $1.4$ percentage points of volatility to $0.3$, which is the sense in which constraints substitute for covariance modelling \cite{JagannathanMa2003}, but they do not vanish and the ordering is preserved. What the constraint does not do is repair the misallocation: the discarded off-diagonal term is $15.5\%$ of predicted variance for the unconstrained portfolio and $14.6\%$ for the constrained one, statistically indistinguishable, because the constraint raises the breadth of the portfolio in~\eqref{eq:misallocation} at the same time as it lowers the residual share. A practitioner who imposes long-only positions and believes the specification error has been contained is mistaken: the constraint improves the allocation and leaves the reported risk as wrong as it was.

\begin{table}[!htb]\centering\scriptsize\setlength{\tabcolsep}{4pt}
\caption{Out-of-sample performance of the global minimum variance portfolio on the twelve-equity panel, 81 rolling windows, median across windows. Volatility annualised in per cent, turnover per rebalance, gross leverage $\|w\|_1$, and volatility under a no-short-sale constraint. The last two columns give circular block bootstrap $p$-values for the paired difference against the conditional estimator in realized variance and in Sharpe ratio.}
\label{tab:ec}
\begin{tabular}{@{}lrrrrrrr@{}}
\toprule
estimator & vol. & Sharpe & turnover & leverage & vol. long-only & $p$ var. & $p$ Sharpe\\
\midrule
conditional (least-defect class) & 14.22 & 0.52 & 0.35 & 1.20 & 14.08 & -- & -- \\
sample & 14.55 & 0.59 & 0.54 & 1.42 & 14.33 & 0.009 & 0.699 \\
PCA, matched dimension & 15.58 & 0.39 & 1.02 & 1.76 & 14.07 & 0.003 & 0.428 \\
Ledoit--Wolf linear & 14.20 & 0.49 & 0.44 & 1.31 & 14.01 & 0.007 & 0.105 \\
Ledoit--Wolf nonlinear & 14.47 & 0.53 & 0.47 & 1.34 & 14.10 & 0.007 & 0.912 \\
equal weights & 16.55 & 0.73 & 0.00 & 1.00 & 16.55 & 0.004 & 0.321 \\
\bottomrule
\end{tabular}
\end{table}

The evaluation so far compares estimators window by window. A portfolio manager measures a strategy over time, so the same pipeline is run as one: at each quarter end from November 2003 to March 2023, eighty rebalances, the trailing 756 sessions are split into a certification half, on which the admissible class is selected, and an estimation half, on which the covariance is estimated; the global-minimum-variance portfolio under each estimator is held for the following quarter, and the realized returns are chained net of a one-way cost of ten basis points per unit of turnover. Table~\ref{tab:strategy} reports the twenty-year path, long-only and unconstrained, with the fraction of rolling three- and five-year sub-periods in which the conditional strategy has lower realized volatility than each alternative,. Long-only, the conditional strategy has the lowest volatility ($17.27\%$ against $17.39$ to $17.50\%$ for the optimized alternatives), a return and Sharpe ratio within a tenth of a point of the best of them, the lowest turnover and cost drag, and lower volatility than the sample covariance in $78\%$ of three-year and $89\%$ of five-year sub-periods; unconstrained, its volatility is again the lowest while linear shrinkage earns a higher return. The differences are small, a few tenths of a point of volatility, and the paired protocol of Table~\ref{tab:ec} is the test that has power to establish them; the path shows that they persist over two decades rather than that they are large. Equal weights earn the highest return and Sharpe ratio on this panel with by far the highest volatility and a drawdown of two thirds, as in Table~\ref{tab:ec}. One number requires care. The ratio of realized to predicted volatility along the path is $0.98$ for the conditional strategy and $0.87$ to $0.93$ for the others, which reads as if the conditional estimator were the best calibrated; it is the net of two biases of opposite sign, the understatement from the diagonal residual established in Section~\ref{sec:market-cost} and the overstatement that a trailing window carrying a past high-volatility regime produces for every estimator, and the near-cancellation is a property of this sample period, not a merit of the restriction. On the shorter panel of S\&P~500 names, ten years and thirty-nine rebalances, two random draws of twelve names give one path on which the conditional strategy is the best of the optimized alternatives on volatility, drawdown and turnover and one on which it is the worst on volatility while remaining the lowest in turnover; with twelve names and a decade the dispersion across draws exceeds the differences between estimators, which is the reason the paper's claims rest on the paired tests rather than on any single path.

\begin{table}[!htb]\centering\scriptsize\setlength{\tabcolsep}{3.5pt}
\caption{Quarterly-rebalanced global-minimum-variance strategies on the twelve-equity panel, eighty rebalances, November 2003 to March 2023, net of ten basis points per unit of turnover. Annualized return and volatility in per cent, maximum drawdown in per cent, turnover per rebalance, cost drag in basis points a year, median ratio of realized to predicted volatility across quarters, and the fraction of rolling three- and five-year sub-periods in which the conditional strategy has lower realized volatility than the row's strategy.}
\label{tab:strategy}
\begin{tabular}{@{}lrrrrrrrrr@{}}
\toprule
strategy & return & vol. & Sharpe & max DD & turnover & cost & real./pred. & 3y & 5y\\
\midrule
\multicolumn{10}{@{}l}{\emph{long-only}}\\
conditional (least-defect class) & 5.24 & 17.27 & 0.38 & -50.2 & 0.18 & 7 & 0.98 & -- & -- \\
sample & 5.29 & 17.39 & 0.38 & -48.8 & 0.21 & 9 & 0.92 & 78\% & 89\% \\
Ledoit--Wolf linear & 5.39 & 17.42 & 0.39 & -50.3 & 0.20 & 8 & 0.91 & 72\% & 84\% \\
Ledoit--Wolf nonlinear & 5.17 & 17.39 & 0.38 & -47.9 & 0.22 & 9 & 0.87 & 78\% & 92\% \\
PCA, matched dimension & 4.74 & 17.50 & 0.35 & -49.0 & 0.25 & 10 & 0.93 & 62\% & 56\% \\
equal weights & 7.51 & 22.42 & 0.44 & -67.3 & 0.00 & 0 & -- & 100\% & 100\% \\
\midrule
\multicolumn{10}{@{}l}{\emph{unconstrained}}\\
conditional (least-defect class) & 4.83 & 17.16 & 0.36 & -49.5 & 0.24 & 9 & 0.98 & -- & -- \\
sample & 5.47 & 17.43 & 0.39 & -43.2 & 0.34 & 13 & 0.98 & 84\% & 100\% \\
Ledoit--Wolf linear & 5.68 & 17.36 & 0.41 & -45.8 & 0.30 & 12 & 0.95 & 78\% & 92\% \\
Ledoit--Wolf nonlinear & 5.43 & 17.51 & 0.39 & -43.8 & 0.37 & 15 & 0.93 & 91\% & 100\% \\
PCA, matched dimension & 3.24 & 18.45 & 0.27 & -48.7 & 0.68 & 27 & 1.05 & 100\% & 100\% \\
equal weights & 7.51 & 22.42 & 0.44 & -67.3 & 0.00 & 0 & -- & 100\% & 100\% \\
\bottomrule
\end{tabular}
\end{table}

\subsection{Comparison of information selection rules}\label{sec:sidebyside}

The comparisons above are organized by estimator. This subsection puts the three ways of choosing a portfolio's information next to one another on the same windows, so that what the framework adds can be seen rather than inferred. Mean--variance practice takes the information as given and estimates the covariance from returns: the sample covariance and its shrinkage estimators. Loss-based representation selection chooses the information by the decision loss it produces: here, on the certification block, the eight drivers whose conditional covariance gives the minimum-variance portfolio with the smallest in-sample variance, the analogue on these data of the rule the look-ahead corollary concerns. No future variable is in the candidate set and every series is dated at its availability, so this is the overfitting form of the mechanism of Corollary~\ref{prop:lookahead} rather than its literal form: the rule is free to choose, among correctly dated series, the ones whose fitted covariance makes the in-sample portfolio look least risky, and that is what it does. The procedure is the same forward search as the admissibility selection, with the in-sample variance of the long-only minimum-variance portfolio replacing the screening defect as the criterion at each step, the same block, the same dimension limit and the same Stage~II; the two rules therefore differ only in what they rank by. This is a deliberately unregularized instance of selection by decision loss; the tuned implementations of decision-focused learning in the literature add validation, regularization or nested evaluation, and the comparison is with the principle those implementations share, not with their best engineering. The admissible problem chooses the information by the admissibility order and only then estimates. Figure~\ref{fig:sidebyside} plots, window by window on the primary panel, the volatility each method promised for its long-only minimum-variance portfolio, the model's own prediction $\sqrt{w^{\Tr}\widehat\Sigma w}$ computed on the estimation block, against the volatility that portfolio realized on the evaluation block; a framework whose promise is kept sits on the diagonal.

Loss-based representation selection is the clearest result in the paper on market data, and it is a failure of the kind the theory predicts. On the primary panel the drivers it selects share none with the least-defect selection and leave the residual dependence where it was before any conditioning; the portfolio it promises has volatility $2.6\%$ and the portfolio it delivers $12.0\%$, a ratio of $4.3$; on the cross-asset universe $0.6\%$ against $1.4\%$; on the twelve-equity panel $7.6\%$ against $14.4\%$, within ten per cent of its promise in $4\%$ of windows. Choosing information by the loss it minimizes on the data at hand selects the information that fits those data, not the information that explains the assets' common variation, and the promise it makes is the one most likely to be broken; the look-ahead corollary is the limiting case of this mechanism, with future information in place of overfitted information, and the certification block cannot protect a loss-based rule from it because the rule is what the block is overfitted to. The least-defect selection, run on the same block with no portfolio quantity in its criterion, promises on the primary panel a volatility whose ratio to the realized one is $2.02$, against $1.66$ for shrinkage and $1.53$ for the sample covariance, and it is the only one of the three whose shortfall was announced in advance by the consequence statistic.

\begin{table}[!htb]\centering\scriptsize\setlength{\tabcolsep}{3pt}
\caption{Head-to-head on the primary panel, 147 S\&P~500 names, 180 candidate series, long-only minimum-variance portfolios, medians over windows. Realized volatility in per cent, ratio of realized to promised volatility, fraction of windows with lower realized volatility than the best competitor, and circular block bootstrap $p$-value for the paired difference in realized variance against it. Best competitor: Ledoit--Wolf linear shrinkage at 189 estimation sessions, DCC-NL at 378. $k$ is the number of certified drivers in each window.}
\label{tab:hh}
\resizebox{\textwidth}{!}{\begin{tabular}{@{}lrrrrrrrr@{}}
\toprule
 & \multicolumn{4}{c}{estimation block of 189 sessions (41 windows)} & \multicolumn{4}{c}{estimation block of 378 sessions (38 windows)}\\
\cmidrule(lr){2-5}\cmidrule(lr){6-9}
covariance & vol. & real./prom. & beats best & $p$ & vol. & real./prom. & beats best & $p$\\
\midrule
sample covariance & 12.30 & 1.56 & 37\% & 0.178 & 12.26 & 1.26 & 29\% & 0.584 \\
Ledoit--Wolf linear shrinkage & 11.57 & 1.66 & 0\% & -- & 11.94 & 1.51 & 29\% & 0.003 \\
nonlinear shrinkage (analytical) & 12.10 & 1.55 & 44\% & 0.120 & 12.02 & 1.25 & 34\% & 0.217 \\
DCC-NL & 11.89 & 1.55 & 63\% & 0.669 & 11.61 & 1.28 & 0\% & -- \\
AFM-DCC-NL, market factor & 11.91 & 1.56 & 54\% & 0.419 & 11.76 & 1.28 & 34\% & 0.134 \\
$k$ principal components, diagonal residual & 12.30 & 1.62 & 41\% & 0.037 & 12.33 & 1.29 & 21\% & 0.006 \\
$k$ principal components, thresholded residual & 11.92 & 1.58 & 49\% & 0.113 & 12.34 & 1.28 & 37\% & 0.205 \\
least-defect drivers, diagonal residual & 12.48 & 1.76 & 49\% & 0.755 & 11.86 & 1.34 & 24\% & 0.732 \\
least-defect drivers, thresholded residual & 12.35 & 1.61 & 56\% & 0.765 & 11.59 & 1.31 & 39\% & 0.449 \\
least-defect drivers, robust covariance at $\rho=\bar c$ (Prop.~\ref{prop:robust}) & 12.41 & 1.58 & 49\% & 0.612 & 12.12 & 1.25 & 26\% & 0.709 \\
least-defect drivers, DCC-NL residual & 11.96 & 1.59 & 56\% & 0.606 & 11.74 & 1.27 & 37\% & 0.696 \\
\bottomrule
\end{tabular}}
\end{table}

Against mean--variance practice the comparison is reported in Table~\ref{tab:hh}, with the dynamic estimators of the literature included, DCC-NL \cite{EngleLedoitWolf2019} and the approximate factor model with DCC-NL residual \cite{DeNardLedoitWolf2021}, the same converged solver for every estimator, and the same residual treatment for principal components as for certified drivers. Three statements survive the paired tests at both estimation lengths. The least-defect class beats principal components of the same dimension, which are significantly worse than the best competitor ($p=0.037$ and $0.006$) while the least-defect estimators are not, so the choice of information matters relative to statistical factors at matched complexity. No certified estimator beats the best competitor: with 189 estimation sessions the best is linear shrinkage and the least-defect estimators sit within half a point of it; with 378 the best is DCC-NL, as the literature on large dynamic covariance matrices predicts, linear shrinkage is then significantly worse ($p=0.003$), and the least-defect class with thresholded residual ties it, $11.59$ against $11.61$ per cent, $p=0.45$. And the information decision shows where the theory says, in the promise rather than the delivery: the least-defect diagonal model has the worst ratio of realized to promised volatility in the table at the shorter estimation length and the consequence statistic had announced it, while its realized volatility is indistinguishable from the others. The paper's claim is therefore stated at its measured size: certification changes what a risk model can say about itself and does not, inserted into the standard minimum-variance rule, change what the portfolio realizes; the allocation rule that would exploit a least-defect class, and the dynamic covariance such a rule would use, are the subject of a separate study.

\begin{figure}[!htb]
\centering
\begin{tikzpicture}
\begin{groupplot}[group style={group size=2 by 1,horizontal sep=1.8cm},height=6.4cm,tick label style={font=\footnotesize},label style={font=\small},legend style={font=\footnotesize}]
\nextgroupplot[width=7.8cm,xlabel={promised volatility (\%)},ylabel={realized volatility (\%)},xmin=0,xmax=20,ymin=0,ymax=30,,legend cell align=left]
\addplot[only marks,mark=*,mark size=1.3pt,black] coordinates {(7.74,8.77) (5.97,11.04) (6.18,11.21) (5.94,12.02) (7.46,9.36) (7.91,11.06) (7.35,16.47) (3.68,10.22) (3.81,11.03) (3.45,11.04) (4.73,13.71) (4.32,14.58) (4.51,14.35) (9.06,11.38) (9.62,10.75) (9.92,11.03) (7.90,9.62) (6.72,7.08) (5.85,6.14) (4.70,7.86) (3.29,11.11) (2.87,11.40) (4.44,12.02) (6.26,13.38) (6.56,15.25) (5.86,17.47) (5.37,10.03) (6.10,9.91) (3.58,37.45) (6.14,39.67) (5.75,37.89) (17.46,15.44) (16.51,14.53) (19.04,12.80) (6.36,10.37) (6.16,11.20) (4.30,13.85) (5.18,15.73) (5.85,17.92) (7.48,16.90) (8.68,15.36)}; \addlegendentry{least-defect class}
\addplot[only marks,mark=o,mark size=1.3pt,gray] coordinates {(8.92,9.82) (7.28,10.49) (7.69,11.15) (7.14,11.85) (8.96,9.52) (8.84,16.34) (8.34,21.19) (7.75,16.77) (6.90,10.65) (6.93,10.60) (7.89,13.42) (8.81,14.75) (8.55,15.30) (10.71,12.24) (11.33,11.00) (11.87,9.97) (9.93,8.58) (8.19,7.60) (7.49,6.99) (6.28,8.76) (5.31,11.20) (4.60,10.79) (5.83,11.62) (8.16,14.09) (8.46,12.91) (8.63,18.25) (8.90,10.09) (9.14,9.66) (9.08,33.18) (8.01,41.31) (7.72,37.83) (20.94,16.89) (21.86,14.55) (23.21,12.29) (11.40,9.46) (9.75,10.27) (7.94,13.91) (7.98,24.12) (8.59,25.48) (11.19,17.98) (12.97,15.05)}; \addlegendentry{sample covariance}
\addplot[only marks,mark=triangle*,mark size=1.6pt,red!70!black] coordinates {(3.63,8.86) (2.78,11.48) (1.85,11.29) (2.77,9.94) (1.91,11.39) (3.57,13.42) (2.00,16.71) (5.09,10.96) (3.24,8.97) (2.90,7.97) (1.59,11.99) (1.81,12.18) (2.45,14.12) (2.56,10.96) (2.59,37.60) (1.84,41.41) (8.09,12.48) (3.57,10.18) (2.29,15.39) (2.03,21.89) (4.33,17.52)}; \addlegendentry{loss-based selection}
\addplot[thick,black!50,domain=0:20] {x};
\nextgroupplot[width=7.8cm,ybar,bar width=5pt,xtick={0,1,2,3},xticklabels={loss-sel.,admiss.,sample,LW lin.},x tick label style={rotate=30,anchor=east,font=\scriptsize},ylabel={median volatility (\%)},ymin=0,ymax=16,]
\addplot[fill=gray!35,draw=gray!60] coordinates {(0,2.59) (1,5.63) (2,8.00) (3,6.98)}; \addlegendentry{promised}
\addplot[fill=black,draw=black] coordinates {(0,11.99) (1,11.40) (2,12.24) (3,11.57)}; \addlegendentry{realized}
\end{groupplot}
\end{tikzpicture}
\caption{Promise and delivery on the primary panel, 147 S\&P~500 names, long-only minimum-variance portfolios, one point per window. Left: volatility promised by each method against the volatility realized on the evaluation block; the diagonal is a kept promise. Loss-based representation selection, which chooses drivers by the in-sample variance they produce, promises a small fraction of what it delivers; the admissible problem and the sample covariance sit above the diagonal, the admissible problem by more, because the diagonal-residual restriction discards a term that is large at this scale and that the consequence statistic reports in advance. Right: medians of promised and realized volatility by method.}
\label{fig:sidebyside}
\end{figure}
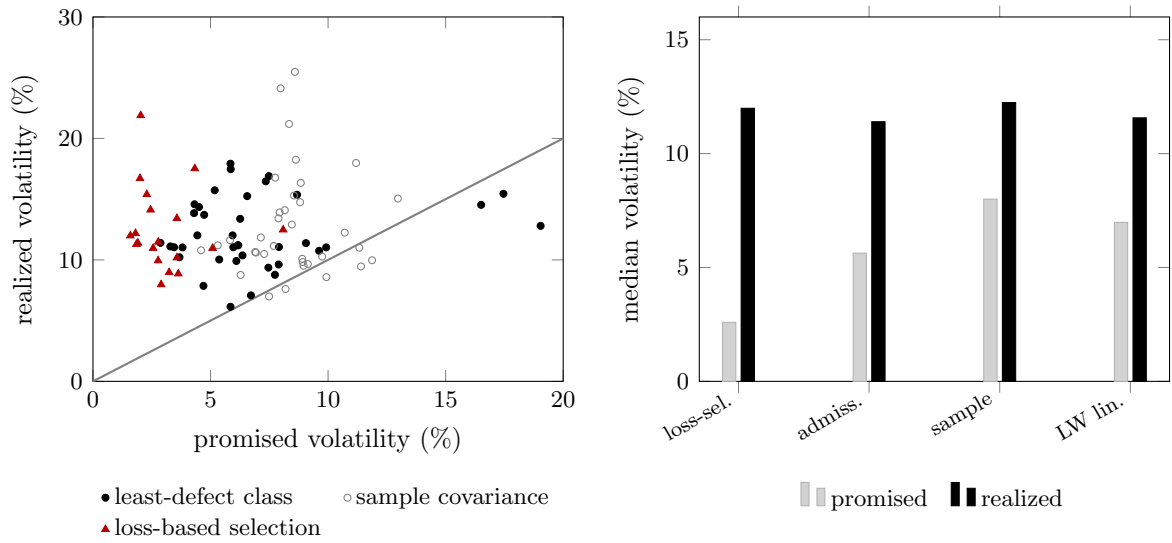

The experiment separates two statements that portfolio practice usually merges. A representation can be good enough to improve an allocation and not good enough to certify the diagonal-residual restriction its covariance imposes, and the framework distinguishes the two before any portfolio is formed: the diagnostic reports the failure, the portfolio is still built if the mandate requires it, and the risk it reports is known in advance to be understated. Closing the remaining dependence requires either a larger candidate universe than the 127 series declared here or an asset universe with fewer common components than twelve sector representatives or seventeen strategy indices, and that is the empirical question the framework poses.

\FloatBarrier
\section{Concluding remarks}\label{sec:discussion}

Portfolio optimization is usually posed after the conditioning information has been chosen, and two errors about that information, using what was not available at the decision time and treating common variation across assets as idiosyncratic, are then carried into the portfolio without detection. This paper makes the information a decision variable under hard admissibility constraints declared in advance, selects an information class by an order in which no risk--return quantity enters, and solves the classical problem inside it. Classical portfolio theory is preserved as the inner problem; the consequences live in the outer one.

The contributions are of three kinds. The first is the construction itself: a decision contract that fixes the claim and the conditions before any data are examined, an admissible class whose membership is certified rather than assumed, and a two-stage problem in which the information is chosen first and the portfolio second. The second is a set of results about that construction. The value of information is monotone only under an embedding that portfolio constraints can break, and a rule that selects information by the loss of the decision it produces therefore prefers look-ahead information whenever the candidate set contains it, which is why the chronological condition must precede the loss rather than be inferred from it. The two-stage solution and the joint envelope are the minimal elements of two orders on one set of admissible pairs, and the two coincide only when the representation is not allowed to change with the target. The diversifiable part of risk is a property of the admissible class, with an exact account of what the diagonal-residual restriction costs a given portfolio, and the restriction can only be met below a basket size set by how much common residual variation the representation leaves unexplained. Requiring causal identification has an explicit price, traded against the size of the search, and the mechanism a representation describes is invariant under interventions only under interventional admissibility, while the systematic risk it carries need not be. The certification produces a bound on every residual correlation, and that bound determines a robust covariance that repairs the risk a model reports without any parameter to tune. The third contribution is the finite-sample theory: both stages are consistent, the regret of the selected portfolio is second order in the estimation error while the weights themselves move at first order, and reselection across decision nodes can be governed by a hysteresis rule that reacts to certification statistics and never to performance.

What the evidence establishes is equally definite. On a controlled system the admissibility order recovers the exact separators, excludes every candidate that violates chronology although those candidates have the lowest in-sample loss, and converges at the rates the theory gives. On market data the separator condition can be met on small sets of individual securities and never on pre-diversified portfolios, an ordering the diversification result predicts and that enlarging the candidate search within the family examined does not overturn; on a book of a practitioner's size the condition is not met either, and there the maximum over pairs stops discriminating between representations while the portfolio-level consequence statistic becomes the operative diagnostic. Where the condition fails, the model built under the diagonal restriction understates the variance of its own portfolio by an amount that depends on the candidate universe, severely when the candidates are macroeconomic series alone and modestly once industry and style candidates are admitted, and the understatement is announced from the certification block before any portfolio is held and confirmed by the realized risk afterwards. Inserted into the standard minimum-variance rule, the covariance of the selected class realizes the volatility of the best static or dynamic shrinkage estimator and is never significantly better, while it dominates statistical factors of the same dimension; it trades less than the sample covariance at every scale. Selecting the same information by the loss it produces, with every series correctly dated, yields portfolios that promise a small fraction of the risk they realize, which is the look-ahead result observed in its overfitting form. Demanding loadings that are stable across regimes, whether defined by the calendar or by the policy instrument, leaves almost no admissible driver and raises out-of-sample variance by a measurable amount.

For practice the framework offers four things. A specification test for an existing factor risk model, which model-validation functions can run on the models they already have, using the screening defect with a size-controlled tolerance and the consequence statistic at the weights the mandate holds. A map of where factor models with diagonal residual should not be expected to hold, which includes any portfolio of pre-diversified assets. A rule for changing the information a book is conditioned on that uses certification statistics and not performance, so that the selection cannot learn from the profit and loss it produces and changes rarely. And a price for a claim that mandates often make implicitly, that the risk model should survive the next regime.

The limits follow from the setting as much as from the evidence. The framework does not lower realized portfolio risk, and the reason is that inside the quadratic mean--variance problem the information enters only through two conditional moments, so that an error in them moves the realized loss at second order however much it moves the reported risk; converting certified information into realized performance requires formulations in which the drivers enter through the trajectory of the constituents' sensitivities to them, which is a design step downstream of certification and is not taken here. The representations the framework selects are contemporaneous decompositions of risk and not signals, since nothing survives a one-session lag. Among the admissibility conditions only the mandate scope is a declaration; the others are certified from provenance, from the enlarged filtration or from returns. The interventional branch prices a stability restriction over observational environments, which coincides with the price of causal identification only under the identifying assumptions stated, and has no application until a mandate makes a claim about interventions that a conditional-law model cannot support. The controlled experiments and the geometry of the capital allocation surface are Gaussian and affine, and how the separator condition degrades under nonlinear response is not treated. Portfolio selection with intertemporal costs inside the selected classes is located by the separability result as the dynamic part of the problem and is the natural next step.

The question the paper poses is whether optimizing over a certified information--portfolio domain produces decisions more appropriate to the declared claim than optimizing weights after an arbitrary representation has been fixed. The controlled experiments answer affirmatively, and the market experiments show the certification reporting, before any portfolio is formed, a failure whose consequence for reported risk is then observed.

\section*{Declarations}
\noindent\textbf{Funding.} No specific funding was received for this research.

\noindent\textbf{Competing interests.} The author is employed by Miralta Finance Bank S.A., which provided institutional access to the licensed market data used in Section~\ref{sec:market}. The views expressed are the author's and not necessarily those of the institution.

\noindent\textbf{Data and code.} The controlled experiments use no external data and are reproduced by the accompanying code with fixed seeds. The equity, driver, individual-name and hedge-fund panels are licensed third-party series and are not redistributed; their construction is specified in Appendix~\ref{app:config}. The Fama--French portfolio files are public.


\appendix
\section{Results used, omitted statements and figures}\label{app:used}

\subsection{The conditional mean--variance frontier}

Given $\mu_H\in\R^{n}$ and $\Sigma_H\succ0$, the mean--variance problem with full investment is $\min\{w^{\Tr}\Sigma_Hw:\ \ind^{\Tr}w=1,\ \mu_H^{\Tr}w=m\}$ \cite{Markowitz1952}. Its solution is classical \cite{Merton1972}.

\begin{proposition}[Conditional frontier; \cite{Merton1972}]\label{prop:merton}
Let $a=\ind^{\Tr}\Sigma_H^{-1}\ind$, $b=\ind^{\Tr}\Sigma_H^{-1}\mu_H$, $c=\mu_H^{\Tr}\Sigma_H^{-1}\mu_H$ and $d=ac-b^{2}>0$. For every target $m$ the minimum-variance portfolio with $\ind^{\Tr}w=1$ and $\mu_H^{\Tr}w=m$ is
\begin{equation}\label{eq:merton-w}
\begin{aligned}
w_H(m)&=\frac{c-bm}{d}\,\Sigma_H^{-1}\ind+\frac{am-b}{d}\,\Sigma_H^{-1}\mu_H,\\
\rho^{\star}_H(m)&=w_H(m)^{\Tr}\Sigma_Hw_H(m)=\frac{am^{2}-2bm+c}{d}.
\end{aligned}
\end{equation}
The minimum-variance portfolio is $w_H(b/a)=\Sigma_H^{-1}\ind/a$, with variance $1/a$; the frontier $m\mapsto\rho^{\star}_H(m)$ is a convex parabola, and its efficient branch is $m\ge b/a$. The utility form $\min\{\tfrac{\gamma}{2}w^{\Tr}\Sigma_Hw-\mu_H^{\Tr}w:\ind^{\Tr}w=1\}$ has the unique minimizer $w_H(m(\gamma))$ with $m(\gamma)=b/a+d/(a\gamma)$ (the budget multiplier is $(b-\gamma)/a$ and the mean of the minimizer is $(c-b(b-\gamma)/a)/\gamma$), so every $\gamma>0$ selects one point of the efficient branch and every efficient target $m>b/a$ is selected by $\gamma=d/(a(m-b/a))$; the two parametrizations describe the same frontier. If a riskless rate $r_f<b/a$ is available, the tangency portfolio is $w_H^{\mathrm{tan}}=\Sigma_H^{-1}(\mu_H-r_f\ind)/(b-ar_f)$ and the capital allocation line is $\mu=r_f+\sqrt{c-2br_f+ar_f^{2}}\;\sigma$.
\end{proposition}

The proposition is invoked in Section~\ref{sec:two-stage} with $H$ the admissibility-optimal class; nothing in it depends on how $H$ was chosen. The sensitivity of $w_H(m)$ to errors in $\mu_H$ is the classical finding of \cite{Michaud1989,ChopraZiemba1993} and reappears in the finite-sample experiments. Figure~\ref{fig:mpt-apo} places the classical objects of this proposition next to their counterparts in the admissible problem, and Table~\ref{tab:homology} in Appendix~\ref{app:used} lists the correspondence.

\subsection{Causal separation, response and residual}

The statistical condition on which the admissibility theory rests is that a set of drivers renders asset returns conditionally independent over the horizon. Its consequences for conditional covariances are proved here.

\begin{definition}[Separator]\label{def:separator}
Let $Z$ be a vector of $k$ drivers with path $Z_{[t,t+h]}$ over the horizon, let $H^{+}:=\sigma(\calG^{0}_t,Z_{[t,t+h]})$ be the horizon-closed information it generates and $H^{-}:=\sigma(\calG^{0}_t,Z_t)$ its decision-time component. The drivers are a \emph{separator} for $A$ if the conditional covariance $\Cov(R\mid H^{+})$ is diagonal almost surely, that is, if the components of $R$ are pairwise conditionally uncorrelated given $H^{+}$; they are a \emph{structural separator} if the components of $R$ are mutually independent conditional on $H^{+}$, which implies the former and is the notion a causal model delivers (Proposition~\ref{prop:commonality}). A separator is \emph{minimal} if no separator generates a strictly smaller $\sigma$-algebra; minimality is with respect to information inclusion and does not order separators that are incomparable under inclusion. The mean--variance theory of this paper uses only the second-moment notion; conditional independence is reserved for the structural and interventional statements that need it.
\end{definition}

\begin{proposition}[Tower decomposition]\label{prop:tower}
For any drivers $Z$, with $\xi:=\E[R\mid H^{+}]-\E[R\mid H^{-}]$ the horizon innovation of the conditional mean,
\begin{equation}\label{eq:tower}
\Cov(R\mid H^{-})=\E\big[\Cov(R\mid H^{+})\,\big|\,H^{-}\big]+\Cov(\xi\mid H^{-}) .
\end{equation}
If $Z$ is a separator, the first term, denoted $\Omega$, is diagonal. If moreover $\E[R\mid H^{+}]$ is affine in a $k$-dimensional statistic $\zeta$ of the driver path, $\E[R\mid H^{+}]=\alpha+B\zeta$, then $\Cov(\xi\mid H^{-})=BVB^{\Tr}$ with $V=\Cov(\zeta\mid H^{-})$; $B$ is the \emph{response operator} and $V$ the innovation covariance. The first term is idiosyncratic and the second systematic.
\end{proposition}
\noindent The proof is in Appendix~\ref{app:proofs}.

Every object in~\eqref{eq:tower} is invariant under invertible bimeasurable recodings of the driver coordinates, since such a recoding leaves $\sigma(Z)$ and hence every conditional expectation unchanged; in the affine case $B$ transforms by the Jacobian and $BVB^{\Tr}$ is unchanged.\label{prop:minimal}

\label{rem:invariance}Interventional invariance follows from a structural model. If the drivers and returns are embedded in a structural causal model in which $Z$ is a cause of $R$, and an intervention acts on variables that are not descendants of $Z$ and leaves the mechanism from $Z$ to $R$ unchanged, then the conditional law of $R$ given $H^{+}$ is the same under the observational and the post-intervention law; this is the modularity of structural causal models \cite[Ch.~1 and~3]{Pearl2009}, \cite[Sec.~6.3]{PetersJanzing2017}. Every object in~\eqref{eq:tower} is then invariant to the intervention. The paper does not assume such a model; it makes invariance across a declared set of regimes the definition of the causal admissibility condition (Definition~\ref{def:claim}), and this remark records when that condition follows from a structural model.

The screening defect used in Section~\ref{sec:apo} to grade candidate representations is the departure from Definition~\ref{def:separator}: for a candidate $H$ it is the maximal absolute conditional correlation between components of $R$ after conditioning on $H$, equal to zero if and only if the candidate is a second-moment separator.

\subsection{The Commonality Principle}

\begin{proposition}[Commonality Principle, \cite{RD2023,RD2025causal}]\label{prop:commonality}
Let the constituents of a portfolio be generated by a structural causal model in which a declared set of drivers is the complete set of their common causes, there is no latent common cause, and no constituent is a cause of another (causal sufficiency). Call a driver common if it is a cause of at least two constituents. Then, by the causal Markov condition, the constituents are mutually independent conditional on the common drivers, so the common drivers are a separator; the drivers that are optimal for the portfolio's diversification dynamics, in the sense that conditioning on them removes the dependence between constituents that is not idiosyncratic, are the common drivers; conditioning on a set that omits a common cause leaves the dependence it induces, and conditioning on a common effect of two constituents creates a dependence absent without it. Portfolio risk conditional on the common drivers decomposes into a systematic part carried by the drivers' innovations and an idiosyncratic part that is diversifiable in the sense of Proposition~\ref{prop:diversification}.
\end{proposition}

The principle is the causal origin of the diversification result of Section~\ref{sec:diversification}. What a portfolio can diversify away is what remains after conditioning on the common causes, and the residual $\Omega$ of Proposition~\ref{prop:tower} is the formal object that remainder denotes.


\begin{table}[!htb]\centering\small
\caption{Correspondence between the objects of modern portfolio theory and their admissible counterparts.}
\label{tab:homology}
\begin{tabular}{@{}p{4.1cm}p{5.5cm}p{5.2cm}@{}}
\toprule
Modern portfolio theory & Admissible portfolio optimization & What changes\\
\midrule
moments $\mu,\Sigma$ taken as given & conditional moments $\mu_H,\Sigma_H$ of an admissible $H$ & the information is a decision variable subject to admissibility conditions\\
feasible set of portfolios & admissible pairs $\Gamma_t(m)$ & pairs, not weights; inadmissible $H$ absent\\
efficient frontier & conditional frontier of $\Hst$ (Def.~\ref{def:frontier}) & one frontier per class; the reported one is that of the admissibility-optimal class\\
minimum-variance portfolio & conditional minimum-variance portfolio of $\Hst$ & unchanged in form\\
tangency portfolio, capital allocation line & admissible tangent portfolio and admissible CAL, one line per decision state; their union over states is the admissible capital allocation surface (Fig.~\ref{fig:cal-surface}) & the slope becomes state-dependent; its root-mean-square is the unconditional Sharpe ratio\\
indifference curve, utility & the same preference, applied only in Stage~II & preference cannot reselect information (Def.~\ref{def:apo-canonical})\\
--- & joint envelope over $\calA^{0}_t$ (Sec.~\ref{sec:joint}) & a different contract, not an object of the classical theory\\
diversification acts on $\Sigma$ & diversification acts on the residual $\Omega_{\Hst}$ (Prop.~\ref{prop:diversification}) & the diversifiable part is fixed by the class\\
estimation error in $\mu,\Sigma$ & estimation error in $\mu_H,\Sigma_H$ plus selection error in $H$ (Sec.~\ref{sec:finite}) & a second source of error with its own bound\\
\bottomrule
\end{tabular}
\end{table}

\subsection{Establishing the admissibility conditions}
Causal admissibility is a constraint rather than a score. For an interventional claim, representations without identified intervention semantics are infeasible; for a conditional-law claim they are not excluded for lacking them.

\label{rem:certified}The conditions differ in how they are established. Chronological availability is a deterministic check on provenance. The enlargement condition is the classical question of the theory of enlargement of filtrations: an initial enlargement violating Jacod's condition can destroy the semimartingale property and create arbitrage, the mechanism by which an insider's expected utility becomes infinite \cite{AmendingerImkellerSchweizer1998}, so the condition is a hypothesis on the enlarged filtration, verifiable within a model. Implementation feasibility is verifiable from exchange rules and liquidity data; the statistical and causal conditions are certified from returns, with sampling error quantified in Section~\ref{sec:finite}. Only the mandate scope is a declaration, with the same status as the set of admissible strategies in \cite{DelbaenSchachermayer1994}, and its misspecification is model risk rather than estimation error.

Three candidates illustrate the parts. A published factor return passes (i) by its release calendar and (ii) automatically, being in the market filtration. A proprietary order-flow signal computed before $t$ passes (i) and is the case for which (ii) is written. A realized horizon return fails (i) on the calendar and (ii) in the extreme form in which the informed agent's utility is unbounded.

\begin{definition}[Decision equivalence]\label{def:equiv}
Two admissible representations $H$ and $\tilde H$ are \emph{decision-equivalent}, $H\eqv\tilde H$, if the feasible sets $\calW(H)$ and $\calW(\tilde H)$ coincide and every population decision primitive entering $J$ coincides almost surely, so that $J(\cdot;H)$ and $J(\cdot;\tilde H)$ coincide as functions on the common feasible set. In the exact conditional-law case, equality of the generated $\sigma$-algebras modulo null sets together with identical implementation constraints is sufficient.
\end{definition}

\subsection{Canonicality, a two-asset example and convexity}
\begin{proposition}[Canonicality of Stage I]\label{thm:canonical}
If the admissibility conditions and the four quantities depend on a candidate only through the information it generates, the selected class is a union of information-equivalence classes and is invariant to invertible recodings of the candidate variables, and when exact admissible separators exist the selected class consists of exact separators, with the one of least causal defect, dimension and burden among them belonging to it. The formal statement is Proposition~\ref{thm:canonical-formal} in Appendix~\ref{app:formal}.
\end{proposition}

\begin{example}[A representation switch that a fixed-representation optimizer cannot express]\label{ex:two-asset}
Consider two assets under full investment and two admissible information states with conditional moments
\[
\Sigma_{H_1}=\begin{pmatrix}1&0.8\\0.8&1\end{pmatrix},\ \mu_{H_1}=\begin{pmatrix}0.12\\0.08\end{pmatrix};\qquad
\Sigma_{H_2}=\begin{pmatrix}1&0.1\\0.1&1\end{pmatrix},\ \mu_{H_2}=\begin{pmatrix}0.09\\0.10\end{pmatrix}.
\]
For $\gamma=2$ and $c_H\equiv0$, problem~\eqref{eq:markowitz} returns $w^{\star}_{H_1}=(0.55,0.45)$ with $J=0.799$, and $w^{\star}_{H_2}=(0.497,0.503)$ with $J=0.455$. A Markowitz program produces either allocation but cannot compare them, since each conditional law is one of its inputs. The envelope problem~\eqref{eq:apo} compares $V(H_1)=0.799$ with $V(H_2)=0.455$ and selects $(H_2,w^{\star}_{H_2})$ although $H_1$ carries the larger conditional mean. APO proper (Definition~\ref{def:apo-canonical}) compares $H_1$ and $H_2$ by admissibility and constructs only the winner's frontier; Proposition~\ref{prop:lex} gives the exact relation between the two.
\end{example}

\begin{proposition}[Convexity of the conditional frontier]\label{thm:markowitz-geometry}
Inside a fixed class with a convex feasible set, convex risk and affine mean, the conditional frontier is convex in the target, the efficient set at each target is convex and is a singleton when the risk is strictly convex, and the frontier is continuous on the interior of the feasible range. Formal statement: Proposition~\ref{thm:markowitz-formal}, Appendix~\ref{app:formal}.
\end{proposition}

\label{sec:setvalued}

\section{Formal statements and proofs}\label{app:proofs}

\subsection{Formal statements}\label{app:formal}

\begin{theorem}[Existence]\label{thm:existence-formal}
If the candidate set is finite and each $\calW(H)$ is compact with $J(\cdot;H)$ lower semicontinuous, the envelope problem~\eqref{eq:apo} has a minimizer. More generally a minimizer exists when $\calA_t/\!\eqv$ is compact, $H\mapsto\calW(H)$ is compact-valued and upper hemicontinuous and $J$ is lower semicontinuous on its graph.
\end{theorem}

\begin{theorem}[Existence of the two solution sets]\label{thm:two-sets-formal}
If the candidate set is finite, $\Ast$ is nonempty, and for every $H\in\Ast$ with $\calW(H)$ nonempty and compact and $\rho_H,\mu_H$ continuous, the conditional efficient set
\begin{equation}\label{eq:eff-set}
\calE_H(m)=\argmin\{\rho_H(w):\ w\in\calW(H),\ \mu_H(w)\ge m\}
\end{equation}
is nonempty and compact for every feasible target $m$. The same holds when $\calA_t/\!\eqv$ is compact and the four admissibility quantities are lower semicontinuous on it.
\end{theorem}

\begin{proposition}[Convexity of the conditional frontier]\label{thm:markowitz-formal}
Let $\rho_H(w)=w^{\Tr}\Sigma_Hw$ with $\Sigma_H\succ0$, $\mu_H(w)=\mu_H^{\Tr}w$, and let $\calW(\Hst)$ be defined by finitely many affine equalities and convex inequalities. Then for every feasible $m$ the Stage-II problem~\eqref{eq:frontier} is a convex program with a unique solution, and $m\mapsto\rho^{\star}_t(m)$ is convex and nondecreasing on the set of feasible targets.
\end{proposition}

\begin{proposition}[Stage-I selection consistency, finite class]\label{prop:stage1-formal}
Let $\calA_t$ be finite, let $s_t$ represent $\prec_{\calA}$ and $\hat s_t$ represent $\widehat\prec_{\calA}$ on $\widehat\calA^{0}_t$, and suppose (i) $\Prob\{\widehat\calA^{0}_t=\calA_t\}\to1$ and (ii) $\max_{H\in\calA_t}|\hat s_t(H)-s_t(H)|\pto0$. If $\Delta_a:=\min\{s_t(H)-\min s_t:\ H\in\calA_t\setminus\Ast\}>0$, then $\Prob\{\widehat\calA^{\star}_t\subseteq\Ast\}\to1$, and $\Prob\{\widehat\calA^{\star}_t=\Ast\}\to1$ whenever $\hat s_t$ takes equal values on $\Ast$ with probability tending to one.
\end{proposition}

\begin{proposition}[Stage-II frontier consistency]\label{prop:stage2-formal}
Assume Stage~I consistently selects a decision-equivalent class $[\Hst]$ and write $\calW:=\calW(\Hst)$, $\rho:=\rho_{\Hst}$, $\mu:=\mu_{\Hst}$. Suppose (i) $\calW$ is known, nonempty, convex and compact, $\rho$ is continuous and convex, and $\mu$ is continuous and affine on $\calW$; (ii) the estimators $\hat\rho_N,\hat\mu_N$ are continuous on $\calW$ and $\eta_N:=\sup_{\calW}|\hat\rho_N-\rho|+\sup_{\calW}|\hat\mu_N-\mu|\pto0$; (iii) the target interval $\calM=[m_-,m_+]$ satisfies $m_-<m_+<\max_{\calW}\mu$. Let $\hat\rho^{\star}_N(m)=\min\{\hat\rho_N(w):w\in\calW,\hat\mu_N(w)\ge m\}$ and $\widehat\calE_N(m)$ the corresponding minimizer set. Then
\begin{equation}\label{eq:frontier-consistency}
\sup_{m\in\calM}\big|\hat\rho^{\star}_N(m)-\rho^{\star}_t(m)\big|\pto0,
\end{equation}
and every accumulation point of every sequence $\hat w_N\in\widehat\calE_N(m)$ belongs to $\calE^{\star}_t(m)$.
\end{proposition}

\begin{proposition}[Canonicality of Stage I]\label{thm:canonical-formal}
Suppose the required admissibility conditions and the four quantities $\varepsilon_t,\delta_t,\kappa_t,\iota_t$ are constant on information-equivalence classes. Then

\smallskip\noindent\textit{Part (a).}\quad $\Ast$ is a union of information-equivalence classes and is invariant to invertible measurable recodings of candidate variables; (b) if an exact admissible separator exists, then $\Ast$ consists of exact admissible separators, those with defect zero; if moreover an exact admissible separator $H^{*}_t$ has, among the exact admissible separators, the least causal defect, then the least dimension and then the least implementation burden, then $[H^{*}_t]\subseteq\Ast$. Minimality under information inclusion alone does not imply the latter, because an exact separator incomparable with $H^{*}_t$ under inclusion may have lower dimension or burden.
\end{proposition}

\subsection{Proofs}

\begin{proof}[Proof of Proposition~\ref{prop:tower}]

\smallskip\noindent\textit{Step 1. The law of total covariance.}\quad Since $H^{-}\subseteq H^{+}$, conditioning first on $H^{+}$ and then on $H^{-}$ gives, for the random vector $R$,
\[
\Cov(R\mid H^{-})=\E\big[\Cov(R\mid H^{+})\,\big|\,H^{-}\big]+\Cov\big(\E[R\mid H^{+}]\,\big|\,H^{-}\big).
\]
Because $\E[R\mid H^{-}]$ is $H^{-}$-measurable, subtracting it inside the second covariance changes nothing, so
\[
\Cov\big(\E[R\mid H^{+}]\,\big|\,H^{-}\big)=\Cov\big(\E[R\mid H^{+}]-\E[R\mid H^{-}]\,\big|\,H^{-}\big)=\Cov(\xi\mid H^{-}),
\]
which is~\eqref{eq:tower}.

\smallskip\noindent\textit{Step 2. Diagonality under separation.}\quad If $Z$ is a separator, $\Cov(R\mid H^{+})$ is diagonal almost surely by definition, hence so is its $H^{-}$-conditional expectation, which is the first term $\Omega$; a structural separator is a separator because conditional independence of square-integrable variables implies zero conditional covariance.

\smallskip\noindent\textit{Step 3. The affine case.}\quad If $\E[R\mid H^{+}]=\alpha+B\zeta$ with $\zeta$ a statistic of the driver path, then $\E[R\mid H^{-}]=\alpha+B\,\E[\zeta\mid H^{-}]$ by the tower property, so $\xi=B\big(\zeta-\E[\zeta\mid H^{-}]\big)$ and
\[
\Cov(\xi\mid H^{-})=B\,\Cov(\zeta\mid H^{-})\,B^{\Tr}=BVB^{\Tr}.\qedhere
\]
\end{proof}

\begin{proof}[Proof of Proposition~\ref{thm:canonical-formal}]

\smallskip\noindent\textit{Part (a).}\quad If $\phi$ is an invertible measurable recoding, $\sigma(\phi(H))=\sigma(H)$, so $\phi(H)$ and $H$ lie in the same information-equivalence class and, by hypothesis, have the same admissibility-condition values and the same four quantities; the order of Definition~\ref{def:adm-order} and its set of minimal elements are therefore unchanged, and $\Ast$ contains every element of a class as soon as it contains one.

\smallskip\noindent\textit{Part (b).}\quad By definition of a least exact separator, $\varepsilon_t(H)=0$ exactly for the exact separators in $\calA_t$, so the first coordinate of the lexicographic order retains the exact admissible separators and only them; under a conditional-law contract $\delta_t\equiv0$ and the second coordinate retains them all. The remaining coordinates, dimension and burden, order the exact separators among themselves, and minimality under information inclusion does not decide that order, since two exact separators can be incomparable under inclusion with different dimensions or burdens. Hence $\Ast$ consists of exact admissible separators, and $H^{*}_t\in\Ast$ exactly when it is minimal among them in $(\delta_t,\kappa_t,\iota_t)$, which is the stated hypothesis; $[H^{*}_t]\subseteq\Ast$ then follows by (a).

\smallskip\noindent\textit{Part (c).}\quad The four quantities are functions of $H$ alone; a selection rule depending on $m$ or on $w$ is a different order and hence a different contract.
\end{proof}

\begin{proof}[Proof of Theorem~\ref{thm:two-sets-formal}]

\smallskip\noindent\textit{Step 1 (Stage I).}\quad Write $a^{1},a^{2},a^{3},a^{4}$ for the four quantities $\varepsilon_t,\delta_t,\kappa_t,\iota_t$ of Definition~\ref{def:adm-order}, in that order. Set $X^{0}:=X_t$ and, recursively, $X^{k}:=\argmin_{x\in X^{k-1}}a^{k}(x)$ for $k=1,\dots,4$. If $X^{k-1}$ is nonempty and compact, then $a^{k}$, being lower semicontinuous, attains its minimum on $X^{k-1}$ \cite[Thm.~2.43]{AliprantisBorder2006}, and $X^{k}=X^{k-1}\cap\{a^{k}\le\min_{X^{k-1}}a^{k}\}$ is the intersection of a compact set with a closed set (lower level sets of lsc functions are closed), hence nonempty and compact. By induction $X^{4}$ is nonempty and compact. It remains to identify $X^{4}$ with the set of $\lex$-minimal elements: $x\in X^{4}$ means that $x$ minimizes $a^{1}$, then $a^{2}$ among $a^{1}$-minimizers, and so on, which is precisely the statement that no $x'$ satisfies $x'\prec_{\calA}x$; conversely, if $x\notin X^{4}$, let $k$ be the first index with $x\notin X^{k}$; then $x\in X^{k-1}$ and any $x'\in X^{k}$ ties $x$ on $a^{1},\dots,a^{k-1}$ and is strictly smaller on $a^{k}$, so $x'\prec_{\calA}x$ and $x$ is not minimal. Hence $\Ast=X^{4}$ (as a union of classes). 

\smallskip\noindent\textit{Step 2 (Stage II).}\quad Fix $H\in\Ast$ and a feasible $m$. The set $F_H(m):=\{w\in\calW(H):\mu_H(w)\ge m\}$ is the preimage of the closed set $[m,\infty)$ under the continuous $\mu_H$, intersected with the compact $\calW(H)$; it is therefore compact and, by assumption, nonempty. The continuous $\rho_H$ attains its minimum on $F_H(m)$, and $\calE_H(m)=F_H(m)\cap\rho_H^{-1}(\{\min_{F_H(m)}\rho_H\})$ is closed in a compact set, hence compact.
\end{proof}

\begin{proof}[Proof of Theorem~\ref{thm:voi}]

\smallskip\noindent\textit{Step 0. The loss used in the paper satisfies the embedding hypothesis.}\quad For the conditional mean--variance loss $\ell(w,R;H)=\tfrac{\gamma}{2}\Var(w^{\Tr}R\mid H)-\E[w^{\Tr}R\mid H]$ and a $\sigma(H_1^{-})$-measurable policy $\pi$, the law of total variance applied to the finer information gives
\[
\E\big[\Var(\pi^{\Tr}R\mid H_2)\big]=\E\big[\Var(\pi^{\Tr}R\mid H_1)\big]-\E\big[\Var\big(\E[\pi^{\Tr}R\mid H_2]\,\big|\,H_1\big)\big]\le\E\big[\Var(\pi^{\Tr}R\mid H_1)\big],
\]
while $\E\big[\E[\pi^{\Tr}R\mid H_2]\big]=\E\big[\E[\pi^{\Tr}R\mid H_1]\big]$. Hence the identity map $\iota\pi=\pi$ satisfies $J(\pi;H_2)\le J(\pi;H_1)$, which is the embedding inequality of Definition~\ref{def:order}; this is the case used throughout the paper, and it is the reason the definition is stated for expected losses, since the pointwise inequality between the two losses does not hold in general.

\smallskip\noindent\textit{Step 1. Weak inequality.}\quad Let $\iota:\Pi(H_1)\to\Pi(H_2)$ be the embedding. For every $\pi\in\Pi(H_1)$, $\iota\pi\in\Pi(H_2)$ and, by the embedding inequality,
\[
J(\iota\pi;H_2)\le J(\pi;H_1).
\]
Taking infima over $\pi\in\Pi(H_1)$ on both sides and enlarging the set on the left,
\[
V(H_1)=\inf_{\pi\in\Pi(H_1)}J(\pi;H_1)\ \ge\ \inf_{\pi\in\Pi(H_1)}J(\iota\pi;H_2)\ \ge\ \inf_{\pi'\in\Pi(H_2)}J(\pi';H_2)=V(H_2).
\]

\smallskip\noindent\textit{Step 2. Strict inequality.}\quad Under the additional hypothesis $J(\pi_2;H_2)<J(\iota\pi_1;H_2)$, hence
\[
V(H_1)=J(\pi_1;H_1)\ \ge\ J(\iota\pi_1;H_2)\ >\ J(\pi_2;H_2)\ \ge\ V(H_2).\qedhere
\]
\end{proof}

\begin{proof}[Proof of Corollary~\ref{cor:not-maximal}]
If $H_2\notin\calA_t$ then $V(H_2)$ does not enter~\eqref{eq:apo}. If the embedding fails, Step~1 of the proof of Theorem~\ref{thm:voi} has no content, and $V(H_2)>V(H_1)$ occurs as soon as $\calW(H_2)$ excludes $w^{\star}_{H_1}$ while the conditional law is unchanged, as in Example~\ref{ex:two-asset} with a representation-dependent trading restriction.
\end{proof}

\begin{proof}[Proof of Corollary~\ref{prop:lookahead}]

\smallskip\noindent\textit{Step 1. Weak inequality.}\quad Since $\sigma(H^{-})\subseteq\sigma(H^{-},X^{+})$ and $\calW(H^{\sharp})=\calW(H)$, the identity map embeds the $H$ problem in the $H^{\sharp}$ problem in the sense of Definition~\ref{def:order}, by Step~0 of the proof of Theorem~\ref{thm:voi}. The theorem gives $V(H^{\sharp})\le V(H)$.

\smallskip\noindent\textit{Step 2. Strict inequality for the quadratic loss.}\quad Let $\pi$ be the $H^{-}$-measurable optimizer and $\pi^{\sharp}$ the $H^{\sharp}$-measurable optimizer. Conditional on $H^{\sharp}$ the objective $w\mapsto\tfrac{\gamma}{2}w^{\Tr}\Cov(R\mid H^{\sharp})w-\E[R\mid H^{\sharp}]^{\Tr}w$ is strictly convex with unconstrained minimizer
\[
w^{\sharp}=\gamma^{-1}\Cov(R\mid H^{\sharp})^{-1}\E[R\mid H^{\sharp}].
\]
and the $H^{-}$-measurable optimizer is $\pi=\gamma^{-1}\Cov(R\mid H^{-})^{-1}\E[R\mid H^{-}]$. On the event where the two differ, which has positive probability by hypothesis, strict convexity of the conditional objective and the fact that $\calW$ contains a neighbourhood of the optimizer give $\pi^{\sharp}$ a strictly smaller conditional objective than $\pi$; elsewhere the two coincide. Since the conditional objective of $\pi$ given $H^{\sharp}$ averages to $J(\pi;H)$ by the tower property, taking expectations gives $J(\pi^{\sharp};H^{\sharp})<J(\pi;H)$, hence $V(H^{\sharp})<V(H)$. A change in the conditional mean alone does not suffice for the hypothesis, because a simultaneous change in the conditional covariance can leave the optimizer unchanged; the condition is placed on the optimizer for that reason.

\smallskip\noindent\textit{Step 3. Consequence for selection rules.}\quad On any candidate set containing both $H$ and $H^{\sharp}$, a rule that returns a minimizer of $V$ cannot return $H$, which is strictly dominated by $H^{\sharp}$ in the set. This alone does not force the minimizer to contain $X^{+}$: a third candidate $G$ excluding $X^{+}$ may satisfy $V(G)<V(H^{\sharp})$. Under either of the two additional hypotheses it does. If $H$ minimizes $V$ among the candidates excluding $X^{+}$, every such candidate has value at least $V(H)>V(H^{\sharp})$ and cannot be a minimizer. If every candidate excluding $X^{+}$ has its own augmentation in the set, each is strictly dominated by its augmentation by Steps~1 and~2, so none is a minimizer. In both cases every minimizer contains $X^{+}$. Decision-loss selection is therefore consistent with the chronological condition only on candidate sets from which that condition has already removed every such candidate.
\end{proof}

\begin{proof}[Proof of Theorem~\ref{thm:existence-formal}]

\smallskip\noindent\textit{Step 1 (the graph is compact).}\quad A compact-valued uhc correspondence into a Hausdorff space has closed graph \cite[Thm.~17.10]{AliprantisBorder2006}, and the image $\calW(X)=\bigcup_{x\in X}\calW(x)$ of the compact set $X$ under a compact-valued uhc correspondence is compact \cite[Lemma~17.8]{AliprantisBorder2006}. Hence $\Gamma$ is a closed subset of the compact set $X\times\calW(X)$ and is therefore compact. 

\smallskip\noindent\textit{Step 2 (attainment).}\quad A lower semicontinuous function with values in $\R\cup\{+\infty\}$ attains its infimum on a nonempty compact set \cite[Thm.~2.43]{AliprantisBorder2006}; the infimum is finite because $J\not\equiv+\infty$. Any point $(x^{\star},w^{\star})$ at which it is attained, together with any representative $\Hs\in x^{\star}$, is a minimizer of~\eqref{eq:apo} by the quotient invariance of $V$.
\end{proof}

\begin{proof}[Proof of Proposition~\ref{prop:lex}]

\smallskip\noindent\textit{Step 1. Every minimal pair lies in the union.}\quad Let $(H,w)\in\Gamma_t(m)$. If $H\notin\Ast$, there is $H'\in\calA_t$ with $H'\prec_{\calA}H$; for every $w'\in\calW(H')$ with $\mu_{H'}(w')\ge m$, which exists because $m$ lies in the common target domain, the first clause of~\eqref{eq:apo-order} gives $(H',w')\prec_{\mathrm{APO}}(H,w)$, so $(H,w)$ is not minimal. If $H\in\Ast$ but $w\notin\calE_H(m)$, there is $w'$ feasible for $H$ with $\rho_H(w')<\rho_H(w)$, and the second clause gives $(H,w')\prec_{\mathrm{APO}}(H,w)$.

\smallskip\noindent\textit{Step 2. Every pair in the union is minimal.}\quad\sloppy  Let $H\in\Ast$, $w\in\calE_H(m)$, and suppose $(H',w')\prec_{\mathrm{APO}}(H,w)$. The first clause would require $H'\prec_{\calA}H$, excluded by the minimality of $H$ in $\Ast$; the second requires $H'=H$ and $\rho_H(w')<\rho_H(w)$, excluded by $w\in\calE_H(m)$. Hence $(H,w)$ is minimal, and~\eqref{eq:lex-char} follows.

\smallskip\noindent\textit{Step 3. The envelope.}\quad Under the total preorder $\preceq_{\mathrm{env}}$ every two pairs are comparable, so the minimal elements are the pairs of least risk over all of $\Gamma_t(m)$, which are the pairs attaining~\eqref{eq:joint}; with $\rho$ replaced by $J$ and the target constraint dropped, they are the minimizers of~\eqref{eq:apo}.
\end{proof}

\begin{proof}[Proof of Proposition~\ref{thm:markowitz-formal}]
The feasible set $F(m)=\calW(\Hst)\cap\{w:\mu_H^{\Tr}w\ge m\}$ is the intersection of convex sets, hence convex, and closed. The objective is strictly convex since $\Sigma_H\succ0$. If $w_1\neq w_2$ were both minimizers, then $\tfrac12(w_1+w_2)\in F(m)$ by convexity and $\rho_H(\tfrac12(w_1+w_2))<\tfrac12\rho_H(w_1)+\tfrac12\rho_H(w_2)=\min_{F(m)}\rho_H$ by strict convexity, a contradiction; existence follows from Theorem~\ref{thm:two-sets} when $\calW(\Hst)$ is bounded and from coercivity of $\rho_H$ otherwise. Monotonicity holds because $F(m')\subseteq F(m)$ for $m'\ge m$. Convexity of the value function: for feasible $m_0,m_1$ and $\theta\in[0,1]$, let $w_i\in\calE_H(m_i)$; then $\mu_H^{\Tr}(\theta w_1+(1-\theta)w_0)\ge\theta m_1+(1-\theta)m_0$ and $\theta w_1+(1-\theta)w_0\in\calW(\Hst)$, so $\rho^{\star}_t(\theta m_1+(1-\theta)m_0)\le\rho_H(\theta w_1+(1-\theta)w_0)\le\theta\rho^{\star}_t(m_1)+(1-\theta)\rho^{\star}_t(m_0)$ \cite[Thm.~29.1]{Rockafellar1970}.
\end{proof}

\begin{proof}[Proof of Proposition~\ref{prop:sharpe}]
(a) is the definition of the supremum and of the Stage-I class.

\smallskip\noindent\textit{Part (b).}\quad For the quadratic loss $J_\gamma(w;H)=\tfrac{\gamma}{2}w^{\Tr}\Sigma_Hw-(\mu_H-r_f\ind)^{\Tr}w$ over $\R^{n}$ the minimizer is $w=\gamma^{-1}\Sigma_H^{-1}(\mu_H-r_f\ind)$ and the conditional minimum value is $-S_t(H)^{2}/(2\gamma)$, so $V_\gamma(H)=-\bar S_t(H)^{2}/(2\gamma)$. Theorem~\ref{thm:voi} applied to $J_\gamma$ gives $V_\gamma(H_2)\le V_\gamma(H_1)$ for every $\gamma>0$, hence $\bar S_t(H_2)^{2}\ge\bar S_t(H_1)^{2}$; the normalization $\ind^{\Tr}w=1$ does not change the maximal slope because the Sharpe ratio is scale-invariant.
\end{proof}

\begin{proof}[Proof of Proposition~\ref{prop:diversification}]

\smallskip\noindent\textit{Step 1. Part (a): the residual is diagonal and its risk spreads.}\quad By Proposition~\ref{prop:tower},
\[
\Cov(R\mid H^{-})=\underbrace{\E\big[\Cov(R\mid H^{+})\,\big|\,H^{-}\big]}_{\Omega_H}+\underbrace{\Cov\big(\E[R\mid H^{+}]\,\big|\,H^{-}\big)}_{B_HV_HB_H^{\Tr}}.
\]
Under separation $\Cov(R\mid H^{+})$ is diagonal almost surely by definition, and so is its $H^{-}$-conditional expectation $\Omega_H$; under structural separation this follows from conditional independence of square-integrable variables. For a diagonal $\Omega_H$,
\[
w^{\Tr}\Omega_Hw=\sum_{i}w_i^{2}\Omega_{H,ii}\le\Big(\max_i\Omega_{H,ii}\Big)\|w\|_2^{2},
\]
which is~\eqref{eq:idio-bound}; for $w=\ind/n$, $\|w\|_2^{2}=1/n$. The systematic term is $w^{\Tr}B_HV_HB_H^{\Tr}w=\beta_w^{\Tr}V_H\beta_w$ with $\beta_w=B_H^{\Tr}w$, and two portfolios with the same exposure $\beta_w$ have the same systematic risk however their weights are spread.

\smallskip\noindent\textit{Step 2. Part (b): off-diagonal residual mass does not spread.}\quad Write $\Omega_H=\operatorname{diag}(\Omega_H)+\Omega^{\mathrm{off}}_H$. Then
\[
w^{\Tr}\Omega_Hw=\sum_iw_i^{2}\Omega_{H,ii}+w^{\Tr}\Omega^{\mathrm{off}}_Hw,
\qquad
\Big(\tfrac{\ind}{n}\Big)^{\Tr}\Omega^{\mathrm{off}}_H\Big(\tfrac{\ind}{n}\Big)=\frac{1}{n^{2}}\sum_{i\neq j}\Omega_{H,ij}=\frac{n-1}{n}\,\overline{\Omega^{\mathrm{off}}_H},
\]
where $\overline{\Omega^{\mathrm{off}}_H}$ is the average off-diagonal entry; the second term tends to that average, not to zero.

\smallskip\noindent\textit{Step 3. Part (c): invariance under the declared regimes.}\quad By Definition~\ref{def:claim}, $\mathcal L^{e}(R\mid H^{+})=\mathcal L^{o}(R\mid H^{+})$ for every $e\in\mathcal E$. The response $B_H$, the conditional covariance $\Cov(R\mid H^{+})$ and hence its diagonality are functionals of this conditional law, so they are identical under $e$ and under the observational regime; $\Omega_H=\E[\Cov(R\mid H^{+})\mid H^{-}]$ is then invariant as well. The covariance $V_H=\Cov(\zeta\mid H^{-})$ of the driver innovation is a functional of the law of $H^{+}$ given $H^{-}$, which the hypothesis does not constrain, so the systematic term is invariant only under the additional condition stated. A conditional-law-admissible $H$ that explained covariation through a channel the intervention severs no longer satisfies the hypothesis of part (a) afterwards, and part (b) applies to its post-intervention residual.

\smallskip\noindent\textit{Step 4. Part (d): the exact decomposition.}\quad Write $\Omega_H=\operatorname{diag}(s)\,C\,\operatorname{diag}(s)$ with $C$ the residual correlation matrix and $s_i=\Omega_{H,ii}^{1/2}$, and set $v=w\odot s$. Then
\[
\begin{aligned}
w^{\Tr}\Omega^{\mathrm{off}}_Hw&=\sum_{i\neq j}v_iv_jc_{ij}=D\,\bar c_w,\\
w^{\Tr}\operatorname{diag}(\Omega_H)w&=\|v\|_2^{2},\\
D&=\Big(\sum_iv_i\Big)^{2}-\|v\|_2^{2}=\|v\|_2^{2}\Big(\frac{(\sum_iv_i)^{2}}{\|v\|_2^{2}}-1\Big),
\end{aligned}
\]
and dividing the first identity by $w^{\Tr}\Sigma_Hw$ and inserting the third gives~\eqref{eq:misallocation}. If $v\ge0$ then $D=\sum_{i\neq j}v_iv_j\ge0$, the weights $v_iv_j/D$ form a probability distribution over the pairs, $\bar c_w$ is a convex combination of the $c_{ij}$, and therefore $|\bar c_w|\le\max_{i\neq j}|c_{ij}|=\varepsilon_t(H)$.
\end{proof}

\begin{proof}[Proof of Theorem~\ref{thm:distance}]

\smallskip\noindent\textit{Step 1. Lower bound.}\quad $\calA_t(\mathrm i)\subseteq\calA_t(\mathrm c)$, so the infimum over the smaller class is no smaller: $V^{\star}_{\mathrm i}\ge V^{\star}_{\mathrm c}$ and $\Delta_{\mathrm{ci}}\ge0$.

\smallskip\noindent\textit{Step 2. Upper bound.}\quad Fix any $G\in\calA_t(\mathrm i)$. By definition of $V^{\star}_{\mathrm i}$ as an infimum over a set containing $G$, and by the Lipschitz property,
\[
V^{\star}_{\mathrm i}\le V(G)\le V(H^{\star}_0)+L\,d(G,H^{\star}_0)=V^{\star}_{\mathrm c}+L\,d(G,H^{\star}_0).
\]
Hence $\Delta_{\mathrm{ci}}\le L\,d(G,H^{\star}_0)$ for every $G\in\calA_t(\mathrm i)$, and taking the infimum over $G$ gives~\eqref{eq:distance}.

\smallskip\noindent\textit{Step 3. The quadratic case.}\quad For $J(w;H)=\tfrac{\gamma}{2}w^{\Tr}\Sigma_Hw-\mu_H^{\Tr}w$ and any $w\in\calW$ with $\|w\|_2\le\bar r$,
\[
|J(w;H)-J(w;G)|\le\tfrac{\gamma}{2}\,\bar r^{2}\,\|\Sigma_H-\Sigma_G\|_{\mathrm{op}}+\bar r\,\|\mu_H-\mu_G\|_2=\bar r\,d(H,G),
\]
and two functions whose pointwise difference is bounded by a constant have infima that differ by at most that constant, so $V$ is $\bar r$-Lipschitz for the metric~\eqref{eq:moment-metric}. When the moments are random functions of the state the same bound holds pointwise in the state, and its expectation gives the Lipschitz property for the expected metric.
\end{proof}

\begin{proof}[Proof of Proposition~\ref{prop:regret2}]

\smallskip\noindent\textit{Step 1. The regret is bounded by a perturbation times a distance.}\quad Write $r:=J(\hat w)-J(w^{\star})\ge0$, $D:=J-\widehat J$, $\Delta_\Sigma:=\Sigma-\widehat\Sigma$ and $\Delta_\mu:=\mu-\hat\mu$. Since $\hat w$ minimizes $\widehat J$, $\widehat J(\hat w)\le\widehat J(w^{\star})$, and therefore
\[
r\le D(\hat w)-D(w^{\star})=\tfrac{\gamma}{2}\big(\hat w^{\Tr}\Delta_\Sigma\hat w-w^{\star\Tr}\Delta_\Sigma w^{\star}\big)-\Delta_\mu^{\Tr}(\hat w-w^{\star}).
\]
The quadratic difference factors as $\tfrac{\gamma}{2}(\hat w+w^{\star})^{\Tr}\Delta_\Sigma(\hat w-w^{\star})$, and $\|\hat w+w^{\star}\|_2\le2\bar r$, so
\[
r\le\big(\gamma\bar r\,\|\Delta_\Sigma\|_{\mathrm{op}}+\|\Delta_\mu\|_2\big)\,\|\hat w-w^{\star}\|_2=:\epsilon\,\|\hat w-w^{\star}\|_2.
\]

\smallskip\noindent\textit{Step 2. Strong convexity bounds the distance by the regret.}\quad $J$ is $\gamma\lambda$-strongly convex because $\Sigma\succeq\lambda I$, and $w^{\star}$ minimizes it over the convex set $\calW$; the first-order optimality condition gives, for every $w\in\calW$,
\[
J(w)-J(w^{\star})\ \ge\ \tfrac{\gamma\lambda}{2}\,\|w-w^{\star}\|_2^{2},
\qquad\text{hence}\qquad
\|\hat w-w^{\star}\|_2\le\sqrt{2r/(\gamma\lambda)}.
\]

\smallskip\noindent\textit{Step 3. Closing the inequality.}\quad Combining the two steps, $r\le\epsilon\sqrt{2r/(\gamma\lambda)}$; if $r>0$, dividing by $\sqrt r$ and squaring gives $r\le2\epsilon^{2}/(\gamma\lambda)$, which is~\eqref{eq:regret2}. The rate follows because $\sqrt N$-consistency of the moments makes $\epsilon=O_p(N^{-1/2})$ and hence $r=O_p(N^{-1})$: the perturbation enters the regret squared, which is the sense in which it is second order, while it enters the reported risk $w^{\Tr}\widehat\Sigma w$ linearly.
\end{proof}

\begin{proof}[Proof of Proposition~\ref{prop:stage1-formal}]
Let $E_N$ be the event $\{\widehat\calA^{0}_t=\calA_t\}\cap\{\max_H|\hat s_t(H)-s_t(H)|<\Delta_a/2\}$; by (i)--(ii) and the union bound, $\Prob(E_N)\to1$. On $E_N$, for every $H\notin\Ast$ and $H^{\star}\in\Ast$,
\[
\hat s_t(H)\ >\ s_t(H)-\tfrac{\Delta_a}{2}\ \ge\ \min s_t+\tfrac{\Delta_a}{2}\ =\ s_t(H^{\star})+\tfrac{\Delta_a}{2}\ >\ \hat s_t(H^{\star}),
\]
where the first and last inequalities use $|\hat s_t-s_t|<\Delta_a/2$ and the middle one the definition of $\Delta_a$. Hence no $H\notin\Ast$ minimizes $\hat s_t$ on $E_N$, so $\widehat\calA^{\star}_t\subseteq\Ast$. If in addition $\hat s_t$ is constant on $\Ast$, every element of $\Ast$ is an $\hat s_t$-minimizer and equality holds.
\end{proof}

\begin{proof}[Proof of Proposition~\ref{prop:stage2-formal}]

\smallskip\noindent\textit{Step 1 (value function).}\quad Under (i), $\rho^{\star}_t$ is finite, nondecreasing and convex on the feasible target interval $(-\infty,\max_{\calW}\mu]$ by the argument of Proposition~\ref{thm:markowitz-geometry} (convexity of $\rho$ and affinity of $\mu$ are all that was used), and a finite convex function is continuous on the interior of its domain \cite[Thm.~10.1]{Rockafellar1970}; by (iii) there is $\eta_0>0$ with $[m_--\eta_0,m_++\eta_0]$ in that interior, so $\rho^{\star}_t$ is uniformly continuous there with modulus $\omega$. 

\smallskip\noindent\textit{Step 2 (sandwich).}\quad Fix $m\in\calM$ and $N$ with $\eta_N<\eta_0$, and write $F(m)=\{w\in\calW:\mu(w)\ge m\}$, $\widehat F_N(m)=\{w\in\calW:\hat\mu_N(w)\ge m\}$. Since $|\hat\mu_N-\mu|\le\eta_N$ on $\calW$, $F(m+\eta_N)\subseteq\widehat F_N(m)\subseteq F(m-\eta_N)$. Minimizing $\hat\rho_N\in[\rho-\eta_N,\rho+\eta_N]$ over these nested sets,
\[
\rho^{\star}_t(m-\eta_N)-\eta_N\ \le\ \hat\rho^{\star}_N(m)\ \le\ \rho^{\star}_t(m+\eta_N)+\eta_N,
\]
and by monotonicity $\rho^{\star}_t(m-\eta_N)\le\rho^{\star}_t(m)\le\rho^{\star}_t(m+\eta_N)$. Hence $|\hat\rho^{\star}_N(m)-\rho^{\star}_t(m)|\le\eta_N+\omega(\eta_N)$ uniformly in $m\in\calM$, and the right-hand side tends to zero in probability, which is~\eqref{eq:frontier-consistency}. 

\smallskip\noindent\textit{Step 3 (minimizers).}\quad Let $\hat w_N\in\widehat\calE_N(m)$ and $\bar w$ be an accumulation point along a subsequence (which exists by compactness of $\calW$). Along it, $\mu(\bar w)=\lim\mu(\hat w_N)\ge\lim(\hat\mu_N(\hat w_N)-\eta_N)\ge m$ and $\rho(\bar w)=\lim\rho(\hat w_N)=\lim(\hat\rho_N(\hat w_N)\pm\eta_N)=\lim\hat\rho^{\star}_N(m)=\rho^{\star}_t(m)$ by Step~2, using continuity of $\mu,\rho$ and closedness of $\calW$. Thus $\bar w\in F(m)$ attains $\rho^{\star}_t(m)$, i.e.\ $\bar w\in\calE^{\star}_t(m)$. Conditioning on correct Stage-I selection is legitimate because its probability tends to one, and sample separation ensures that Stage-II performance cannot alter Stage-I eligibility.
\end{proof}

\begin{proof}[Proof of Proposition~\ref{thm:regret}]
On $E_N$, apply~\eqref{eq:uniform-dev} at $(\widehat H,\hat w)$ and at $(\Hs,w^{\star})$ and use the optimality of $(\widehat H,\hat w)$ for $\widehat J_N$:
\[
J(\hat w;\widehat H)\le\widehat J_N(\hat w;\widehat H)+\varepsilon_N\le\widehat J_N(w^{\star};\Hs)+\varepsilon_N\le J(w^{\star};\Hs)+2\varepsilon_N .\qedhere
\]
\end{proof}

\begin{proof}[Proof of Theorem~\ref{thm:oracle}]

\smallskip\noindent\textit{Step 1. A uniform deviation over the class.}\quad Set
\[x_r:=\sigma\sqrt{2\log(2M_r/\delta)/N},\qquad\text{so that}\qquad 2\exp\{-Nx_r^{2}/(2\sigma^{2})\}=\delta/M_r.\]
By~\eqref{eq:subgauss} and the union bound over the $M_r$ elements of $\calA_r$,
\[
\Prob\Big\{\max_{H\in\calA_r}\big|\widehat V_N(H)-V(H)\big|>x_r\Big\}\ \le\ M_r\cdot\frac{\delta}{M_r}\ =\ \delta.
\]

\smallskip\noindent\textit{Step 2. The oracle inequality on the complementary event.}\quad On the event of probability at least $1-\delta$ where the deviation is below $x_r$, let $H^{\circ}\in\argmin_{\calA_r}V$, which exists because $\calA_r$ is finite. Then
\[
V(\widehat H_r)\le\widehat V_N(\widehat H_r)+x_r\le\widehat V_N(H^{\circ})+x_r\le V(H^{\circ})+2x_r,
\]
the middle inequality being the optimality of $\widehat H_r$ for $\widehat V_N$; this is~\eqref{eq:oracle}.

\smallskip\noindent\textit{Step 3. Relative to the conditional-law oracle.}\quad Since $\min_{\calA_t(\mathrm i)}V=V^{\star}_{\mathrm i}$,
\[
V(\widehat H_{\mathrm i})-V^{\star}_{\mathrm c}=\big(V(\widehat H_{\mathrm i})-V^{\star}_{\mathrm i}\big)+\big(V^{\star}_{\mathrm i}-V^{\star}_{\mathrm c}\big)\le2x_{\mathrm i}+\Delta_{\mathrm{ci}}
\]
by Step~2 and the definition~\eqref{eq:price}, which is~\eqref{eq:oracle-c2}.

\smallskip\noindent\textit{Step 4. On the deviation hypothesis.}\quad Inequalities of the form~\eqref{eq:subgauss} hold for bounded or truncated losses by Hoeffding's inequality \cite{BLM2013,Hoeffding1963}, and under a finite-variance assumption alone for the value computed with a robust mean estimator in place of the sample mean \cite{Catoni2012}. Heavy-tailed portfolio losses such as those of experiment Experiment~\ref{exp:E3} are outside the hypotheses unless one of these devices is used, and Experiment~\ref{exp:E3} is reported as a comparison of designs rather than as an instance of the bound.
\end{proof}

\begin{proof}[Proof of Corollary~\ref{cor:regularization}]
Both bounds are upper bounds on the excess loss over the common oracle $V^{\star}_{\mathrm{c}}$; the conditional-law bound is $2x_{\mathrm{c}}$ and the interventional bound is $\Delta_{\mathrm{ci}}+2x_{\mathrm{i}}$. The inequality $\Delta_{\mathrm{ci}}+2x_{\mathrm{i}}<2x_{\mathrm{c}}$ is~\eqref{eq:regularization} after substituting $x_r$.
\end{proof}

\begin{proof}[Proof of Proposition~\ref{prop:hysteresis}]

\smallskip\noindent\textit{Part (a).}\quad Reselection means $\hat s_t(\widehat H_{t-1})>\min\hat s_t+\eta$; since $\hat s_t(\widehat H_{t-1})\le s_t(\widehat H_{t-1})+\delta_t$ and $\min\hat s_t\ge\min s_t-\delta_t$, we obtain $s_t(\widehat H_{t-1})-\min s_t>\eta-2\delta_t$.

\smallskip\noindent\textit{Part (b).}\quad If the incumbent is population-optimal, the left side of (a) is zero, so (a) would require $\eta<2\delta_t$.

\smallskip\noindent\textit{Part (c).}\quad For $\widehat H_t\in\argmin\hat s_t$ and $H^{\star}\in\argmin s_t$, $s_t(\widehat H_t)\le\hat s_t(\widehat H_t)+\delta_t\le\hat s_t(H^{\star})+\delta_t\le s_t(H^{\star})+2\delta_t$.
\end{proof}

\subsection{A continuous candidate class}
\begin{proposition}[A continuous candidate class]\label{prop:continuous}
Let $(R,Z)$ be jointly Gaussian with $Z\in\R^{k}$, $\Cov(Z)\succ0$ and $\Cov(R\mid Z)$ diagonal, so that $Z$ is a separator. For $\theta$ on the unit sphere $S^{k-1}$ let $H_\theta:=\sigma(\calG^{0}_t,\theta^{\Tr}Z)$ and let $\calW(H_\theta)\equiv\calW$ be a fixed nonempty compact set. Then the quotient of $\{H_\theta:\theta\in S^{k-1}\}$ by information equivalence is compact, $\theta\mapsto\calW(H_\theta)$ is upper hemicontinuous, the screening defect $\theta\mapsto\varepsilon_t(H_\theta)$ and the conditional value $\theta\mapsto V(H_\theta)$ are continuous, and both the envelope problem~\eqref{eq:apo} and the Stage-I problem have solutions.
\end{proposition}
\begin{proof}
Information equivalence identifies $\theta$ with $-\theta$ and no other pairs, since $\sigma(\theta^{\Tr}Z)=\sigma(\theta'^{\Tr}Z)$ for a nondegenerate Gaussian $Z$ if and only if $\theta'=\pm\theta$; the quotient is the projective space $S^{k-1}/\{\pm1\}$, which is compact. A constant correspondence is upper hemicontinuous. By Gaussian conditioning, $\Cov(R\mid H_\theta)=\Cov(R)-\Cov(R,\theta^{\Tr}Z)\Cov(\theta^{\Tr}Z)^{-1}\Cov(\theta^{\Tr}Z,R)$, a rational function of $\theta$ with denominator $\theta^{\Tr}\Cov(Z)\theta\ge\lambda_{\min}(\Cov(Z))>0$, hence continuous on the sphere, as is $\mu_{H_\theta}$; the maximal absolute off-diagonal correlation is a continuous function of a covariance with positive diagonal, so $\varepsilon_t$ is continuous, and $V(H_\theta)=\min_{w\in\calW}J(w;H_\theta)$ is continuous by Berge's maximum theorem \cite[Thm.~17.31]{AliprantisBorder2006} since $J$ is continuous in $(w,\theta)$. Theorem~\ref{thm:existence} and Theorem~\ref{thm:two-sets} apply.
\end{proof}

\section{Admissible tangent portfolio, capital allocation line and surface}\label{app:cal}
Let a riskless rate $r_f$ be available and let $\Hst$ be the fixed class with conditional mean $\mu_{\Hst}$ and covariance $\Sigma_{\Hst}\succ0$.

\label{def:tangent}The admissible tangent portfolio at $t$ is the solution of $\max\{(\mu_{\Hst}^{\Tr}w-r_f)/\sqrt{w^{\Tr}\Sigma_{\Hst}w}:\ w\in\calW(\Hst),\ \ind^{\Tr}w=1\}$, and the admissible capital allocation line is the half-line from $(0,r_f)$ through its risk--return point. Its slope $S_t(\Hst)$ is the admissible conditional Sharpe ratio; it is a random variable through the conditional moments, and $\bar S_t(H):=\E[S_t(H)^{2}]^{1/2}$ is the unconditional Sharpe ratio of the conditional tangent strategy, the object of \cite{HansenRichard1987,FersonSiegel2001}.

For $\calW(\Hst)=\{w:\ind^{\Tr}w=1\}$ the tangent portfolio is $w_T\propto\Sigma_{\Hst}^{-1}(\mu_{\Hst}-r_f\ind)$ and $S_t(\Hst)^{2}=(\mu_{\Hst}-r_f\ind)^{\Tr}\Sigma_{\Hst}^{-1}(\mu_{\Hst}-r_f\ind)$, as in the classical case; what changes is which $\Sigma_{\Hst}$ is allowed to enter. Two properties distinguish it from the tangent portfolio of an arbitrary conditioning set.

\begin{proposition}[Admissible Sharpe ratio]\label{prop:sharpe}
Let the feasible set be $\{w:\ind^{\Tr}w=1\}$ for every admissible representation.

\smallskip\noindent\textit{Part (a).}\quad For every $H\in\calA_t$, $\bar S_t(H)\le\sup_{G\in\calA_t}\bar S_t(G)$, and the difference for $H=\Hst$ is the admissibility price in Sharpe units; it is zero if and only if some Stage-I-optimal class attains the maximal unconditional Sharpe ratio over the admissible class.

\smallskip\noindent\textit{Part (b).}\quad If $H_1\preceq H_2$ are admissible and $H_2$ embeds the $H_1$ decision problem, then $\bar S_t(H_2)\ge\bar S_t(H_1)$: the unconditional Sharpe ratio of the conditional tangent strategy is nondecreasing along the information order inside the admissible class.
\end{proposition}
\noindent The proof is in Appendix~\ref{app:proofs}.

Part (b) is the admissible form of the statement that conditioning information cannot lower the maximal Sharpe ratio \cite{HansenRichard1987}, with the same qualification as Theorem~\ref{thm:voi}: it holds inside the admissible class and under embedding, and fails for look-ahead information because that information is not admissible.

Because $\mu_{\Hst}$ and $\Sigma_{\Hst}$ depend on the decision state, the admissible CAL is one line per state and their union is a ruled surface in (risk, state, return) space, the \emph{admissible capital allocation surface} $\mu=r_f+S_t(\Hst)(z)\,\sigma$, whose root-mean-square slope over the state distribution is $\bar S_t(\Hst)$. The surface is not an illustration only: its root-mean-square slope is the quantity on which Proposition~\ref{prop:sharpe} is stated, and the fan of lines is the state-dependence that an unconditional Sharpe ratio averages away. Figure~\ref{fig:geometry} draws the geometry at one state and Figure~\ref{fig:cal-surface} the line across states and the surface: the line of the admissible class is steeper at each state than that of a representation failing separation, and the classical line is a horizontal section of the surface.

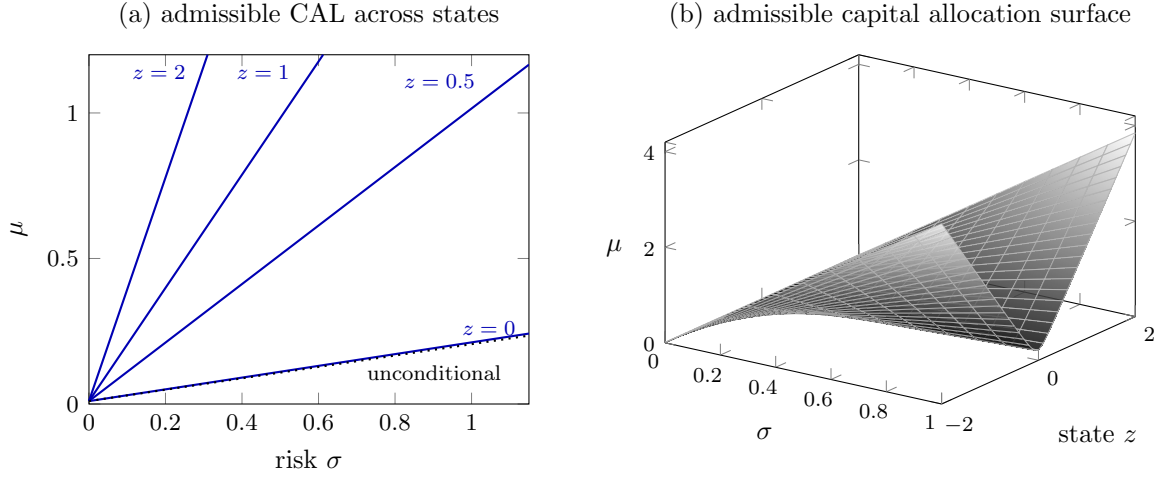
\begin{figure}[!htb]
\centering
\begin{tikzpicture}
\begin{groupplot}[group style={group size=2 by 1,horizontal sep=1.8cm},tick label style={font=\footnotesize},label style={font=\small}]
\nextgroupplot[width=7.4cm,height=6.2cm,xlabel={risk $\sigma$},ylabel={$\mu$},xmin=0,xmax=1.15,ymin=0,ymax=1.2,title={(a) admissible CAL across states},title style={font=\small},]
\addplot[thick,blue!70!black] coordinates {(0,0.010) (1.150,0.242)};
\node[font=\scriptsize,blue!70!black,anchor=south west] at (axis cs:0.950,0.201) {$z=0$};
\addplot[thick,blue!70!black,forget plot] coordinates {(0,0.010) (1.150,1.166)};
\node[font=\scriptsize,blue!70!black,anchor=south east] at (axis cs:1.035,1.050) {$z=0.5$};
\addplot[thick,blue!70!black,forget plot] coordinates {(0,0.010) (0.612,1.200)};
\node[font=\scriptsize,blue!70!black,anchor=south east] at (axis cs:0.551,1.081) {$z=1$};
\addplot[thick,blue!70!black,forget plot] coordinates {(0,0.010) (0.310,1.200)};
\node[font=\scriptsize,blue!70!black,anchor=south east] at (axis cs:0.279,1.081) {$z=2$};
\addplot[black,dotted,thick] coordinates {(0,0.010) (1.150,0.235)}; \node[font=\scriptsize,anchor=north west] at (axis cs:0.70,0.17) {unconditional};
\nextgroupplot[width=7.8cm,height=6.2cm,view={35}{28},xlabel={$\sigma$},ylabel={state $z$},zlabel={$\mu$},zmin=0,zmax=4.2,colormap/blackwhite,title={(b) admissible capital allocation surface},title style={font=\small},tick label style={font=\scriptsize},z label style={rotate=-90}]
\addplot3[surf,shader=faceted interp,mesh/rows=41,mesh/cols=11,faceted color=black!30,opacity=0.85] table {
0.00 -2.00 0.010
0.10 -2.00 0.387
0.20 -2.00 0.763
0.30 -2.00 1.140
0.40 -2.00 1.516
0.50 -2.00 1.893
0.60 -2.00 2.269
0.70 -2.00 2.646
0.80 -2.00 3.022
0.90 -2.00 3.399
1.00 -2.00 3.775

0.00 -1.90 0.010
0.10 -1.90 0.368
0.20 -1.90 0.725
0.30 -1.90 1.083
0.40 -1.90 1.440
0.50 -1.90 1.798
0.60 -1.90 2.155
0.70 -1.90 2.513
0.80 -1.90 2.870
0.90 -1.90 3.228
1.00 -1.90 3.586

0.00 -1.80 0.010
0.10 -1.80 0.349
0.20 -1.80 0.687
0.30 -1.80 1.026
0.40 -1.80 1.364
0.50 -1.80 1.703
0.60 -1.80 2.042
0.70 -1.80 2.380
0.80 -1.80 2.719
0.90 -1.80 3.057
1.00 -1.80 3.396

0.00 -1.70 0.010
0.10 -1.70 0.330
0.20 -1.70 0.649
0.30 -1.70 0.969
0.40 -1.70 1.289
0.50 -1.70 1.608
0.60 -1.70 1.928
0.70 -1.70 2.248
0.80 -1.70 2.567
0.90 -1.70 2.887
1.00 -1.70 3.207

0.00 -1.60 0.010
0.10 -1.60 0.311
0.20 -1.60 0.611
0.30 -1.60 0.912
0.40 -1.60 1.213
0.50 -1.60 1.514
0.60 -1.60 1.814
0.70 -1.60 2.115
0.80 -1.60 2.416
0.90 -1.60 2.716
1.00 -1.60 3.017

0.00 -1.50 0.010
0.10 -1.50 0.292
0.20 -1.50 0.574
0.30 -1.50 0.855
0.40 -1.50 1.137
0.50 -1.50 1.419
0.60 -1.50 1.701
0.70 -1.50 1.982
0.80 -1.50 2.264
0.90 -1.50 2.546
1.00 -1.50 2.828

0.00 -1.40 0.010
0.10 -1.40 0.273
0.20 -1.40 0.536
0.30 -1.40 0.799
0.40 -1.40 1.061
0.50 -1.40 1.324
0.60 -1.40 1.587
0.70 -1.40 1.850
0.80 -1.40 2.113
0.90 -1.40 2.376
1.00 -1.40 2.639

0.00 -1.30 0.010
0.10 -1.30 0.254
0.20 -1.30 0.498
0.30 -1.30 0.742
0.40 -1.30 0.986
0.50 -1.30 1.230
0.60 -1.30 1.474
0.70 -1.30 1.717
0.80 -1.30 1.961
0.90 -1.30 2.205
1.00 -1.30 2.449

0.00 -1.20 0.010
0.10 -1.20 0.235
0.20 -1.20 0.460
0.30 -1.20 0.685
0.40 -1.20 0.910
0.50 -1.20 1.135
0.60 -1.20 1.360
0.70 -1.20 1.585
0.80 -1.20 1.810
0.90 -1.20 2.035
1.00 -1.20 2.260

0.00 -1.10 0.010
0.10 -1.10 0.216
0.20 -1.10 0.422
0.30 -1.10 0.628
0.40 -1.10 0.834
0.50 -1.10 1.041
0.60 -1.10 1.247
0.70 -1.10 1.453
0.80 -1.10 1.659
0.90 -1.10 1.865
1.00 -1.10 2.071

0.00 -1.00 0.010
0.10 -1.00 0.197
0.20 -1.00 0.384
0.30 -1.00 0.572
0.40 -1.00 0.759
0.50 -1.00 0.946
0.60 -1.00 1.133
0.70 -1.00 1.321
0.80 -1.00 1.508
0.90 -1.00 1.695
1.00 -1.00 1.882

0.00 -0.90 0.010
0.10 -0.90 0.178
0.20 -0.90 0.347
0.30 -0.90 0.515
0.40 -0.90 0.683
0.50 -0.90 0.852
0.60 -0.90 1.020
0.70 -0.90 1.189
0.80 -0.90 1.357
0.90 -0.90 1.525
1.00 -0.90 1.694

0.00 -0.80 0.010
0.10 -0.80 0.160
0.20 -0.80 0.309
0.30 -0.80 0.459
0.40 -0.80 0.608
0.50 -0.80 0.758
0.60 -0.80 0.907
0.70 -0.80 1.057
0.80 -0.80 1.206
0.90 -0.80 1.356
1.00 -0.80 1.505

0.00 -0.70 0.010
0.10 -0.70 0.141
0.20 -0.70 0.271
0.30 -0.70 0.402
0.40 -0.70 0.533
0.50 -0.70 0.664
0.60 -0.70 0.794
0.70 -0.70 0.925
0.80 -0.70 1.056
0.90 -0.70 1.187
1.00 -0.70 1.317

0.00 -0.60 0.010
0.10 -0.60 0.122
0.20 -0.60 0.234
0.30 -0.60 0.346
0.40 -0.60 0.458
0.50 -0.60 0.570
0.60 -0.60 0.682
0.70 -0.60 0.794
0.80 -0.60 0.906
0.90 -0.60 1.018
1.00 -0.60 1.130

0.00 -0.50 0.010
0.10 -0.50 0.103
0.20 -0.50 0.197
0.30 -0.50 0.290
0.40 -0.50 0.384
0.50 -0.50 0.477
0.60 -0.50 0.570
0.70 -0.50 0.664
0.80 -0.50 0.757
0.90 -0.50 0.851
1.00 -0.50 0.944

0.00 -0.40 0.010
0.10 -0.40 0.085
0.20 -0.40 0.160
0.30 -0.40 0.235
0.40 -0.40 0.310
0.50 -0.40 0.385
0.60 -0.40 0.460
0.70 -0.40 0.535
0.80 -0.40 0.610
0.90 -0.40 0.685
1.00 -0.40 0.760

0.00 -0.30 0.010
0.10 -0.30 0.067
0.20 -0.30 0.124
0.30 -0.30 0.181
0.40 -0.30 0.237
0.50 -0.30 0.294
0.60 -0.30 0.351
0.70 -0.30 0.408
0.80 -0.30 0.465
0.90 -0.30 0.522
1.00 -0.30 0.579

0.00 -0.20 0.010
0.10 -0.20 0.050
0.20 -0.20 0.089
0.30 -0.20 0.129
0.40 -0.20 0.169
0.50 -0.20 0.208
0.60 -0.20 0.248
0.70 -0.20 0.287
0.80 -0.20 0.327
0.90 -0.20 0.367
1.00 -0.20 0.406

0.00 -0.10 0.010
0.10 -0.10 0.035
0.20 -0.10 0.060
0.30 -0.10 0.085
0.40 -0.10 0.110
0.50 -0.10 0.135
0.60 -0.10 0.160
0.70 -0.10 0.185
0.80 -0.10 0.211
0.90 -0.10 0.236
1.00 -0.10 0.261

0.00 0.00 0.010
0.10 0.00 0.030
0.20 0.00 0.050
0.30 0.00 0.070
0.40 0.00 0.091
0.50 0.00 0.111
0.60 0.00 0.131
0.70 0.00 0.151
0.80 0.00 0.171
0.90 0.00 0.191
1.00 0.00 0.212

0.00 0.10 0.010
0.10 0.10 0.040
0.20 0.10 0.070
0.30 0.10 0.100
0.40 0.10 0.130
0.50 0.10 0.160
0.60 0.10 0.190
0.70 0.10 0.221
0.80 0.10 0.251
0.90 0.10 0.281
1.00 0.10 0.311

0.00 0.20 0.010
0.10 0.20 0.056
0.20 0.20 0.102
0.30 0.20 0.148
0.40 0.20 0.194
0.50 0.20 0.240
0.60 0.20 0.287
0.70 0.20 0.333
0.80 0.20 0.379
0.90 0.20 0.425
1.00 0.20 0.471

0.00 0.30 0.010
0.10 0.30 0.074
0.20 0.30 0.137
0.30 0.30 0.201
0.40 0.30 0.265
0.50 0.30 0.329
0.60 0.30 0.392
0.70 0.30 0.456
0.80 0.30 0.520
0.90 0.30 0.584
1.00 0.30 0.647

0.00 0.40 0.010
0.10 0.40 0.092
0.20 0.40 0.174
0.30 0.40 0.256
0.40 0.40 0.338
0.50 0.40 0.420
0.60 0.40 0.502
0.70 0.40 0.584
0.80 0.40 0.666
0.90 0.40 0.748
1.00 0.40 0.830

0.00 0.50 0.010
0.10 0.50 0.111
0.20 0.50 0.211
0.30 0.50 0.312
0.40 0.50 0.412
0.50 0.50 0.513
0.60 0.50 0.613
0.70 0.50 0.714
0.80 0.50 0.814
0.90 0.50 0.915
1.00 0.50 1.015

0.00 0.60 0.010
0.10 0.60 0.129
0.20 0.60 0.248
0.30 0.60 0.368
0.40 0.60 0.487
0.50 0.60 0.606
0.60 0.60 0.725
0.70 0.60 0.844
0.80 0.60 0.964
0.90 0.60 1.083
1.00 0.60 1.202

0.00 0.70 0.010
0.10 0.70 0.148
0.20 0.70 0.286
0.30 0.70 0.424
0.40 0.70 0.562
0.50 0.70 0.700
0.60 0.70 0.838
0.70 0.70 0.976
0.80 0.70 1.114
0.90 0.70 1.251
1.00 0.70 1.389

0.00 0.80 0.010
0.10 0.80 0.167
0.20 0.80 0.324
0.30 0.80 0.480
0.40 0.80 0.637
0.50 0.80 0.794
0.60 0.80 0.951
0.70 0.80 1.107
0.80 0.80 1.264
0.90 0.80 1.421
1.00 0.80 1.578

0.00 0.90 0.010
0.10 0.90 0.186
0.20 0.90 0.361
0.30 0.90 0.537
0.40 0.90 0.712
0.50 0.90 0.888
0.60 0.90 1.064
0.70 0.90 1.239
0.80 0.90 1.415
0.90 0.90 1.590
1.00 0.90 1.766

0.00 1.00 0.010
0.10 1.00 0.204
0.20 1.00 0.399
0.30 1.00 0.593
0.40 1.00 0.788
0.50 1.00 0.982
0.60 1.00 1.177
0.70 1.00 1.371
0.80 1.00 1.566
0.90 1.00 1.760
1.00 1.00 1.955

0.00 1.10 0.010
0.10 1.10 0.223
0.20 1.10 0.437
0.30 1.10 0.650
0.40 1.10 0.863
0.50 1.10 1.077
0.60 1.10 1.290
0.70 1.10 1.504
0.80 1.10 1.717
0.90 1.10 1.930
1.00 1.10 2.144

0.00 1.20 0.010
0.10 1.20 0.242
0.20 1.20 0.475
0.30 1.20 0.707
0.40 1.20 0.939
0.50 1.20 1.171
0.60 1.20 1.404
0.70 1.20 1.636
0.80 1.20 1.868
0.90 1.20 2.100
1.00 1.20 2.333

0.00 1.30 0.010
0.10 1.30 0.261
0.20 1.30 0.512
0.30 1.30 0.764
0.40 1.30 1.015
0.50 1.30 1.266
0.60 1.30 1.517
0.70 1.30 1.768
0.80 1.30 2.019
0.90 1.30 2.271
1.00 1.30 2.522

0.00 1.40 0.010
0.10 1.40 0.280
0.20 1.40 0.550
0.30 1.40 0.820
0.40 1.40 1.090
0.50 1.40 1.361
0.60 1.40 1.631
0.70 1.40 1.901
0.80 1.40 2.171
0.90 1.40 2.441
1.00 1.40 2.711

0.00 1.50 0.010
0.10 1.50 0.299
0.20 1.50 0.588
0.30 1.50 0.877
0.40 1.50 1.166
0.50 1.50 1.455
0.60 1.50 1.744
0.70 1.50 2.033
0.80 1.50 2.322
0.90 1.50 2.611
1.00 1.50 2.900

0.00 1.60 0.010
0.10 1.60 0.318
0.20 1.60 0.626
0.30 1.60 0.934
0.40 1.60 1.242
0.50 1.60 1.550
0.60 1.60 1.858
0.70 1.60 2.166
0.80 1.60 2.474
0.90 1.60 2.782
1.00 1.60 3.090

0.00 1.70 0.010
0.10 1.70 0.337
0.20 1.70 0.664
0.30 1.70 0.991
0.40 1.70 1.318
0.50 1.70 1.645
0.60 1.70 1.972
0.70 1.70 2.299
0.80 1.70 2.625
0.90 1.70 2.952
1.00 1.70 3.279

0.00 1.80 0.010
0.10 1.80 0.356
0.20 1.80 0.702
0.30 1.80 1.048
0.40 1.80 1.394
0.50 1.80 1.739
0.60 1.80 2.085
0.70 1.80 2.431
0.80 1.80 2.777
0.90 1.80 3.123
1.00 1.80 3.469

0.00 1.90 0.010
0.10 1.90 0.375
0.20 1.90 0.740
0.30 1.90 1.104
0.40 1.90 1.469
0.50 1.90 1.834
0.60 1.90 2.199
0.70 1.90 2.564
0.80 1.90 2.929
0.90 1.90 3.293
1.00 1.90 3.658

0.00 2.00 0.010
0.10 2.00 0.394
0.20 2.00 0.778
0.30 2.00 1.161
0.40 2.00 1.545
0.50 2.00 1.929
0.60 2.00 2.313
0.70 2.00 2.696
0.80 2.00 3.080
0.90 2.00 3.464
1.00 2.00 3.848
};
\end{groupplot}
\end{tikzpicture}
\caption{The admissible capital allocation line across decision states and the surface it sweeps, reference system, class $\{Z_1,Z_2\}$, $r_f=0.01$. (a) The admissible CAL at four states; its slope is the conditional Sharpe ratio $S_t(\Hst)(z)$ and grows with the information carried by the state, while the unconditional line does not move. (b) The union of the lines over states is the admissible capital allocation surface $\mu=r_f+S_t(\Hst)(z)\,\sigma$; its root-mean-square slope over the state distribution is the unconditional Sharpe ratio $\bar S_t(\Hst)$ of the conditional tangent strategy, $1.90$ against $1.29$ for the proxy and $0.20$ unconditionally, and the classical capital allocation line is one horizontal section of it.}
\label{fig:cal-surface}
\end{figure}

Markowitz portfolio theory is the degenerate case in which a single admissible class $\{H_0\}$ with $\rho_{H_0}(w)=w^{\Tr}\Sigma w$ makes Stage~I vacuous and Stage~II the classical program, and the trivial information object gives the unconditional problem.

\section{Experiment specification}\label{app:config}

Every design that was executed in the course of this work is reported, either in the main text or in Section~\ref{sec:finite}: the controlled Experiments~\ref{exp:E1} to~\ref{exp:E4}; the robustness runs over seeds, tolerance levels ($0.05$, $0.10$, $0.15$), tolerance calibrations (per-pair Fisher, \v{S}id\'ak, multiplier bootstrap) and the reselection weight ($0.01$, $0.02$, $0.05$); the lag test at one, two and five sessions; the six asset universes and the aggregation and candidate-search designs; the evaluation protocol with six estimators, unconstrained and long-only; the calibration decomposition at three sets of weights; and the regime-invariance instance. No executed design was omitted because of its outcome; the two that ran against the framework are in Section~\ref{sec:market}.

Table~\ref{tab:spec} gives, for every design reported, the panel, the window structure, the selection parameters, the estimators and the tests; the data transformations are as follows. Asset prices enter as log differences. Candidate drivers are classified by type: yields, swap rates and risk reversals (22 series) enter as first differences; prices, indices, exchange rates, commodities and total-return indices (105 series) as log differences; sentinel values ($-1$) and non-finite entries are set to missing, single-day moves beyond twenty standard deviations are treated as missing, sessions with fewer than $90\%$ of drivers observed are dropped and isolated gaps are zero innovations. Industry and factor portfolios from Kenneth French's library enter as returns. Bond prices from the 1{,}336-issue file are dated one session later than the market data and are shifted back one session before alignment. The screening tolerance is $0.10$ plus a \v{S}id\'ak bound for the maximum over the $n(n-1)/2$ pairs at the certification-block length, with the multiplier-bootstrap calibration reported where stated. Forward selection is greedy up to the stated dimension. Long-only minimum-variance portfolios are computed by accelerated projected gradient with exact simplex projection to a relative tolerance of $10^{-10}$; unconstrained ones in closed form. The circular block bootstrap resamples the sequence of per-window paired differences in blocks of four windows with $2{,}000$ draws. Every quantity quoted in the text is produced by the scripts named in the table and checked against their outputs by the verification script in the repository.

\begin{table}[!htb]\centering\scriptsize\setlength{\tabcolsep}{3pt}\renewcommand{\arraystretch}{1.15}
\caption{Specification of every design reported.}
\label{tab:spec}
\begin{tabular}{@{}p{2.9cm}p{3.5cm}p{3.1cm}p{2.8cm}p{2.5cm}@{}}
\toprule
design & panel & windows and blocks & selection & estimators / tests \\
\midrule
Experiments \ref{exp:E1}--\ref{exp:E4} (controlled) & Gaussian structural system, 8 assets, 2 states, 41 candidates from 6 variables & \ref{exp:E2}: $N\in\{250,\dots,8000\}$, 50/25/25 split, 300 replications; \ref{exp:E4}: 200 replications & tolerance $0.10$ plus Fisher bound; hysteresis margin $0.01$, $0.02$, $0.05$ & exact recovery; selection and frontier error rates; intervention $\sigma_{Z_1}\to1.5\sigma_{Z_1}$; reselection \\
primary screening & 147 S\&P~500 names, 180 candidates (127 series + 48 industries + 5 factors), 2010--2024 & 756 or 945 sessions stepped by 63; certification 378, estimation 189 or 378, evaluation 189 & $K=12$ & defect, consequence statistic, bias statistics \\
head-to-head & as above & as above, both estimation lengths & $K=12$ & sample, LW linear, nonlinear (analytical), DCC-NL, AFM-DCC-NL, $k$ PCs $\pm$ threshold, least-defect $\pm$ threshold, robust; block bootstrap vs best \\
macro-only primary & 147 names, 127 series & 756 stepped by 63, 50/25/25 & $K=8$ & screening, lag 1, halves, consequence, bias, loss-based selection \\
regime price & 147 names, 127 series & every second window & $K=8$; invariance at 3 pooled s.e. & calendar regimes 2010--14/15--19/20--24; policy regimes by FOMC dates \\
in-loop invariance & 147 names, 127 series & quarterly, 80 rebalances & thirds of the certification block as environments, 2.5 s.e. & defect, admissible fraction, long-only strategy \\
twelve-equity panel & 12 US equities, 127 series, 2001--2023 & 756 stepped by 63 (81 windows), 50/25/25 & $K=8$ & screening, calibration, protocol (Table~\ref{tab:ec}), strategy (Table~\ref{tab:strategy}), reselection \\
universe grid & 48 industries; 25 size/value; 25 Europe and Developed size/momentum & 756 stepped by 63 & $K=8$, pair-calibrated tolerance & screening only \\
aggregation & 147 names in baskets of $k\in\{1,2,4,8,12\}$, 6 draws & 756 stepped by 189 & $K=8$ & screening \\
candidate search & 12 equities; 25 size/value & random candidate subsets of 4 to 127, 8 draws & $K=8$ & attainable defect \\
cross-asset & 147 equities and 1{,}336 Spanish bonds, 2014--2018; 12 equities, 12 bonds, 6+6, 50+50 & 504 stepped by 63 or 126 & $K=8$, rank-correlation defect & screening, lag 1, consequence, bias \\
hedge funds & 17 strategy indices, monthly, 271 months & 120 months stepped by 6 & $K=4$ & screening only \\
\bottomrule
\end{tabular}
\end{table}

\end{document}